\documentclass [fleqn,10pt] {article}
\usepackage[dvipdfm]{graphicx}
\usepackage{amsfonts}
\usepackage{amsmath}
\usepackage{amssymb}
\usepackage{bm}
\usepackage[alphabetic]{amsrefs}
\usepackage{fullpage,theorem}
\usepackage{tensor}
\usepackage{graphics}
\usepackage{pdflscape}
\usepackage{tipa}
\usepackage[greek,english]{babel}
\usepackage{multirow}

\newtheorem{thm}{Theorem}[section]
\newtheorem{cor}[thm]{Corollary}
\newtheorem{lem}[thm]{Lemma}
\newtheorem{prop}[thm]{Proposition}

\newtheorem{conjec}[thm]{Conjecture}

\theorembodyfont{\rmfamily}
\newtheorem{defn}[thm]{Definition}
\newtheorem{exa}[thm]{Example}

\newtheorem{rem}[thm]{Remark}

\numberwithin{equation}{section}
\newenvironment{proof}{\noindent \emph{Proof.}}{\hspace{\stretch{1}}$\Box$}

\newcommand{\derv}[2] {\frac{\mathrm{d} #1}{\mathrm{d} #2}}
\newcommand{\parderv}[2] {\frac{\partial#1}{\partial#2}}

\newcommand{\tldo} {{\tilde{1}}}

\newcommand{\yogh} {\text{\textit{\textyogh}}}
\newcommand{\wynn} {\text{\textit{\textwynn}}}
\newcommand{\mae} {\text{\textit{\ae}}}

\newcommand{\mqoppa} {\text{\foreignlanguage{greek}{\qoppa}}}
\newcommand{\mstigma} {\text{\foreignlanguage{greek}{\stigma}}}
\newcommand{\msampi} {\text{\foreignlanguage{greek}{\sampi}}}
\newcommand{\mdigamma} {\text{\foreignlanguage{greek}{\ddigamma}}}

\newcommand{\mcA} {\mathcal{A}}
\newcommand{\mcB} {\mathcal{B}}

\newcommand{\mcK} {\mathcal{K}}
\newcommand{\mcL} {\mathcal{L}}
\newcommand{\mcM} {\mathcal{M}}
\newcommand{\mcN} {\mathcal{N}}

\newcommand{\dd} {\mathrm{d}}

\newcommand{\ii} {\mathrm{i}}

\newcommand{\sspan} {\mathop{\mathrm{span}}}

\newcommand{\ind} {\indices}

\newcommand{\lp} [1] {{\left( #1 \right. }}
\newcommand{\rp} [1] {{\left. #1 \right) }}
\newcommand{\lb} [1] {{\left[ #1 \right. }}
\newcommand{\rb} [1] {{\left. #1 \right] }}

\newcommand{\Rho} {\mathrm{P}}

\newcommand{\U} {\mathrm{U}}

\newcommand{\Tgt} {\mathrm{T}}

\newcommand{\Riem} {\mathrm{R}}

\newcommand{\uu} {\mathfrak{u}}

\newcommand{\R} {\mathbb{R}}
\newcommand{\C} {\mathbb{C}}

\begin{document}
\title{Optical structures, algebraically special spacetimes, and the Goldberg-Sachs theorem in five dimensions}
\author{Arman Taghavi-Chabert\\
{\small Masaryk University, Faculty of Science, Department of Mathematics and Statistics,}\\
 {\small Kotl\'{a}\v{r}sk\'{a} 2, 611 37 Brno, The
Czech Republic } }
\date{}

\maketitle

\begin{abstract}
Optical (or Robinson) structures are one generalisation of four-dimensional shearfree congruences of null geodesics to higher dimensions. They are Lorentzian analogues of complex and CR structures. In this context, we extend the Goldberg-Sachs theorem to five dimensions. To be precise, we find a new algebraic condition on the Weyl tensor, which generalises the Petrov type II condition, in the sense that it ensures the existence of such congruences on a five-dimensional spacetime, vacuum or under weaker assumptions on the Ricci tensor. This results in a significant simplification of the field equations. We discuss possible degenerate cases, including a five-dimensional generalisation of the Petrov type D condition. We also show that the vacuum black ring solution is endowed with optical structures, yet fails to be algebraically special with respect to them. We finally explain the generalisation of these ideas to higher dimensions, which has been checked in six and seven dimensions.
\end{abstract}

\section{Introduction and motivation}
Shearfree congruences of null geodesics or shearfree rays (SFR) have occupied a central position in classical general relativity, featuring in many important solutions to Einstein's equations, such as the Kerr black hole and pp-waves. For vacuum spacetimes, they also have the added property that their generator $k^a$, say, is a repeated principal null direction (PND) of the Weyl tensor, or equivalently, that the Weyl tensor is algebraically special, i.e. of Petrov type II or more degenerate:
\begin{align} \label{eq-PND}
 k^a k^c C_{a b c \lb{d}} k_{\rb{e}} & = 0 \, .
\end{align}
The converse is also true: the Goldberg-Sachs theorem (1962) states \cite{Goldberg2009} that a vacuum spacetime admits a shearfree congruence of null geodesics if and only if its Weyl tensor is algebraically special.

SFR can also be equivalently described from the perspective of complex geometry. Naively, one should think of them as arising from the real intersections of two-dimensional complex leaves of a foliation of `complexified' spacetime, which are totally null or degenerate with respect to the complexified metric. At each point along a SFR, the (two-dimensional) space orthogonal to its generator is then naturally equipped with a complex structure. More generally, if $(\mcM, \bm{g})$ is a $(2m+\epsilon)$-dimensional Lorentzian manifold, where $\epsilon = 0,1$, an \emph{optical (or Robinson) structure $(\mcN,\mcK)$} on $(\mcM, \bm{g})$ is an integrable $m$-dimensional complex distribution $\mcN$ of the complexified tangent bundle, totally null with respect to the complexified metric, and which intersects its complex conjugate in the complexification of a real null line bundle $\C \otimes \mcK := \mcN \cap \overline{\mcN}$. Sections of $\mcK$ are tangent to null geodesics, and the fibers of the screenspace $\mcK^\perp/\mcK$ are naturally equipped with a complex structure when $\epsilon=0$, and with a CR structure when $\epsilon=1$. In four dimensions, these null geodesics are shearfree, but not so in higher dimensions \cite{Trautman2002a}. They nevertheless preserve the complex or CR structure on the screenspace.

This aspect of SFR was instrumental in the discovery of the Kerr black hole \cite{Kerr1963}, thanks to what is now known as the Kerr theorem \cite{Cox1976}. Later, these ideas would play a central r\^{o}le in the genesis and development of Roger Penrose's twistor theory \cites{Penrose1967,Penrose1986}. In general, the philosophy is to think of an optical structure as a \emph{real} manifestation of a null foliation of complexified spacetime. For different choices of metric signature, the same null structure gives rise to different real geometries. For instance, a metric-compatible complex structure on a Riemannian manifold is determined by a null distribution. Nurowski and Trautman's description \cites{Nurowski2002,Trautman2002} of an optical structure as a `Lorentzian analog of a Hermite structure' is particularly fitting in that sense.
Further, the fact that the Goldberg-Sachs theorem was later generalised \cites{Kundt1962,Robinson1963,Pleba'nski1975,Przanowski1983,Penrose1986,Apostolov1997,Apostolov1998,Gover2010} to any real or complex (pseudo-) Riemannian manifolds -- and under assumptions weaker than Ricci-flatness -- testifies of the centrality of complex methods in the study of SFR in four dimensions.

In recent decades, new theories of physics have made the study of solutions to Einstein's field equations in higher dimensions particularly relevant. The mathematics and physics communities have been engaged not only in generalising known four-dimensional  solutions such as pp-waves \cite{Coley2003}, the Kerr black hole \cites{Gibbons2005,Chen2006} or the Kerr-Schild ansatz \cites{Ortaggio2009,Ortaggio2009a} to higher dimensions, but also in investigating the broader geometrical properties of higher-dimensional spacetimes \cites{Coley2004a,Coley2004}. To this effect, Coley, Milson, Pravd\'{a}, and Pravdov\'{a} have proposed a spacetime classification purporting to generalise the Petrov classification \cites{Coley2004,Milson2005,Coley2006,Coley2008}. In this setting, various attempts have been made to generalise classical results of general relativity to higher dimensions, with some measure of success \cites{Pravda2004,Pravda2005,Pravda2006,Pravda2007,Ortaggio2007,Pravdova2008}. While the geodesy property of null congruences was successfully \cite{Durkee2009} related to the Coley-Milson-Pravd\'{a}-Pravdov\'{a} (CMPP) classification of the Weyl tensor in higher dimensions, the shearfree condition does not appear to hold any particularly privileged place in more than four dimensions. In fact, it is established in \cite{Frolov2003} that the five-dimensional Myers-Perry black hole admits a pair of PNDs in the sense of equation \eqref{eq-PND}, but these are not shearfree, unlike its four-dimensional counterpart. In particular, a higher-dimensional Goldberg-Sachs theorem cannot be formulated along these lines.

In contrast, there is strong evidence that optical structures still retain a central position in higher dimensions. An early step towards this is provided by the higher-dimenensional generalisation of the Kerr theorem by Hughston and Mason in \cite{Hughston1988}. More recently, Mason and the present author \cite{Mason2010} generalised a result due to Walker and Penrose \cite{Walker1970} and proposed a higher-dimensional version of the Petrov type D condition distinct from that of the CMPP classification. To be precise,
a conformal Killing-Yano (CKY) $2$-form on a $(2m+\epsilon)$-dimensional (pseudo-)Riemannian manifold  $(\mcM, \bm{g})$, in normal form, with generically distinct eigenvalues, and whose exterior derivative satisfies a certain condition, gives rises to $2^m$ integrable canonical null distributions. Further, the integrability condition of the CKY $2$-form implies that the Weyl tensor must be degenerate on each of these null distributions and their orthogonal complements. The prime example of a solution characterised by such a CKY $2$-form is the Kerr-NUT-AdS metric \cite{Chen2006}. All these results seem to point towards a higher-dimensional generalisation of the Goldberg-Sachs theorem in the context of null distributions and optical structures.

This motivates the following two questions, which this paper aims to answer:
\begin{itemize}
 \item Given a five-dimensional spacetime $(\mcM, \bm{g})$, does there exist an algebraic condition on the Weyl and Ricci tensor which guarantees the integrability of an almost optical structure $(\mcN, \mcK)$?
 \item Conversely, given a five-dimensional spacetime, vacuum or weaker, does the existence of an integrable almost optical structure imply the same algebraic condition on the Weyl tensor?
\end{itemize}
We shall prove the following result. Suppose that the Weyl tensor and the Cotton-York tensor as defined in Appendix \ref{App-Notation} satisfy
\begin{align*}
 \bm{C} ( \bm{X}, \bm{Y}, \bm{Z}, \cdot ) & = 0 \, , & \bm{A} ( \bm{Z}, \bm{X}, \bm{Y}) & = 0 \, ,
\end{align*}
respectively, for all $\bm{X}, \bm{Y} \in \Gamma ( \mcN^\perp ) $, $\bm{Z} \in \Gamma (\mcN )$. Then, assuming the Weyl tensor does not degenerate any further, the almost optical structure $(\mcN ,\mcK)$ is integrable.

On the other hand, we shall show that the five-dimensional vacuum black ring \cites{Emparan2002,Emparan2006,Emparan2007} provides a counterexample to the converse, i.e. it admits an (integrable) optical structure but its Weyl tensor fails to be algebraically special relative to it.

The structure of the paper is as follows. In the first section, we recall the basic definitions and properties of null and optical geometries in arbitrary spacetimes. While some of this material can already be found in \cite{Nurowski2002} as far as the even-dimensional case goes, we have extended their approach to odd-dimensional spacetime. The integrability condition for an almost optical structure is then stated, and we use the Sachs equations in five dimensions to illustrate how the curvature obstructs the propagation of an optical structure along a null geodesic.

In the second section, we give a definition of algebraic special spacetimes relative to an almost optical structure $(\mcN, \mcK)$, as a generalisation of the Petrov type II condition. In five dimensions, this is shown to imply that the eigenvalues of the Weyl curvature operator take on very simple expressions. Comparisons with the CMPP classification of the Weyl tensor and the De Smet classification of the conformal Weyl spinor are drawn. We then justify our choice of definition, by proving that provided the Cotton-York tensor is degenerate with respect to the optical structure, and the Weyl tensor satisfies some genericity assumption, then the algebraic speciality implies the integrability of $(\mcN,\mcK)$. We then show how the Bianchi identity greatly simplifies for such spacetimes. This is followed by a discussion of more degenerate algebraic conditions on the Weyl tensor, including the Petrov type D condition.

We end the section by a study of the five-dimensional vacuum black ring, which, we show, is endowed with optical structures, but fails to be algebraically special relative to them. In some regions, the null structures do not intersect in a real line bundle, but define a pair of CR structures. In the limit where the black ring becomes the Myers-Perry black hole with one rotation coefficient, these CR structures persist, and in fact, co-exist with a pair of optical structures.

Finally, we conjecture that most of the results of the present paper can be generalised to higher dimensions, to arbitrary signatures, and to complex manifolds. The conjecture has already been verified in six and seven dimensions using \emph{Mathematica}, and the results will be presented in a future publication \cite{Taghavi-Chaberta}.

We have relegated the basic notational setup, including an analysis of the Weyl curvature operator for algebraically special spacetimes, together with the component forms of the Ricci and Bianchi identities to three appendices. Some of the equations therein were simplified by means of the symbolic computer algebra system \emph{Cadabra} \cites{Peeters2007,Peeters2007a}.

\paragraph{Acknowledgments} I would like to thank Lionel Mason for stimulating my interest in the subject, for the many fruitful discussions I have had the pleasure to have with him, and for his comments on a draft version of this paper. Thanks also to Alan Coley and Sigbj\o{}rn Hervik for their feedback on the first version of this preprint. This work was carried out at the Mathematical Institute, Oxford, UK, and revised in Brno, the Czech Republic, where the author is supported by an Eduard \v{C}ech Center postdoctoral fellowship.

\section{Preliminaries}
\subsection{Optical structures}
Throughout, $(\mcM, \bm{g})$ will denote a $(2m+\epsilon)$-dimensional (pseudo-)Riemannian manifold, where $\epsilon \in \{ 0 , 1 \}$. Given a real (complex) subbundle $E$ of the (complexified) tangent bundle $\Tgt \mcM$, the orthogonal complement of $E$  with respect to the (complexified) metric will be denoted by $E^\perp$, and the space of sections of $E$ by $\Gamma (E)$. Denote by $\nabla$ the Levi-Civita connection.

The definition of an optical structure in even dimensions is given in \cite{Nurowski2002}. Here, we introduce the notion for odd-dimensional spacetimes, which may be described as a Lorentzian analogue of a CR structure.
\begin{defn}\label{defn-NStr}
 An \emph{almost null structure} $\mcN$ on $\mcM$ is a complex subbundle of the complexified tangent bundle of $\mcM$, which is totally null with respect to the complexified metric, and has maximal rank, i.e. $\mcN \subset \mcN^\perp$ and $\mcN$ has rank $m$. We say that $\mcN$ is a \emph{null structure}, when $\mcN$ is integrable in the sense of Frobenius, i.e. $[ \Gamma (\mcN) , \Gamma (\mcN) ] \subset \Gamma (\mcN) $.
\end{defn}
The complexified tangent bundle is equipped with a natural reality structure, induced from an involutory complex-conjugation operation, denoted $\bar{}$, which preserves the underlying real metric. The signature of the real metric imposes restrictions on the possible reality conditions that an almost null structure can satisfy. This motivates the following definition and lemma.
\begin{defn}\label{defn-realind}
The \emph{real index} of an almost null structure $\mcN$ is the dimension of the intersection $\mcN \cap \overline{\mcN}$.
\end{defn}
\begin{lem}[\cite{Kopczy'nski1992}]\label{lem-realinddim}
 Suppose that $(\mcM, \bm{g})$ is Lorentzian, and let $\mcN$ be an almost null structure. Then, the only possible values of the real index $r$ of $\mcN$ are
\begin{itemize}
 \item $r=1$ when $\epsilon=0$, and
 \item $r=0,1$ when $\epsilon=1$.
\end{itemize}
\end{lem}
This allows us to make the following definitions.
\begin{defn}\label{defn-OptStr}
 An \emph{almost optical (or Robinson) structure} $(\mcN , \mcK)$ on $\mcM$ is an almost null structure on $\mcM$ of real index $1$, and $\mcK \rightarrow \mcM$ is the real null line bundle defined by $\mcN \cap \overline{\mcN} =: \C \otimes \mcK$. We say that $(\mcN , \mcK)$ is an \emph{optical (or Robinson) structure} on $\mcM$, when both $\mcN$ and its orthogonal complement $\mcN^\perp$ are integrable, in the sense of Frobenius. The \emph{screen space} of $(\mcN, \mcK)$ is the rank-$(2m-2 + \epsilon)$ bundle $\mcK^\perp / \mcK \rightarrow \mcM$.
\end{defn}

\begin{rem}
 In even dimensions, there is some redundancy in the formulation of Definition \ref{defn-OptStr}, since in this case $\mcN = \mcN^\perp$. This remark applies to the subsequent definitions and results which also hold in even dimensions.
\end{rem}

The next proposition gives some properties of optical structures. We omit the proof, which can already be found in \cite{Nurowski2002} in the even-dimensional case, while the odd-dimensional case is a straightforward extension.
\begin{prop}
 Let $(\mcN , \mcK)$ be an optical structure on $\mcM$. Then,
\begin{enumerate}
 \item the integrable curves of the generators of $\mcK$ are null geodesics;
 \item the flow of $\mcK$ preserves the null structure $\mcN$ and its orthogonal complement $\mcN^\perp$;
 \item the screen space $\mcK^\perp / \mcK$ is naturally equipped with a complex structure when $\epsilon = 0$, and with a CR structure when $\epsilon = 1$.
\end{enumerate}
\end{prop}

From now on, we shall assume that $(\mcM, \bm{g})$ is a five-dimensional spacetime unless otherwise stated. We shall make use of the spin coefficient notation given in Appendix \ref{App-Notation}, which is a modified version of that of \cite{Garc'ia-ParradoG'omez-Lobo2009}. We can always assume that the metric takes the form
\begin{align*}
 g_{a b} = 2 k_{\lp{a}} \ell_{\rp{b}} + 2 m_{\lp{a}} \bar{m}_{\rp{b}} + u_a u_b \, ,
\end{align*}
where $k_a$, $\ell_a$ and $u_a$ are real, and $m_a$ complex basis $1$-forms. Indices are lowered and raised via the metric. This gives us two canonical almost optical structures $(\mcN_0 , \mcL)$ and $(\mcN_1, \mcK)$, where
\begin{align*}
 \mcN_0 & = \sspan\{ \ell^a, \bar{m}^a \} \, , & \mcL & = \sspan\{ \ell^a \} \, , \\
 \mcN_1 & = \sspan\{ k^a, \bar{m}^a \} \, ,  & \mcK & = \sspan\{ k^a \} \, .
\end{align*}
When no ambiguity can arise, we shall drop the subscript on $\mcN_{0}$ and $\mcN_{1}$. The following lemma is an easy application of the commutation relations \eqref{eq-Commut-Rel-2}, \eqref{eq-Commut-Rel-3} and \eqref{eq-Commut-Rel-7}.
\begin{lem}
Let $(\mcN , \mcK)$ be the almost optical structure defined above. Then
\begin{itemize}
 \item $\mcN$ is integrable if and only if
\begin{align} \label{eq-quasifolcond_5_L}
 \kappa = \sigma & = 0 \, , & \chi & = \psi \, ,
\end{align}
 \item \emph{both} $\mcN$ and $\mcN^\perp$ are integrable if and only if
\begin{align} \label{eq-folcond_5_L}
 \kappa = \sigma = \mdigamma = \phi = \chi = \psi = \eta = 0 \, .
\end{align} 
\end{itemize}
\end{lem}

\begin{rem}
 The conditions $\kappa = \mdigamma = 0$ is equivalent to $k^a$ being geodetic. Hence, the lemma tells us that the integrability of $\mcN$ only is not sufficient to guarantee the geodesy of $k^a$, while the integrability of both $\mcN$ and $\mcN^\perp$ is.
\end{rem}

\begin{rem}\label{rem-spinor}
 Optical structures, or more generally null structures, can be conveniently expressed in terms of \emph{pure spinor fields} \cite{Cartan1981}. Roughly, a spinor field $\bm{\xi}$ on a (pseudo-) Riemannian manifold $( \mcM, \bm{g} )$ defines a totally null distribution $\mcN_{\bm{\xi}}$ on the complexified tangent bundle $\Tgt \mcM$, given by
\begin{align*}
  \mcN_{\bm{\xi}} & := \left\{ \bm{X} \in \Gamma (\C \otimes \Tgt \mcM ) | \bm{X} \cdot \bm{\xi} = 0 \right\} \, ,
\end{align*}
where the $\cdot$ denotes the Clifford action. When $\mcN_{\bm{\xi}}$ is maximal totally null, then $\bm{\xi}$ is said to be \emph{pure}. One can then talk of the real index of a pure spinor field in the obvious way. An optical structure $(\mcN,\mcK)$ can then be defined by a pure spinor field (up to scale) of real index $1$, which satisfies the integrability conditions
\begin{align*}
 \bm{Y} \cdot \left( \nabla_{\bm{X}} \bm{\xi} \right) & = 0 \, , & & \mbox{and} & \bm{X} \cdot \bm{Y} \cdot \left( \nabla_{\bm{Z}} \bm{\xi} \right) & = 0 \, , 
\end{align*}
for all vector fields $\bm{X}, \bm{Y} \in \Gamma(\mcN_{\bm{\xi}})$, $\bm{Z} \in \Gamma(\mcN_{\bm{\xi}}^\perp)$, and where $\nabla$ denotes the spin connection induced from the Levi-Civita connection. Sections of $\mcK$ can be expressed as the tensor product of $\bm{\xi}$ and its conjugate. We shall avoid the use of spinors in this paper, and the reader is referred to the literature for further details \cites{Cartan1981,Penrose1984,Penrose1986,Hughston1988,Budinich1987,Budinich1988,Budinich1989,Kopczy'nski1992,Kopczy'nski1997,Mason2010}.
\end{rem}

\subsection{Integrability conditions on the Weyl conformal tensor}
The existence of an optical structure on a Lorentzian manifold is subject to integrability conditions on the Weyl tensor as given by the following proposition.
\begin{prop}
 Let $(\mcM , \bm{g} )$ be a Lorentzian manifold endowed with an almost optical structure $(\mcN , \mcK)$. If $\mcN$ and $\mcN^\perp$ are integrable, then the Weyl tensor satisfies
\begin{align} \label{eq-int-cond-Weyl}
 \bm{C} ( \bm{X}, \bm{Y}, \bm{Z}, \bm{W} ) & = 0 \, ,
\end{align}
for all $\bm{X}, \bm{Y} \in \Gamma( \mcN^\perp )$ and $\bm{Z}, \bm{W} \in \Gamma (\mcN )$.
\end{prop}
We omit the proof, which is already given in \cite{Mason2010} for the complex even-dimensional case, the odd-dimensional case being similar.

In four dimensions, condition \eqref{eq-int-cond-Weyl} is equivalent to the Petrov type I condition. While it is always satisfied in four dimensions, in the sense that one can find an almost null structure such that at least one of the five components of the Weyl tensor vanishes, condition \eqref{eq-int-cond-Weyl} is non-trivial in higher dimensions. This makes it somewhat more `special' than in four dimensions, but we shall reserve the term for a stronger algebraic condition defined later.

\paragraph{Sachs equations}
The integrability condition \eqref{eq-int-cond-Weyl} can also be read off directly from the Sachs equations, which can be obtained from the Ricci identity by choosing a parallely transported frame along the affinely parametrised null geodesic generated by the vector field $k^a$. Here, we present a subset of these equations in the five-dimensional context by setting $\epsilon = \kappa = \mdigamma = \pi = \chi = \mqoppa = 0$ in the Ricci identity given in Appendix \ref{App-Ricci}:
\begin{align}
\tag{\ref{eq-R-8}}
D \sigma & =  - \eta \psi - \rho \sigma - \bar{\rho} \sigma + \Psi_0  \, , \\
\tag{\ref{eq-R-10}}
D \psi & =  - \psi \bar{\rho} - \bar{\psi} \sigma - \psi \yogh + \Psi^0_0  \, , \\
\tag{\ref{eq-R-12}}
D \phi & =  - \mae \psi - \phi \bar{\rho} - \sigma \upsilon + \Psi_{ 0 3 }  \, , \\
\tag{\ref{eq-R-34}}
D \eta & =  - \eta \rho - \bar{\eta} \sigma - \eta \yogh + \Psi^0_0 \, .
\end{align}
It becomes apparent that $\Psi_0$, $\Psi^0_0$ and $\Psi_{0 3}$ are obstructions to the propagation of the CR structure of $\mcK^\perp / \mcK$ along the geodesic, and this, independently of the Einstein equations.

\section{Algebraically special spacetimes}
 The integrability condition \eqref{eq-int-cond-Weyl} is not sufficient to guarantee the existence of an optical structure on a manifold. In four dimensions, the Goldberg-Sachs theorem, however, tells us that a vacuum four-dimensional spacetime admits an optical structure if and only if it is algebraically special, i.e. it is at least of Petrov type II. In five dimensions, we first define which condition an `algebraically special' spacetime should satisfy, and we then demonstrate that it naturally leads to a partial generalisation of the Goldberg-Sachs thorem in five dimensions.

\subsection{A higher-dimensional version of Petrov type II spacetimes}
Here, we introduce an alternative characterisation of the algebraic speciality of the Weyl tensor, which we shall argue in the next section generalises the Petrov type II condition.
\begin{defn}\label{defn-AlgSp}
 Let $(\mcM , \bm{g} )$ be a Lorentzian manifold endowed with an almost optical structure $(\mcN , \mcK)$. We say that the Weyl tensor is \emph{algebraically special (relative to $\mcN$)} if it satisfies
\begin{align} \label{eq-AlgSpW}
 \bm{C} ( \bm{X}, \bm{Y}, \bm{Z}, \cdot ) & = 0 \, ,
\end{align}
for all $\bm{X}, \bm{Y} \in \Gamma ( \mcN^\perp )$ and $\bm{Z} \in \Gamma (\mcN)$. If in addition, the Weyl tensor does not degenerate any further, we say that the Weyl tensor is \emph{generically} algebraically special.
\end{defn}

Reexpressing the above conditions in the five-dimensional Lorentzian context using the five-dimensional Newman-Penrose spin coefficient formalism of Appendix \ref{App-Notation}, we have
\begin{lem}
Let $(\mcM , \bm{g} )$ be a five-dimensional Lorentzian manifold. Then, the Weyl tensor is algebraically special if ond only if
\begin{align} \label{eq-Weylcond_5_L}
 \Psi_0 = \Psi^0_0 = \Psi_{0 3} = \Psi_1  = \Psi^0_1 = \Psi^1_1 = \Psi^2_1 = \Psi^0_{1 3} = \Psi^1_{1 3} = \Psi^2_{1 3} = \Psi_{1 4} = 0 \, .
\end{align}
\end{lem}

\paragraph{Jordan normal form of the Weyl curvature operator}
In four dimensions, algebraically special spacetime are equivalently characterised by the eigenvalue structure of the Weyl curvature operator, i.e. the Weyl tensor viewed as an operator on $2$-forms or bivectors. More precisely, the Weyl tensor is algebraically special if its corresponding operator has a repeated eigenvalue. Recently, this approach has been considered for higher-dimensional spacetimes \cites{Hervik2010,Coley2010,Coley2010a}. In particular, it was found that the Weyl operator of a five-dimensional spacetime of CMPP type II has at least three eigenvalues of (at least) multiplicity 2 \cite{Coley2010}. Here, we assume the stronger condition that the Weyl tensor is generically algebraically special relative to an almost optical structure. Details of the following results are relegated to Appendix \ref{App-Notation}.
\begin{prop}
Let $(\mcM, \bm{g})$ be a five-dimensional spacetime, and suppose that the Weyl tensor is generically algebraically special relative to an optical structure $\mcN$. Then, the Weyl curvature operator $C\ind{_\alpha ^\beta}$ has Segre characteristic $[2 2 2 (1 1) 1 1]$ with eigenvalues
\begin{align*}
 \Psi_2 \, , \, \, \, \bar{\Psi}_2 \, , \, \, \, \Psi_2^0 \, , \, \, \, - \Psi_2^0 \, , \, \, \,
  - (\Psi_2 + \bar{\Psi}_2 ) + \sqrt{ (\Psi_2 - \bar{\Psi}_2)^2 + (\Psi_2^0)^2 } \, , \, \, \, - (\Psi_2 + \bar{\Psi}_2 ) - \sqrt{ (\Psi_2 - \bar{\Psi}_2)^2 + (\Psi_2^0)^2 } \, ,
\end{align*}
respectively. Further, the $3$-forms and $2$-forms annihilating $\mcN$ and $\mcN^\perp$ respectively are eigenforms of $C \ind{_\alpha^\beta}$. 
\end{prop}

It is also worth mentioning that this Jordan normal form does \emph{not} imply the algebraic condition \eqref{eq-AlgSpW}: if we weaken this condition by assuming $\Psi_{1 4} \neq 0$, the Weyl curvature operator has the same Jordan normal form and eigenvalues as the algebraically special case. Some properties of its eigenforms are pointed out in Appendix \ref{App-Notation}.

Finally, computations have shown that assuming fewer vanishing components of the Weyl tensor does not lead to such concise expressions of the eigenvalues.

\paragraph{Relation to the CMPP classification}
Roughly, the idea behind the CMPP classification of the Weyl tensor is that spacetimes are classified according to the vanishing of components of the Weyl tensor of a given boost weight \cite{Coley2004}. According to this scheme, the Weyl tensor of a CMPP Type II spacetime has vanishing components of boost weights $2$ and $1$. While in four dimensions, this agrees with the algebraic conditions given in the present paper, it is clear that in five (or higher) dimensions, neither the integrability condition \eqref{eq-int-cond-Weyl} for an almost optical structure $\mcN$, nor the algebraic speciality \eqref{eq-AlgSpW} of the Weyl tensor relative to it,  fits into the CMPP classification. This can be inferred directly from Table \ref{table-Weylboostweight} of the appendix. The algebraic speciality implies the vanishing of not only all components of boost weights $2$ and $1$ ($\Psi_0$, $\Psi^0_0$, $\Psi^0_1$, $\Psi_{0 3}$, $\Psi_1$, $\Psi^1_1$, and $\Psi^2_1$), but also of three components of boost weight $0$ ($\Psi^0_{1 3}$, $\Psi^1_{1 3}$, and $\Psi^2_{1 3}$), and one of boost weight $-1$ ($\Psi_{1 4}$). Thus, a spacetime algebraically special with respect to an almost optical structure is a stronger condition than the CMPP type II condition.

\paragraph{Relation to the De Smet classification}
The De Smet classification of spacetimes was first proposed in \cite{De2005} as a generalisation of the classification of the  Weyl conformal spinor to five dimensions. The classification was more recently refined in \cite{Godazgar2010}, where reality conditions on spin space are taken into account. We recall that the Weyl tensor in four and five dimensions can be viewed as a totally symmetric element $\Psi_{A B C D}$ of the spinor algebra. In four dimensions, one can always express the  Weyl conformal spinor as a symmetric product of spinors. However, this is not always the case in five dimensions. In the De Smet classification, spacetimes are categorised according to the factorisibility of $\Psi_{A B C D}$. This procedure imposes stringent algebraic conditions on the Weyl tensor, and the geometry of algebraically special spacetimes in the sense of De Smet may be excessively restrictive, as some examples in \cite{Godazgar2010} show. In fact, an algebraically special spacetime relative to an optical structure can still be algebraically general in the sense of De Smet. The point will be illustrated by Example \ref{exa-KerrNUTAdS} and Remark \ref{rem-Kerr5} below.

\subsection{A five-dimensional Goldberg-Sachs theorem}
We now come to the main result of the paper. The Cotton-York tensor is defined in Appendix \ref{App-Notation}
\begin{thm}[Partial five-dimensional Lorentzian Goldberg-Sachs theorem]\label{thm-GS-5-L}
 Let $(\mcM, \bm{g})$ be a five-dimensional Lorentzian manifold, and $(\mcN, \mcK)$ an almost optical structure. Suppose that the Weyl tensor is generically algebraically special relative to $\mcN$, and the Cotton-York tensor satisfies
\begin{align}\label{eq-CYcond_5_L}
\bm{A}( \bm{Z}, \bm{X}, \bm{Y}) & = 0 \, ,
\end{align}
for all $\bm{X}, \bm{Y} \in \Gamma ( \mcN^\perp ) $, $\bm{Z} \in \Gamma (\mcN )$.
Then, the almost optical structure $(\mcN , \mcK)$ is integrable.
\end{thm}

\begin{proof}
By assumption, the Weyl tensor satisfies relations \eqref{eq-Weylcond_5_L}, and the Cotton-York tensor
\begin{align*}
 A_{1 1 2} = A_{1 1 0} = A_{1 2 0} = A_{2 1 2} = A_{2 2 0} = 0 \, .
\end{align*}
By the genericity assumption, no other components of the Weyl tensor vanish identically except possibly at isolated points. This reduces parts of the Bianchi identity given in Appendix \ref{App-Bianchi} to the following algebraic equations:
\begin{align}\tag{\ref{eqB1}}
0 & = \sigma \Psi^0_2 + 3 \sigma \Psi_2 - \kappa \bar{\Psi}^1_3 \\
\tag{\ref{eqB2}}
0 & = \kappa \Psi^0_2 - 3 \kappa \Psi_2 \\ 
\tag{\ref{eqB13}}
0 & =  (  - \chi + 2 \psi )  \Psi^0_2 + \psi \Psi_2 +  ( \chi + \psi )  \bar{\Psi}_2 - \mdigamma \bar{\Psi}_3 - \kappa \Psi^2_3 \\
\tag{\ref{eqB14}}
0 & = \phi \Psi_2 - \phi \bar{\Psi}_2 + \kappa \bar{\Psi}^0_4 +  (  - 2 \chi + \psi )  \bar{\Psi}^1_3 +  ( \chi + \psi )  \bar{\Psi}_3 - \mdigamma \bar{\Psi}_4 + 2 \sigma \bar{\Psi}^2_3 + \sigma \Psi^2_3 \\
\tag{\ref{eqB15}}
0 & =  - \mdigamma \Psi_2 + \mdigamma \bar{\Psi}_2
\\
\tag{\ref{eqB16}}
0 & =  ( 2 \chi - \psi )  \Psi^0_2 - \chi \Psi_2 +  (  - \chi - \psi )  \bar{\Psi}_2 - \mdigamma \bar{\Psi}^1_3 + \mdigamma \bar{\Psi}_3 - 2 \kappa \bar{\Psi}^2_3 + \kappa \Psi^2_3 \\
\tag{\ref{eqB42}}
0 & =  (  - \eta + 2 \psi )  \Psi^0_2 +  (  - \chi - 3 \eta + \psi )  \Psi_2 +  ( \chi + \psi )  \bar{\Psi}_2 - \mdigamma \bar{\Psi}^1_3 - \mdigamma \bar{\Psi}_3 - 2 \kappa \bar{\Psi}^2_3 - \kappa \Psi^2_3
\\
\tag{\ref{eqB43}}
0 & = \mdigamma \Psi^0_2 + \mdigamma \bar{\Psi}_2
\\
\tag{\ref{eqB44}}
0 & = \phi \Psi^0_2 - \phi \bar{\Psi}_2 + \kappa \bar{\Psi}^0_4 +  (  - 2 \chi - \eta )  \bar{\Psi}^1_3 +  ( \chi + \psi )  \bar{\Psi}_3 - \mdigamma \bar{\Psi}_4 + \sigma \Psi^2_3
\\
\tag{\ref{eqB45}}
0 & =  ( 2 \chi + \eta )  \Psi^0_2 +  (  - \chi - 3 \eta + \psi )  \Psi_2 +  (  - \chi - \psi )  \bar{\Psi}_2 - \mdigamma \bar{\Psi}^1_3 + \mdigamma \bar{\Psi}_3 - 2 \kappa \bar{\Psi}^2_3 + \kappa \Psi^2_3
\\
\tag{\ref{eqB47}}
0 & =  - 2 \kappa \Psi^0_2
\\
\tag{\ref{eqB48}}
0 & = 2 \sigma \Psi^0_2 - 2 \kappa \bar{\Psi}^1_3
\\
\tag{\ref{eqB51}}
0 & =  ( \chi - \eta )  \Psi^0_2 +  (  - \chi - 3 \eta )  \Psi_2 - \mdigamma \bar{\Psi}^1_3 - 2 \kappa \bar{\Psi}^2_3
\\
\tag{\ref{eqB53}}
0 & = \phi \Psi^0_2 - \phi \Psi_2 +  (  - \eta - \psi )  \bar{\Psi}^1_3 - 2 \sigma \bar{\Psi}^2_3
\\
\tag{\ref{eqB54}}
0 & =  ( \eta + \psi )  \Psi^0_2 +  (  - 3 \eta + \psi )  \Psi_2
\end{align}
We now show that the connection components satisfy relations \eqref{eq-folcond_5_L}.
\begin{itemize}
 \item From equation \eqref{eqB47}, one obtains $\kappa=0$ immediately.
 \item It follows from equation \eqref{eqB48} that $\sigma=0$.
 \item By the genericity assumption, we have that $\mdigamma = 0$ by equations \eqref{eqB15} or \eqref{eqB43}.
 \item Now, \eqref{eqB51} + \eqref{eqB54} - \eqref{eqB13} - \eqref{eqB16} gives $0= - 6 \eta \Psi_2 $, i.e. $\eta=0$.
 \item Hence, from equations \eqref{eqB51}, \eqref{eqB54}, \eqref{eqB13} and \eqref{eqB16} again, we deduce $\chi = \psi = 0$.
 \item Finally, any of equations \eqref{eqB14}, \eqref{eqB44} or \eqref{eqB53} yields $\phi = 0$. 
\end{itemize}
At this stage, we observe that the condition on the Weyl tensor is frame independent as can be seen from the gauge transformations given in Appendix \ref{App-Notation}. Hence, all the required connection components are zero in any null frame.
\end{proof}

Before we proceed, a number of remarks are in order.
\begin{rem}
 There arises the issue of whether the algebraic speciality \eqref{eq-CYcond_5_L} is the \emph{mininum} condition which guarantees the existence of an optical structure. So, suppose that $\Psi_{1 4} \neq 0$ (and for simplicity we maintain the same condition on the Cotton-York tensor), which by gauge invariance, is the only next weakest assumption possible. Then, the only algebraic equation involving $\phi$ is
\begin{align}\tag{\ref{eqB53}}
0 & = \phi \Psi^0_2 - \phi \Psi_2 +  (  - \eta - \psi )  \bar{\Psi}^1_3 + \yogh \bar{\Psi}_{ 1 4 } - 2 \sigma \bar{\Psi}^2_3 \, .
\end{align}
In particular, $\phi$ cannot vanish unless $\yogh$ vanishes itself, which is impossible under our assumptions.
 We must thus conclude that any weaker condition on the Weyl tensor will not guarantee the integrability of the optical structure, i.e. the integrability of \emph{both} $\mcN$ and $\mcN^\perp$. In fact, it can be checked that weaker conditions (such as the one we have just considered) may guarantee the integrability of the null distribution $\mcN$, but not its orthogonal complement $\mcN^\perp$.
\end{rem}
\begin{rem}
 As pointed earlier, the weaker condition given by $\Psi_{1 4} \neq 0$ still implies a very simple eigenvalue structure of the Weyl curvature operator $C\ind{_\alpha ^\beta}$. In the light of the previous remark, there thus appears to be an affinity between the integrability of $\mcN$ and the Jordan normal form of $C\ind{_\alpha ^\beta}$.
\end{rem}
\begin{rem}
 It is worth noting that the assumptions of Theorem \eqref{thm-GS-5-L} are \emph{conformally invariant}: if Theorem \eqref{thm-GS-5-L} holds for a frame, then it will also hold for another one conformally related to the first one. To see the conformal invariance of the optical structure, suppose for definiteness that each basis $1$-form is rescaled by a positive factor $\Omega$. Then, denoting the connection components relative to the rescaled by a $\hat{}$, we have
\begin{align*}
 \hat{\kappa} & = \Omega^{-1} \kappa \, , & \hat{\mdigamma} & = \Omega^{-1} \mdigamma \, ,  & \hat{\sigma} & = \Omega^{-1} \sigma \, ,
 & \hat{\chi} & = \Omega^{-1} \chi \, , & \hat{\psi} & = \Omega^{-1} \psi \, , & \hat{\eta} & = \Omega^{-1} \eta \, , & \hat{\phi} & = \Omega^{-1} \phi \, , 
\end{align*}
and the invariance of the integrability follows. The conformal invariance of the theorem itself is less straightforward, but can be checked from the tranformation laws of the Bianchi identity. This property is well-known in four dimensions \cites{Penrose1986,Gover2010}, and we defer its generalisation to higher dimensions for a future publication \cite{Taghavi-Chaberta}.
\end{rem}

We end this section by another aspect of the Goldberg-Sachs theorem, whose proof is a straightforward modification of that of Theorem \ref{thm-GS-5-L}.
\begin{prop}\label{prop-GS-5-L-conv}
 Let $(\mcM, \bm{g})$ be a five-dimensional Lorentzian manifold, and $(\mcN, \mcK)$ an optical structure. Suppose that the Weyl tensor is  algebraically special relative to $\mcN$. Then, the Cotton-York tensor satisfies
\begin{align*}
\bm{A}( \bm{Z}, \bm{X}, \bm{Y}) & = 0 \, ,
\end{align*}
for all $\bm{X}, \bm{Y} \in \Gamma ( \mcN^\perp ) $, $\bm{Z} \in \Gamma (\mcN )$.
\end{prop}

\subsection{Simplification of the Bianchi identity}
As a consequence of Theorem \ref{thm-GS-5-L}, the Bianchi identity for an algebraically special spacetime simplifies considerably. Here we record some of the most notable equations, which present an elegant simplicity:
\begin{align}\tag{\ref{eqB21}}
D \Psi_2 & =  - A _{ 1  1  \tilde{1} } - A _{ 2  1  \bar{2} } - \rho \Psi^0_2 - 3 \rho \Psi_2 \\
\tag{\ref{eqB22}}
\delta \Psi_2 & = A _{ 1  \tilde{1}  2 } - A _{ 2  2  \bar{2} } - \tau \Psi^0_2 + 3 \tau \Psi_2 - \rho \bar{\Psi}^1_3 \\
\tag{\ref{eqB91}}
\mathring{\delta} \Psi_2 & = A _{ 1  \tilde{1}  0 } + A _{ 2  \bar{2}  0 } +  ( \mstigma - \upsilon )  \Psi^0_2 +  ( \mstigma + \upsilon )  \Psi_2 +  ( 2 \rho - \yogh )  \bar{\Psi}^2_3
\end{align}
\begin{align}\tag{\ref{eqB75}}
D \Psi^0_2 & =  - A _{ 0  1  0 } - A _{ 1  1  \tilde{1} } - 2 \yogh \Psi^0_2 - \yogh \Psi_2 - \yogh \bar{\Psi}_2 \\
\tag{\ref{eqB79}}
\delta \Psi^0_2 & =  - A _{ 0  2  0 } + A _{ 1  \tilde{1}  2 } +  (  - \mae + 2 \tau )  \Psi^0_2 + \mae \bar{\Psi}_2 + \bar{\rho} \bar{\Psi}^1_3 - \yogh \bar{\Psi}_3
\end{align}
\begin{align}\tag{\ref{eqB76}}
D \bar{\Psi}^1_3 & =  - A _{ 2  1  \tilde{1} } +  (  - 2 \bar{\pi} + 2 \tau )  \Psi^0_2 - \mae \Psi_2 + \mae \bar{\Psi}_2 +  (  - 2 \bar{\epsilon} - \yogh )  \bar{\Psi}^1_3 - \yogh \bar{\Psi}_3 \\
\tag{\ref{eqB80}}
\delta \bar{\Psi}^1_3 & = A _{ 2  \tilde{1}  2 } - 2 \bar{\lambda} \Psi^0_2 +  (  - 2 \mae - 2 \bar{\alpha} + 2 \tau )  \bar{\Psi}^1_3 + \mae \bar{\Psi}_3 - \yogh \bar{\Psi}_4
\end{align}
\begin{align}
\tag{\ref{eqB34}}
D \bar{\Psi}^2_3 & = A _{ 0  1  \tilde{1} } +  ( \mqoppa - \mstigma )  \Psi^0_2 +  ( \mqoppa - \upsilon )  \Psi_2 +  (  - \mstigma + \upsilon )  \bar{\Psi}_2 +  (  - \epsilon - \bar{\epsilon} - 2 \rho )  \bar{\Psi}^2_3 - \rho \Psi^2_3 \\ 
\tag{\ref{eqB36}}
\delta \bar{\Psi}^2_3 & =  - A _{ 0  \tilde{1}  2 } +  (  - 2 \omega + \xi )  \Psi^0_2 +  ( \omega + \xi )  \Psi_2 + \omega \bar{\Psi}_2 + \rho \bar{\Psi}^0_4 +  ( \mstigma - 2 \upsilon + \bar{\upsilon} )  \bar{\Psi}^1_3 \\
& \qquad {} +  (  - \mstigma + \upsilon )  \bar{\Psi}_3 +  (  - \bar{\alpha} - \beta + 2 \tau )  \bar{\Psi}^2_3 - \tau \Psi^2_3 \nonumber
\end{align}
We also have the following equations, which do not depend on the derivatives of the Weyl tensor components:
\begin{align}\tag{\ref{eqB81}}
0 & =  - A _{ 1  2  \bar{2} } +  ( 2 \rho - 2 \bar{\rho} )  \Psi^0_2 - \yogh \Psi_2 + \yogh \bar{\Psi}_2 \\
\tag{\ref{eqB35}}
0 & =  - A _{ 0  2  \bar{2} } +  (  - \upsilon + \bar{\upsilon} )  \Psi^0_2 +  ( \mstigma - \bar{\upsilon} )  \Psi_2 +  (  - \mstigma + \upsilon )  \bar{\Psi}_2 + \bar{\rho} \bar{\Psi}^2_3 - \rho \Psi^2_3 \, .
\end{align}
Further equations can be derived, such as
\begin{align}
 0 & = - A _{ 0  2  0 } - A _{ \tilde{1}  1  2 } + ( \mae + 2 \tau ) \Psi^0_2 - \mae \Psi_2  +  ( 2 \rho  - \yogh )  \bar{\Psi}^1_3 \, ,
\end{align}
which is given by adding equation \eqref{eqB79} to equation \eqref{eqB82}.

\subsection{Further degeneracy of the Weyl tensor}
In four dimensions, it is well-known that further algebraic specialities of the Weyl tensor (Petrov types D, III, and N) together with appropriate conditions on the Cotton-York tensor still guarantee the integrability of an almost optical structure. In this section, we briefly consider what stronger algebraic conditions on the Weyl tensor imply for the connection components. We start by the following remark:
\begin{rem}\label{rem-proof}
 From the proof of Theorem \ref{thm-GS-5-L}, one can see that only the assumption that the components $\Psi_2$ and $\Psi^0_2$ of the Weyl tensor be generic, and in particular non-vanishing, is needed, and none of the remaining components ($\Psi_3$, $\Psi^0_3$, $\Psi^1_3$,$\Psi^2_3$,$\Psi_4$,$\Psi^1_4$) play any part in the vanishing of the connection components.
\end{rem}

\subsubsection{Five-dimensional generalisation of Petrov type D spacetimes}
In four dimensions, Petrov type D vacuum spacetimes are characterised by the existence of \emph{two} distinct integrable optical structures. In \cite{Mason2010}, it is argued that a $(2m+\epsilon)$-dimensional generalisation of the Petrov type D condition should also give rise to $2^m$ distinct optical structures. Recall that on a five-dimensional spacetime where the metric has the form
\begin{align*}
 g_{a b} = 2 k_{\lp{a}} \ell_{\rp{b}} + 2 m_{\lp{a}} \bar{m}_{\rp{b}} + u_a u_b \, ,
\end{align*}
one can canonically define two almost optical structures $(\mcN_0, \mcL)$ and $(\mcN_1, \mcK)$, where $\mcN_0 := \sspan \{ \ell^a, m^a \}$ and $\mcN_1 := \sspan \{ k^a, m^a \}$.
From Remark \ref{rem-proof}, we know that if all the components of the Weyl tensor vanish with the exception of $\Psi_2$ and $\Psi^0_2$, then $(\mcN_1,\mcK)$ is integrable. But by symmetry, $(\mcN_0,\mcL)$ must also be integrable. We have thus proved the following corollary to Theorem \ref{thm-GS-5-L}.
\begin{cor}
  Let $(\mcM, g)$ be a five-dimensional Lorentzian manifold, and $(\mcN_0, \mcL)$ and $(\mcN_1, \mcK)$ be the almost optical structures defined as above. Suppose that the Weyl tensor and the Cotton-York tensor satisfy
\begin{align*}
 \bm{C} ( \bm{X}, \bm{Y}, \bm{Z}, \cdot ) & = 0 \, , & \bm{A}( \bm{Z}, \bm{X}, \bm{Y}) & = 0 \, ,
\end{align*}
for all $\bm{X}, \bm{Y} \in \Gamma ( \mcN_i ^\perp ) $, $\bm{Z} \in \Gamma (\mcN_i )$, for each $i=0,1$, and that the Weyl tensor does not degenerate otherwise (i.e. generic $\Psi_2, \Psi^0_2 \neq 0$).
Then, the almost optical structures $(\mcN_0, \mcL)$ and $(\mcN_1, \mcK)$ are integrable.
\end{cor}
In this case, the Bianchi identity simplifies drastically:
\begin{align}\tag{\ref{eqB21}}
D \Psi_2 & = - A _{ 1  1  \tilde{1} } - A _{ 2  1  \bar{2} } - \rho \Psi^0_2 - 3 \rho \Psi_2 \\
\tag{\ref{eqB11}}
\tilde{D} \Psi_2 & =A _{ \tilde{1}  1  \tilde{1} } - A _{ \bar{2}  \tilde{1}  2 } - \mu \Psi^0_2 - 3 \mu \Psi_2 \\
\tag{\ref{eqB22}}
\delta \Psi_2 & =A _{ 1  \tilde{1}  2 } - A _{ 2  2  \bar{2} } - \tau \Psi^0_2 + 3 \tau \Psi_2 \\
\tag{\ref{eqB12}}
\bar{\delta} \Psi_2 & =A _{ \tilde{1}  1  \bar{2} } + A _{ \bar{2}  2  \bar{2} } - \pi \Psi^0_2 + 3 \pi \Psi_2 \\
\tag{\ref{eqB46}}
\mathring{\delta} \Psi_2 & =A _{ \tilde{1}  1  0 } + A _{ \bar{2}  2  0 } +  ( \mqoppa - \bar{\upsilon} )  \Psi^0_2 +  ( \mqoppa + \bar{\upsilon} )  \Psi_2
\end{align}
\begin{align}\tag{\ref{eqB75}}
D \Psi^0_2 & = - A _{ 0  1  0 } - A _{ 1  1  \tilde{1} } - 2 \yogh \Psi^0_2 - \yogh \Psi_2 - \yogh \bar{\Psi}_2 \\
\tag{\ref{eqB83}}
 - \tilde{D} \Psi^0_2 & =A _{ 0  \tilde{1}  0 } - A _{ \tilde{1}  1  \tilde{1} } + 2 \wynn \Psi^0_2 + \wynn \Psi_2 + \wynn \bar{\Psi}_2 \\
\tag{\ref{eqB82}}
 - \delta \Psi^0_2 & = - A _{ 0  2  0 } - A _{ 2  2  \bar{2} } + 2 \mae \Psi^0_2 - \mae \Psi_2 - \mae \bar{\Psi}_2 \\
\label{eqXtra55-70}
\mathring{\delta} \Psi^0_2  & = - A _{ 1  \tilde{1}  0 } - A _{ \tilde{1}  1  0 } +  (  \mqoppa +  \mstigma + \upsilon + \bar{\upsilon} )  \Psi^0_2 + \tfrac{1}{2} ( \mqoppa + \mstigma - \upsilon - \bar{\upsilon} ) ( \Psi_2 + \bar{\Psi}_2 )
\end{align}
where equation \eqref{eqXtra55-70} can be obtained from equations \eqref{eqB55} and \eqref{eqB70}.
Other relations can be obtained using the following equations
\begin{align}
\tag{\ref{eqB19}}
0 & =A _{ 0  2  \bar{2} } +  ( \upsilon - \bar{\upsilon} )  \Psi^0_2 +  ( \mqoppa - \upsilon )  \Psi_2 +  (  - \mqoppa + \bar{\upsilon} )  \bar{\Psi}_2 \\
\tag{\ref{eqB34}}
0 & =A _{ 0  1  \tilde{1} } +  ( \mqoppa - \mstigma )  \Psi^0_2 +  ( \mqoppa - \upsilon )  \Psi_2 +  (  - \mstigma + \upsilon )  \bar{\Psi}_2 \\
\tag{\ref{eqB35}}
0 & = - A _{ 0  2  \bar{2} } +  (  - \upsilon + \bar{\upsilon} )  \Psi^0_2 +  ( \mstigma - \bar{\upsilon} )  \Psi_2 +  (  - \mstigma + \upsilon )  \bar{\Psi}_2 \\
\tag{\ref{eqB76}}
0 & = - A _{ 2  1  \tilde{1} } +  (  - 2 \bar{\pi} + 2 \tau )  \Psi^0_2 - \mae \Psi_2 + \mae \bar{\Psi}_2 \\
\tag{\ref{eqB81}}
0 & = - A _{ 1  2  \bar{2} } +  ( 2 \rho - 2 \bar{\rho} )  \Psi^0_2 - \yogh \Psi_2 + \yogh \bar{\Psi}_2 \\
\tag{\ref{eqB89}}
0 & = - A _{ \tilde{1}  2  \bar{2} } +  (  - 2 \mu + 2 \bar{\mu} )  \Psi^0_2 + \wynn \Psi_2 - \wynn \bar{\Psi}_2
\end{align}
When the Einstein equations are imposed, the Cotton-York tensor vanishes and these equations acquire an even simpler form.

\begin{rem}
 Evidently, whether the Weyl tensor is algebraically special with respect to $(\mcN_1, \mcK)$ only, or to both $(\mcN_0, \mcL)$ and $(\mcN_1, \mcK)$, the eigenvalues of the Weyl curvature operator remain the same. In the latter case, the Segre characteristic becomes $[(11)(11)(11)(11)(11)11]$.
\end{rem}

\begin{exa}[Five-dimensional (Lorentzian) Kerr-NUT-AdS black hole]\label{exa-KerrNUTAdS}
 The five-dimensional\linebreak (Lorentzian) Kerr-NUT-AdS black hole metric can be cast in the form
\begin{align*}
 \bm{g} & = 2 \bm{\theta}^1 \odot \tilde{\bm{\theta}}^{\tilde{1}} + 2 \bm{\theta}^2 \odot \bar{\bm{\theta}}^{\bar{2}} + \bm{e}^0 \otimes \bm{e}^0
\end{align*}
where the basis $1$-forms $\{ \bm{\theta}^1 , \tilde{\bm{\theta}}^{\tilde{1}} , \bm{\theta}^2 , \bar{\bm{\theta}}^{\bar{2}} , \bm{e}^0 \}$ are given by
\begin{align*}
 \bm{\theta}^1 & := \frac{1}{\sqrt{2 Q}} \left( - \dd r + Q \left( \dd \psi + y^2 \dd \phi \right) \right) \, , & \bm{\theta}^2 & := \frac{1}{\sqrt{2 P}} \left( \ii \dd y + P \left( \dd \psi - r^2 \dd \phi \right) \right) \, , \\
 \tilde{\bm{\theta}}^{\tilde{1}} & := \frac{1}{\sqrt{2 Q}} \left( \dd r + Q \left( \dd \psi + y^2 \dd \phi \right) \right) \, , & \bm{e}^0 & := \frac{\sqrt{c}}{ry} \left( \dd \psi + (y^2 - r^2) \dd \phi - y^2 r^2 \dd \chi \right) \, ,
\end{align*}
and
\begin{align*}
 Q & := \frac{X}{U} \, , & X & := -a r^2 + b r^4 - \frac{c}{r^2} - 2 A \, , & U & := y^2 + r^2 \, , \\
 P & := \frac{Y}{V} \, , & Y & := a y^2 + b y^4 + \frac{c}{y^2} - 2 B \, , & V & := - ( y^2 + r^2 )\, .
\end{align*}
Here, $a$ and $c$ are constants related to the angular momenta of the black hole solutions, $b$ is related to the cosmological constant, and $A$ and $B$ are the mass and NUT parameters. It is shown \cite{Mason2010} to be algebraically special relative to the two optical structures\footnote{Here, the basis  vector fields are denoted by a subscript.} $\mcN_0 = \left\{ \tilde{\bm{\theta}}_{\tilde{1}} ,\bm{\theta}_2 \right\}$ and $\mcN_1 = \left\{ \bm{\theta}_{1} ,\bm{\theta}_2 \right\}$. Hence, it has only two non-vanishing components, which are found to be
\begin{align*}
 \Psi_2 & = 2 (A-B) \frac{(r - \ii y)^2}{(y^2 + r^2)^3} \, , &  \Psi^0_2 = 2 \frac{(A-B)}{(y^2 + r^2)^2} \, .
\end{align*}
At a glance, one notices that
\begin{align} \label{eq-Kerr5Cond}
 \Psi^0_2 & = \sqrt{\Psi_2 \bar{\Psi}_2 } \, .
\end{align}
This relation in fact implies that the  Weyl conformal spinor $\Psi _{A B C D}$ can be factorised into the square of a symmetric $2$-valent spinor in normal form, i.e. it is of De Smet type $(\underline{2},\underline{2})$.
\end{exa}

\begin{rem}\label{rem-Kerr5} It is clear from this example that the Kerr-NUT-AdS metric is only one of many possible metrics of `Petrov type D'. Weakening the relation \eqref{eq-Kerr5Cond} should lead to new families of five-dimensional metrics, which are not algebraically special in the sense of De Smet.
\end{rem}

\subsubsection{Other assumptions}
Further degeneracy of the Weyl tensor will \emph{not} guarantee the integrability of the optical structure in general. We shall make no attempt to give a complete study of every case in this paper, but we shall briefly comment on four cases to illustrate the situation:
\begin{itemize}
 \item By a slight alteration of its proof, one can demonstrate that Theorem \ref{thm-GS-5-L} remains valid under the stronger assumption that \emph{either} $\Psi_2$ \emph{or} $\Psi^0_2$ vanishes. Such assumptions may lead to the vanishing of more connection components if more components of the Cotton-York vanish (this can be seen directly from the simplified Bianchi identity).
 \item Assuming further $\Psi_2 = \Psi^0_2 = 0$ and, for simplicity, Ricci-flatness, and no other components vanishing, one can show that this does not put the connection component $\phi$ under any algebraic restriction, and thus the orthogonal complement of the null distribution itself may not be integrable. However, enough connection components vanish to conclude that the null distribution remains integrable.
 \item Choosing $\Psi_2 = \Psi^0_2 = \Psi_3 = \Psi^1_3 = \Psi^2_3 = 0$ and Ricci-flatness, one checks that there are no algebraic equations involving either $\phi$ or $\chi$, and so the null distribution is not necessarily integrable.
 \item At the other extreme, one can consider the case of a vacuum spacetime, where all components of the Weyl tensor vanish except for $\Psi^0_3$. One can show that such spacetime must be of De Smet type $(\underline{1 1}, \underline{1 1})$, i.e. its Weyl tensor is solely determined by a spinor field. As explained in Remark \ref{rem-spinor}, this spinor field also defines a null distribution of real index $1$ of the complexifed tangent bundle, and hence an almost optical structure $(\mcN,\mcK)$. This is one of the cases studied in \cite{Godazgar2010}. The Bianchi identity then tells us that not only does the spacetime admit an optical structure (i.e. equations \eqref{eq-folcond_5_L} are satisfied), but 
 also the congruence of rays generated by $\mcK$ is divergence-free ($\yogh = \rho = 0$, i.e. $\nabla^a k_a = 0$ , for any $k^a \in \Gamma (\mcK)$), and $\upsilon - \bar{\upsilon} = 0$.
\end{itemize}

\section{Optical structures and the black ring}
\subsection{An algebraically general spacetime admitting optical structures}
While we have shown that there exists an algebraic condition on the Weyl tensor, which guarantees the existence of an optical structure in five dimensions, we now use the black ring solution found in \cites{Emparan2002} as a counterexample to the converse.

The five-dimensional black ring metric takes the form
\begin{align*}
 \bm{g} & = 2 \bm{\theta}^{1} \odot \tilde{\bm{\theta}}^{\tilde{1}} + 2 \bm{\theta}^{2} \odot \tilde{\bm{\theta}}^{\tilde{2}} + \bm{e}^0 \otimes \bm{e}^0  \, ,
\end{align*}
where the basis $1$-forms $\{ \bm{\theta}^1 , \tilde{\bm{\theta}}^{\tilde{1}} , \bm{\theta}^2 , \tilde{\bm{\theta}}^{\tilde{2}} , \bm{e}^0 \}$ are given by
\begin{align*}
 \bm{\theta}^1 & := \frac{R \sqrt{-F(x) G(y)}}{\sqrt{2} (x-y)} \left(\frac{\sqrt{F(y)}}{G(y)} \dd y + \ii  \dd \psi \right) \, , &
 \tilde{\bm{\theta}}^{\tilde{1}} & := \frac{R \sqrt{-F(x) G(y)}}{\sqrt{2} (x-y)} \left( \frac{\sqrt{F(y)}}{G(y)} \dd y - \ii  \dd \psi \right) \, , \\
 \bm{\theta}^2 & := \frac{R F(y)\sqrt{G(x)}}{\sqrt{2} (x-y) \sqrt{F(x)}} \left( \frac{\sqrt{F(x)}}{G(x)} \dd x + \ii  \dd \phi \right) \, , &
 \tilde{\bm{\theta}}^{\tilde{2}} & := \frac{R F(y)\sqrt{G(x)}}{\sqrt{2} (x-y) \sqrt{F(x)}} \left( \frac{\sqrt{F(x)}}{G(x)} \dd x - \ii  \dd \phi \right) \, , \\
 \bm{e}^0 & := \sqrt{-\frac{F(x)}{F(y)}} \left( \dd t + R \sqrt{\lambda \nu} (1 + y ) \dd \psi \right) \, , &
\end{align*}
and
\begin{align*}
 F(\xi) & := 1 - \lambda \xi \, , & G(\xi) & := (1 - \xi^2)(1-\nu \xi) \, .
\end{align*}
Here, $R$, $\lambda$ and $\mu$ are positive constants with $\lambda, \nu <1$.

In terms of the dual basis of vector fields $\{ \bm{\theta}_1 , \tilde{\bm{\theta}}_{\tilde{1}} , \bm{\theta}_2 , \tilde{\bm{\theta}}_{\tilde{2}} , \bm{e}_0 \}$. the canonical distributions
\begin{align*}
 \mcN_1 & := \sspan\{ \bm{\theta}_1, \tilde{\bm{\theta}}_{\tilde{2}} \} \, , & \mcN_2 & := \sspan \{ \tilde{\bm{\theta}}_{\tilde{1}}, \bm{\theta}_2 \} \, , \\ 
 \mcN_0 & := \sspan \{ \tilde{\bm{\theta}}_{\tilde{1}}, \tilde{\bm{\theta}}_{\tilde{2}} \} & \mcN_{12} & := \sspan\{ \bm{\theta}_1, \bm{\theta}_2 \} \, ,
\end{align*}
are maximal totally null. Depending on the range of $x$ and $y$, the basis $1$-forms may be real or complex, and accordingly the eigenvalues of the metric will change signs \cite{Pravda2005}. For specificity, we shall consider only the following two regions:
\begin{align*}
 \mcA & := \left\{ (x,y,\phi,\psi,t) : -1 < x < 1 \, , y < -1 \right\} \, , &
 \mcB & := \left\{ (x,y,\phi,\psi,t) : -1 < x < 1 \, , \tfrac{1}{\lambda} < y < \tfrac{1}{\nu} \right\} \, .
\end{align*}
 In both regions, the metric is of Lorentzian signature. Then,
\begin{itemize}
 \item in region $\mcA$, we have $\overline{\bm{e}_0} = - \bm{e}_0$,
$\overline{\mcN_{1}} =\mcN_2$ and $\overline{\mcN_0} =\mcN_{12}$, so that $\mcN_0$ and $\mcN_1$ are of real index $0$.
 \item in region $\mcB$, $\overline{\bm{e}_0} = \bm{e}_0$;
$\overline{\mcN_{1}} =\mcN_{12}$ and $\overline{\mcN_{0}} =\mcN_2$, so that $\mcN_0$ and $\mcN_1$ are of real index $1$, and thus define two distinct almost optical structures.
\end{itemize}
Now, computing the structure equations gives
\begin{align*}
  \dd \bm{e}^0 & = \left( - \frac{ \lambda (x-y) \sqrt{G(x)}}{2 \sqrt{2}  R F(x) F(y)} ( \bm{\theta}^1 + \tilde{\bm{\theta}}^{\tilde{1}} ) + \frac{ \lambda (x-y) \sqrt{G(y)}}{2 \sqrt{2}  R \sqrt{-F(x)} F(y) \sqrt{F(y)}} ( \bm{\theta}^2 + \tilde{\bm{\theta}}^{\tilde{2}} ) \right) \wedge \bm{e}^0 \\
   & \qquad - \ii \frac{\lambda \nu (x-y)^2 \sqrt{G(x)} \sqrt{G(y)} }{R \sqrt{F(x)} \sqrt{F(y)} F(y) \sqrt{G(y)}} \bm{\theta}^2 \wedge \tilde{\bm{\theta}}^{\tilde{2}} \, , \\
  \dd \bm{\theta}^1 & = \left( - \frac{\sqrt{G(x)} ( 2 F(x) + \lambda (x-y))}{ \sqrt{2} R F(x) F(y)} ( \bm{\theta}^2 + \tilde{\bm{\theta}}^{\tilde{2}} ) \right. \\
   & \qquad \left. + \frac{\sqrt{-G(y)}}{2 \sqrt{2} R \sqrt{F(x)} \sqrt{F(y)} G(y)} \left( 2 G(y) + \derv{}{y} G(y) (x-y) \right) \tilde{\bm{\theta}}^{\tilde{1}} \right) \wedge \bm{\theta}^1 \, , \\
\dd \bm{\theta}^2 & = \left( \frac{1}{\sqrt{2}} \frac{(F(y) - \lambda(x-y))\sqrt{-G(y)}}{R \sqrt{F(x)}\sqrt{F(y)} F(y)} (\bm{\theta}^1 + \tilde{\bm{\theta}}^{\tilde{1}}) \right. \\
   & \qquad \left. + \frac{1}{2 \sqrt{2} R F(y) F(x) \sqrt{G(x)}} \left(-2 G(x)F(x) + \derv{}{x}G(x) F(x) (x-y) + \lambda G(x) (x-y) \right) \tilde{\bm{\theta}}^{\tilde{2}} \right) \wedge \bm{\theta}^2 \, ,
\end{align*}
and similarly for $\tilde{\bm{\theta}}^{\tilde{1}}$ and $\tilde{\bm{\theta}}^{\tilde{2}}$. By the dual form of the Frobenius theorem, one therefore sees that both distributions $\mcN_0$ and $\mcN_1$ and their orthogonal complements $(\mcN_0)^\perp$ and $(\mcN_1)^\perp$ are integrable. Hence,
\begin{itemize}
 \item in region $\mcA$, $\mcN_0$ and $\mcN_1$ define two distinct CR structures;
 \item in region $\mcB$, $\mcN_0$ and $\mcN_1$ define two distinct optical structures.
\end{itemize}

\begin{rem}
 The structure equations tells further us that the basis $1$-forms $\bm{\theta}^1$, $\tilde{\bm{\theta}}^{\tilde{1}}$, $\bm{\theta}^2$ and $\tilde{\bm{\theta}}^{\tilde{2}}$ are hypersurface orthogonal. Specialising to the case where $\mcN_0$ and $\mcN_1$ are of real index $0$, one can then introduce complex coordinates $z$ and $w$ defined by
\begin{align*}
  \dd z & := \frac{\sqrt{F(x)}}{G(x)} \dd x + \ii \dd \phi \, , & \dd w & := \frac{\sqrt{F(y)}}{G(y)} \dd y + \ii \dd \psi \, .
\end{align*}
The case where $\mcN_0$ and $\mcN_1$ are of real index $1$ is similar except that one obtains one complex coordinate, and two real null coordinates.
\end{rem}

It is however well-known that the black ring metric is generically of CMPP type $\mathrm{I}_i$, and algebraically general in the De Smet classification \cites{Pravda2005,Godazgar2010}. In fact, a \emph{Maple} computation shows that
the Weyl tensor is \emph{not} algebraically special with respect to any of the null structures. In the notation of Appendix \ref{App-Notation}, and restricting ourselves to region $\mcB$, we find
\begin{align}
 \Psi_0 = \Psi^0_0 = \Psi_{ 0 3 } = 0 = \Psi_4 = \Psi^0_4 = \Psi_{ 1 4 } \, , \label{eq-trivial}
\end{align}
which are none other than the integrability conditions for the optical structures $(\mcN_0, \mcL)$ and $(\mcN_1, \mcK)$. In addition, the remaining vanishing components lead to the relations
\begin{align}
 \Psi_2 & = \bar{\Psi}_2 \, , & \Psi^2_1 & = \bar{\Psi}^2_1 \, , & \Psi^2_3 & = \bar{\Psi}^2_3 \, , & \Psi^1_{1 3} & = - \bar{\Psi}^2_{1 3} \, , \label{eq-RWC}
\end{align}
i.e. some components of the Weyl tensor that are generically complex satisfy some reality conditions.

\subsection{A five-dimensional spacetime with eight null structures}
In four dimensions, the maximal number of optical or null structures that a non-trivial spacetime can admit is four\footnote{Here, we have counted complex pairs of null structures as two distinct null structures.}, in which case the spacetime is of Petrov type D. In \cite{Mason2010}, we argued that the generalisation of Petrov Type D spacetimes in $(2m + \epsilon)$ dimensions is characterised by the existence of $2^m$ null structures. Here, we show that a five-dimensional spacetime may admit more than four null structures.

It is shown in \cite{Emparan2006}, that setting $\lambda=1$ in the black ring metric, one recovers the five-dimensional Myers-Perry black hole with only \emph{one} rotation coefficient. In this case, one can choose \cite{Pravda2005} 
\begin{multline*}
 \bm{k}^\pm := \frac{1}{(x^2-1)(-1 + \nu y)} \left( \frac{\nu y x - y + \nu x + 1 - 2 \nu y}{x-y} R \parderv{}{t} - \sqrt{\nu} \parderv{}{\psi} \right) \\
\pm \sqrt{\frac{\nu x - 1}{(x-y)(y-1)}} \left( \parderv{}{x} + \frac{y^2-1}{x^2-1} \parderv{}{y} \right) \, ,
\end{multline*}
which turn out to be Weyl aligned null directions (WAND) for the Weyl tensor.   Now, it is shown in \cite{Mason2010} that both WANDs $\bm{k}^+$ and $\bm{k}^-$ are sections of two real null lines bundles $\mcK^+$ and $\mcK^-$ determined by two independent optical structures $(\mcN^\pm,\mcK^\pm)$, (i.e. null distributions of real index $1$). But it is clear that the null structures $\mcN_0$ and $\mcN_1$ of real index $0$, which we discovered in the generic black ring, subsist in this limit. In fact, in terms of the null basis, one has
\begin{multline*}
 \bm{k}^\pm = - \frac{R \sqrt{1-y}}{\sqrt{1-x}(x+1)(x-y)} \bm{e}_0 \mp \frac{R \sqrt{1-y}}{\sqrt{2} \sqrt{(1-x^2)(x-y)^3}} \left( \bm{\theta}_2 + \tilde{\bm{\theta}}_{\tilde{2}} \right) \\
\mp \frac{R \sqrt{y^2-1}}{\sqrt{2} \sqrt{(1-x)(1 - \nu y)(x-y)^3}(x+1)} \left( \sqrt{1-\nu x} \mp \ii \sqrt{\nu (x-y)} \right) \bm{\theta}_1 \\ \mp \frac{R \sqrt{y^2-1}}{\sqrt{2} \sqrt{(1-x)(1- \nu y)(x-y)^3}(x+1)} \left( \sqrt{1-\nu x} \pm \ii \sqrt{\nu (x-y)} \right) \tilde{\bm{\theta}}_{\tilde{1}} \, .
\end{multline*}
Therefore, \emph{the five-dimensional Myers-Perry black hole with one rotation coefficient admits two pairs of conjugate null distributions ($\mcN^\pm$, $\overline{\mcN}^\pm$) of real index $1$ and another two pairs of null distributions ($\mcN_0$, $\mcN_1$, $\mcN_2$, $\mcN_{12}$) of real index $0$}.

\section{Conclusion and outlook}
We have introduced the notion of an optical structure $(\mcN, \mcK)$  on odd-dimensional spacetimes, and we have proposed an algebraic condition on the Weyl tensor, which generalises the Petrov type II condition, and we have shown that in five dimensions, the corresponding Weyl curvature operator admits a remarkably simple Jordan normal form. This definition is justified by Theorem \ref{thm-GS-5-L}, which extends the celebrated Goldberg-Sachs theorem to five dimensions. To be precise, an algebraically special spacetime with respect to $\mcN$, provided the Cotton-York is degenerate on $\mcN$, is generically endowed with an optical structure. As a result, the Ricci and Bianchi identities simplify considerably, which should open up new avenues leading to the discovery of five-dimensional solutions to Einstein's field equations in a systematic fashion, as in the four-dimensional case \cite{Stephani2003}.

Further, we have briefly highlighted that stronger algebraic conditions on the Weyl tensor, which includes a five-dimensional `Petrov' type D condition, may or may not ensure the integrability of an almost optical structure. It would be interesting to conduct a comprehensive study of these special cases.

It is also worth emphasising that the existence of an optical structure on a spacetime guarantees the existence of complex coordinates on the null foliation of complexified spacetime. It is anticipated that under certain conditions these coordinates could provide suitable (semi-)complex coordinates on the real spacetime, thus reducing the number of independent components of the metric.

While Theorem \ref{thm-GS-5-L} is stated and proved in the context of five-dimensional spacetimes, the underlying geometry is more easily formulated in terms of null distributions in the complexification. In fact, we have checked using \emph{Mathematica} that the above theorem \emph{does} hold in the complex case too. More generally, the following conjecture has been verified to be true in dimensions six and seven.
\begin{conjec}
 Let $(\mcM , \bm{g})$ be a $(2m+\epsilon)$-dimensional complex Riemannian manifold. Let $\mcN$ denote a maximal totally null subbundle of the tangent bundle, i.e. $\mcN \subset \mcN^\perp$ and $\mcN$ has rank $m$. Suppose that the Weyl tensor and the Cotton-York tensor satisfy
\begin{align*}
 \bm{C} ( \bm{X}, \bm{Y}, \bm{Z}, \cdot ) & = 0 \, , & \bm{A} (\bm{Z}, \bm{X}, \bm{Y}) & = 0 \, ,
\end{align*}
respectively, for all $\bm{X}, \bm{Y} \in \mcN^\perp$ and $\bm{Z} \in \mcN$. Then, assuming the Weyl tensor does not degenerate any further, the distributions $\mcN$ and $\mcN^\perp$ are integrable in the sense of Frobenius, i.e.
\begin{align*}
 [ \mcN , \mcN ] & \in \mcN \, , &  [\mcN^\perp, \mcN^\perp] & \in \mcN^\perp \, .
\end{align*}
\end{conjec}
As in the five-dimensional case, there are real versions of the conjecture, which can be obtained by taking appropriate real slices. In a forthcoming paper \cite{Taghavi-Chabert}, we shall lay out the details of the conjecture together with a proof for the six- and seven-dimensional versions of the theorem.

We have also shown that the five-dimensional black ring admits optical structures in certain regions of spacetime, and integrable null distributions of real index $0$ in others. Yet, this solution is \emph{not} algebraically special. This counterexample thus invalidates the converse of Theorem \ref{thm-GS-5-L}, in contradistinction to the four-dimensional case. We also pointed out that the Kerr black hole limit of the black ring solution retains the integrable null distributions of real index $0$, while gaining a pair of optical structures.
It would then be instructive to determine which geometric properties differentiate optical structures that are algebraically special, from those that are not. While not every optical structure can be obtained from Theorem \ref{thm-GS-5-L}, algebraically special spacetimes should nevertheless provide a wide class of solutions admitting optical structures. In \cite{Taghavi-Chaberta}, we study a special class of higher-dimensional Kerr-Schild metrics
\begin{align*}
 g_{a b} = \eta_{a b} + 2 H k_a k_b \, ,
\end{align*}
where $\eta_{a b}$ is the flat metric, $H$ a function, and $k^a$ is a real null geodesic vector field defined by an optical structure $(\mcN,\mcK)$. It turns out that such metrics are algebraically special with respect to $(\mcN,\mcK)$.

\appendix

\section{Setup and notation} \label{App-Notation}
In this appendix, we set up the notation used throughout this paper. We essentially follow and extend the index notation and the spin coefficient convention of \cite{Penrose1984}. While we have avoided the use of spinors, our notation is inspired by \cite{Garc'ia-ParradoG'omez-Lobo2009}, which is an extension of the Penrose-Newman formalism using four-spinor calculus in five dimensions. There are some minor differences between their convention and ours, which will be pointed out in due course.
 \subsection{Null basis}
Let $(\mcM , \bm{g} )$ be a five-dimensional Lorentzian manifold. The metric can then be put in the form
\begin{align*}
 g_{a b} = 2 k_{\lp{a}} \ell _{\rp{b}} + 2 m_{\lp{a}} \bar{m} _{\rp{b}} + u_a u_b \, ,
\end{align*}
where $k_a$, $\ell_a$ and $u_a$ are real, and $m_a$ complex $1$-forms. As usual, indices are raised and lowered via the metric, so $\ell^a$, $k^a$, $\bar{m}^a$, $m^a$ and $u^a$ are basis vector fields (in \cite{Garc'ia-ParradoG'omez-Lobo2009}, $u^a$ is chosen to be of norm $\sqrt{2}$). For future use, we note that this naturally defines two canonical almost optical structures $(\mcN_0 , \mcL)$ and $(\mcN_1, \mcK)$, where
\begin{align*}
 \mcN_0 & = \sspan\{ \ell^a, \bar{m}^a \} \, , & \mcL & = \sspan\{ \ell^a \} \, , &
 \mcN_1 & = \sspan\{ k^a, \bar{m}^a \} \, ,  & \mcK & = \sspan\{ k^a \} \, .
\end{align*}
It is convenient to introduce numerical indices, which are particularly useful for taking components of various tensorial quantities. Writing, $\delta \ind*{_a ^b}$ in place of $g \ind{_a ^b}$, we assign the following numerical values to the basis vectors and $1$-forms:
\begin{align*}
 	k^a & = \delta \ind*{^a _1} \, , & \ell^a & = \delta \ind*{^a _\tldo} \, , & \bar{m}^a & = \delta \ind*{^a _2} \, , & u^a & = \delta \ind*{^a _0} \, , & \mbox{and} & & 
	k_a & = \delta \ind*{_a ^\tldo} \, , & \ell_a & = \delta \ind*{_a ^1} \, , & m_a & = \delta \ind*{_a ^2} \, , & u_a & = \delta \ind*{_a ^0} \, ,
\end{align*}
respectively.
Thus, for instance, components of a tensor $A \ind{_{a b} ^c}$ in this basis will be denoted
\begin{align*}
 A \ind{_{1 2} ^{\tilde{1}}} := A \ind{_{a b} ^c} k^a \bar{m}^b k_c \, ,
\end{align*}
and so on.

With this notation, we shall denote the induced basis bivectors and basis $2$-forms by
\begin{align*}
 \delta_\alpha^{[i j]} & := \sqrt{2} \delta^{i}_{\lb{a}} \delta^{j}_{\rb{b}} \, , & \delta^\alpha_{[i j]} & := \sqrt{2} \delta_{i}^{\lb{a}} \delta_{j}^{\rb{b}} \, , & & (i,j \in \{1, \tilde{1}, 2 , \bar{2}, 0 \} )\, ,
\end{align*}
respectively. These are clearly dual to each other, i.e. $\delta_\alpha^{[i j]} \delta^\alpha_{[k \ell]} = 2 \delta_{\lb{k}}^i \delta_{\rb{\ell}}^j$. Here, a Greek index $\alpha$ can be viewed as a multi-index, or in the abstract index language \cite{Penrose1984}, as being `isomorphic' to $[a_1 a_2]$, say.

There is an induced metric on the space of bivectors given by
\begin{align*}
 g_{\alpha \beta} := g_{a_1 \lb{b_1}} g_{\rb{b_2} a_2} \, ,
\end{align*}
which can be reexpressed as
\begin{align*}
 g_{\alpha \beta} = - \delta^{1 \tilde{1}}_{\alpha} \delta^{1 \tilde{1}}_{\beta} - \delta^{2 \bar{2}}_{\alpha} \delta^{2 \bar{2}}_{\beta} + 2 \delta_{\lp{\alpha}}^{1 2} \delta_{\rp{\beta}}^{\tilde{1} \bar{2}} + 2 \delta_{\lp{\alpha}}^{\tilde{1} 2} \delta_{\rp{\beta}}^{1 \bar{2}} + 2 \delta_{\lp{\alpha}}^{1 0} \delta_{\rp{\beta}}^{\tilde{1} 0} + 2 \delta_{\lp{\alpha}}^{2 0} \delta_{\rp{\beta}}^{\bar{2} 0} \, .
\end{align*}
Similarly, the induced metric on the space of $2$-forms is given by
\begin{align*}
 g^{\alpha \beta} := g^{a_1 \lb{b_1}} g^{\rb{b_2} a_2} \, ,
\end{align*}
and clearly satisfies
\begin{align*}
 g_{\alpha \gamma} g^{\gamma \beta} = \delta_{a_1}^{\lb{b_1}} \delta_{a_2}^{\rb{b_2}} = \delta_\alpha^\beta \, .
\end{align*}

\subsection{Connection components}
Let $\nabla$ be the Levi-Civita connection on $(\mcM, \bm{g})$, i.e. the unique tosion-free connection preserving the metric $g_{ab}$, and define, for convenience, the differential operators
\begin{align*}
 D & := k^a \nabla_a \, , & \tilde{D} & := \ell^a \nabla_a  \, , &
 \delta & := \bar{m}^a \nabla_a  \, , & \mathring{\delta} & := u^a \nabla_a  \, .
\end{align*}
The components of the connection $1$-form can then be expressed as follows:
\begin{align*}
 (\epsilon + \bar{\epsilon}) & := \Gamma \ind{_{1 1 \tldo}} = \ell^a D k_a \, ,
& (\gamma + \bar{\gamma}) & := \Gamma \ind{_{\tldo 1 \tldo}} = \ell^a \tilde{D} k_a \, ,
&  (\bar{\alpha} +\beta) & := \Gamma \ind{_{2 1 \tldo}} = \ell^a \delta k_a \, , & \\
 \kappa & := \Gamma \ind{_{1 1 2}} = \bar{m}^a D k_a \, ,
& \tau & := \Gamma \ind{_{\tldo 1 2 }}  = \bar{m}^a \tilde{D} k_a \, , &
 \sigma & := \Gamma \ind{_{2 1 2 }} = \bar{m}^a \delta k_a \, , \\
 \bar{\pi} & := \Gamma \ind{_{1 \tldo 2}} =  \bar{m}^a D \ell \ind*{_a}  \, , &
\bar{\nu} & := \Gamma \ind{_{\tldo \tldo 2}} = \bar{m}^a \tilde{D}  \ell \ind*{_a}  \, , &
 \bar{\lambda} & :=  \Gamma \ind{_{2 \tldo 2 }} = \bar{m}^a \delta \ell \ind*{_a}  \, , \\
 (\epsilon - \bar{\epsilon}) & := \Gamma \ind{_{1 2 \bar{2}}} = m^a D \bar{m}_a \, , &
 (\gamma - \bar{\gamma}) & := \Gamma \ind{_{\tldo 2 \bar{2}}} = m^a \tilde{D} \bar{m}_a \, , &
 (\beta - \bar{\alpha}) & := \Gamma \ind{_{2 2 \bar{2}}} = m^a \delta \bar{m}_a \, , \\
 \mqoppa  & := \Gamma \ind{_{1 \tldo 0}}  = u^a  D \ell_a \, ,
& \msampi  & := \Gamma \ind{_{\tldo \tldo 0}} = u^a \tilde{D} \ell_a \, , &
\xi & := \Gamma \ind{_{2 \tldo 0}} = u^a \delta \ell_a  \, , \\
\mdigamma  & := \Gamma \ind{_{1 1 0 }} = u^a D k_a \, ,
& \mstigma  & := \Gamma \ind{_{\tldo 1 0 }} = u^a \tilde{D} k_a \, ,&
 \psi  & := \Gamma \ind{_{2 1 0 }} = u^a \delta k_a \, , \\
 \chi & := \Gamma \ind{_{1 2 0 }} = u^a D \bar{m}_a \, , &
\omega & := \Gamma \ind{_{\tldo 2 0 }} = u^a \tilde{D} \bar{m}_a \, ,&
 \phi & := \Gamma \ind{_{2 2 0}} = u^a \delta \bar{m}_a \, ,\\
( \theta + \bar{\theta} ) & := \Gamma \ind{_{0 1 \tldo }}  = \ell^a \mathring{\delta} k_a \, , \\
\eta & := \Gamma \ind{_{0 1 2 }} = \bar{m}^a \mathring{\delta} k_a \, , &  \rho & := \Gamma \ind{_{\bar{2} 1 2}} = \bar{m}^a \bar{\delta} k_a \, , \\
\bar{\zeta} & := \Gamma \ind{_{0 \tldo 2 }}  =  \bar{m}^a \mathring{\delta} \ell \ind*{_a}  \, , & \bar{\mu} & := \Gamma \ind{_{\bar{2} \tldo 2 }} = \bar{m}^a \bar{\delta} \ell \ind*{_a}  \, , \\
(\theta - \bar{\theta} ) & := \Gamma \ind{_{0 2 \bar{2}}} = m^a \mathring{\delta} \bar{m}_a \, , & \\
\wynn & := \Gamma \ind{_{0 \tldo 0}} = u^a \mathring{\delta} \ell_a  \, , & \\
\yogh & := \Gamma \ind{_{0 1 0}} = u^a \mathring{\delta} k_a \, , & \\
\mae  & := \Gamma \ind{_{0 2 0}} = u^a \mathring{\delta} \bar{m}_a \, , & \upsilon & := \Gamma \ind{_{\bar{2} 2 0}} = u^a \bar{\delta} \bar{m}_a \, ,
\end{align*}
Here, we have introduced the ancient Greek letters $\msampi$ (`sampi'), $\mdigamma$ (`digamma'), $\mqoppa$ (`qoppa') and $\mstigma$ (`stigma'), together with the Anglo-Saxon letters $\yogh$ (`yogh'), $\wynn$ (`wynn') and $\mae$ (`asc' pronounced `ash'). These are denoted by gothic letters in \cite{Garc'ia-ParradoG'omez-Lobo2009}.

The commutation relations among the basis vector fields are then found to be 
\begin{align}
  [ D , \tilde{D} ] & = - ( \gamma + \bar{\gamma} ) D - ( \epsilon + \bar{\epsilon} ) \tilde{D} + \left( \pi - \bar{\tau} \right) \delta + \left( \bar{\pi} - \tau \right) \bar{\delta} + \left( \mqoppa - \mstigma \right) \mathring{\delta} \, , \label{eq-Commut-Rel-1} \\
  [ D , \delta ] & = - \left( \bar{\pi} + \bar{\alpha} + \beta \right) D - \kappa \tilde{D} + \left(  \epsilon - \bar{\epsilon} - \bar{\rho} \right) \delta - \sigma \bar{\delta} + \left( \chi - \psi \right) \mathring{\delta} \, , \label{eq-Commut-Rel-2} \\
  [ D , \mathring{\delta} ] & = - \left( \mqoppa + \theta + \bar{\theta} \right) D  - \mdigamma \tilde{D} - \left( \bar{\chi} + \bar{\eta} \right) \delta - \left( \chi + \eta \right) \bar{\delta} - \yogh \mathring{\delta} \, , \label{eq-Commut-Rel-3}  \\
  [ \tilde{D} , \delta ] & = - \bar{\nu} D + \left( - \tau  + \bar{\alpha} + \beta \right) \tilde{D} + \left( \gamma - \bar{\gamma} - \mu \right) \delta - \bar{\lambda} \bar{\delta} + \left( \omega - \xi \right) \mathring{\delta} \, , \label{eq-Commut-Rel-4} \\
  [ \tilde{D} , \mathring{\delta} ] & = - \msampi  D + \left( - \mstigma + \theta + \bar{\theta} \right) \tilde{D} - \left( \bar{\omega} + \zeta \right) \delta - \left( \omega + \bar{\zeta} \right) \bar{\delta} - \wynn \mathring{\delta} \, , \label{eq-Commut-Rel-5} \\
  [ \delta , \bar{\delta} ] & = \left( - \mu + \bar{\mu}  \right) D + \left( - \bar{\rho} + \rho \right) \tilde{D} + ( \bar{\beta} - \alpha ) \delta - ( \beta - \bar{\alpha} ) \bar{\delta} + \left( \bar{\upsilon} - \upsilon \right) \mathring{\delta} \, , \label{eq-Commut-Rel-6} \\
  [ \delta , \mathring{\delta} ] & = \left( - \xi + \bar{\zeta} \right) D + \left( - \psi + \eta \right) \tilde{D} + \left( - \bar{\upsilon} - \theta + \bar{\theta} \right) \delta  - \phi \bar{\delta} - \mae \mathring{\delta} \, . \label{eq-Commut-Rel-7} 
\end{align}

\subsection{Curvature tensors}
Denote by $R \ind{_{a b c}^d}$ the Riemann curvature associated with $\nabla$, which we shall take to be conventionally defined by
\begin{align*}
 R \ind{_{a b d} ^c} V^d := 2 \nabla_{\lb{a}} \nabla_{\rb{b}} V^c \, .
\end{align*}
For the purpose of the article, we shall only be concerned with the decomposition of the Riemann tensor as
\begin{align*}
 R _{ a b c d } & = C _{ a b c d } - 4 g _{\lb{a} | \lb{c}} \Rho _{\rb{d}|\rb{b}} \, ,
\end{align*}
where $C_{a b c d}$ and $\Rho _{a b}$ are the Weyl tensor and the Rho tensor respectively.

\paragraph{The Weyl tensor and the Weyl curvature operator}
In five dimensions, the Weyl tensor has 35 independent components. In Lorentzian signature, and with respect to an almost null structure of real index 1, it is convenient to label them as follows
\begin{align*}
 \Psi_0 & := C_{1 2 1 2} \, , & \Psi_1 & := C_{1 2 1 \tilde{1}} \, , & \Psi_2 & := C_{1 2 \tilde{1} \bar{2}} & \Psi_3 & := C_{\tilde{1} \bar{2} \tilde{1} 1} \, , & \Psi_4 & := C_{\tilde{1} \bar{2} \tilde{1} \bar{2}} \, , \\
 \Psi^0_0 & := C_{1 0 1 2} \, , &  \Psi^0_1 & := C_{1 2 1 \bar{2}} \, , & \Psi^0_2 & := C_{1 0 \tilde{1} 0} & \Psi^0_3 & := C_{\tilde{1} \bar{2} \tilde{1} 2} \, , & \Psi^0_4 & := C_{\tilde{1} 0 \tilde{1} \bar{2}} \, , \\
& & \Psi^1_1 & := C_{1 0 2 0} \, , & & &  \Psi^1_3 & := C_{\tilde{1} 0 \bar{2} 0} \, ,  \\
& & \Psi^2_1 & := C_{2 0 1 \bar{2}} & & & \Psi^2_3 & := C_{\bar{2} 0 \tilde{1} 2} & & &  \\
\Psi_{0 3} & := C_{2 0 1 2} \, ,  & \Psi^1_{1 3} & := C_{1 0 \tilde{1} 2} & \Psi^0_{1 3} & := C_{1 2 \tilde{1} 2} & \Psi^2_{1 3} & := C_{\tilde{1} 0 1 \bar{2}} \, , & \Psi_{1 4} & := C_{\bar{2} 0 \tilde{1} \bar{2}} \, ,
\end{align*}
where $\Psi^0_1$, $\Psi^0_2$, $\Psi^0_3$ are real, and the 16 remaining components complex.
While we have made the same choice of independent components as \cite{Garc'ia-ParradoG'omez-Lobo2009}, their notation is more adapted to the use of a spinor basis of real index $1$. Despite the importance of pure spinors in the theory of optical structures, the algebraic speciality of the Weyl tensor is somewhat expressed more simply tensorially, and the aim of this notation is to reflect the curvature properties of the canonical optical structures, while extending the four-dimensional Petrov classification. In particular, the vanishing of the components $\Psi^{\cdot}_{0 \cdot}$ characterises the integrability condition of the almost optical structure $(\mcN_1, \mcK)$. If further the components $\Psi^{\cdot}_{1 \cdot}$ vanish, the spacetime is algebraically special with respect to $(\mcN_1, \mcK)$. By symmetry, components $\Psi^{\cdot}_{\cdot 4}$ and $\Psi^{\cdot}_{\cdot 3}$ are related to the curvature properties of the optical structure $(\mcN_0, \mcL)$. Components $\Psi^{\cdot}_{2}$ are the only non-vanishing components for `Petrov' type D spacetimes. 

The Weyl curvature tensor $C \ind{_{a b c d}}$ gives rise to the Weyl curvature operator $C \ind{_{a b} ^{c d}} = C_{a b e f} g^{c e} g^{d f}$, which we shall also denote $C \ind{_\alpha ^\beta}$ in the notation introduced earlier. In five dimensions, one can thus regard $C \ind{_\alpha ^\beta}$ as a tracefree $10 \times 10$ matrix, satisfying $C_{\alpha \beta} = C_{\beta \alpha}$, acting on $2$-forms, or dually, on bivectors. From the symmetries of the Weyl tensors, one can express the entries of $C \ind{_\alpha ^\beta}$ in terms of the components of $C_{a b c d}$. We omit the most general form of $C \ind{_\alpha ^\beta}$, and focus instead on Weyl tensors algebraically special with respect to the optical structure $(\mcN_1, \mcK)$. So, assume $\Psi^\cdot_{0 \cdot} = \Psi^\cdot_{1 \cdot} = 0$.
In the induced bases of $2$-forms $\{ \delta_\alpha ^{[ij]} \}$ and of bivectors $\{ \delta^\alpha _{[ij]} \}$ adapted to $(\mcN_1, \mcK)$, choose the index ordering $\left\{ [1 2], [1 \bar{2}], [1 0] ,[2 0], [\bar{2} 0], [1 \tilde{1}], [2 \bar{2}], [\tilde{1} 2], [\tilde{1} \bar{2}], [\tilde{1} 0] \right\}$. Then, the Weyl curvature operator $C \ind{_\alpha ^\beta}$ reduces to
{ \small
\begin{multline} \label{eq-WTriangul-AlgSp}
(\bm{C}) \ind{_\alpha ^\beta} = \\ \begin{pmatrix}
 \Psi_2 & 0 & 0 & 0 & 0 & 0 & 0 & 0 & 0 & 0 \\
 0 & \bar{\Psi}_2 & 0 & 0 & 0 & 0 & 0 & 0 & 0 & 0 \\
 0 & 0 & \Psi^0_2 & 0 & 0 & 0 & 0 & 0 & 0 & 0 \\
\bar{\Psi}^2_3 & 0 & \bar{\Psi}^1_3 & - \Psi^0_2 & 0 & 0 & 0 & 0 & 0 & 0 \\
 0 & \Psi^2_3 & \Psi^1_3 & 0 & - \Psi^0_2 & 0 & 0 & 0 & 0 & 0 \\
- \Psi_3 & - \bar{\Psi}_3 & \bar{\Psi}^2_3 + \Psi^2_3  & 0 & 0 & - \Psi^0_2 - \Psi_2 - \bar{\Psi}_2 & - \Psi_2 + \bar{\Psi}_2 & 0 & 0 & 0 \\
- \Psi_3 + \Psi^1_3 & \bar{\Psi}_3 - \bar{\Psi}^1_3 & - \Psi^2_3 + \bar{\Psi}^2_3 & 0 & 0 & - \Psi_2 + \bar{\Psi}_2 & - \Psi_2 - \bar{\Psi}_2 + \Psi^0_2   & 0 & 0 & 0 \\
\Psi_4 & \Psi^0_3 & \Psi^0_4 & 0 & \bar{\Psi}^2_3 & \Psi_3 & \Psi_3 - \Psi^1_3 & \Psi_2 & 0 & 0 \\
\Psi^0_3 & \bar{\Psi}_4 & \bar{\Psi}^0_4 & \Psi^2_3 & 0 & \bar{\Psi}_3 & - \bar{\Psi}_3 + \bar{\Psi}^1_3   & 0 & \bar{\Psi}_2 & 0 \\
\Psi^0_4 & \bar{\Psi}^0_4 & - 2 \Psi^0_3 & \Psi^1_3 & \bar{\Psi}^1_3 & - \bar{\Psi}^2_3 - \Psi^2_3  &  \Psi^2_3 - \bar{\Psi}^2_3 & 0 & 0 & \Psi^0_2
\end{pmatrix} \, .
\end{multline}
}
This is an almost lower triangular matrix, and after a little algebra, one can read off its eigenvalues. In fact, a computer computation reveals that the Weyl curvature operator has Segre characteristic $[2 2 2 (1 1) 1 1]$. Table \ref{table-evalues-WeylOp} lists the eigenvalues and the eigenforms for the repeated eigenvalues. These eigenforms can be seen to be null with respect to the induced metric on $2$-forms. The eigenforms for the simple eigenvalues are omitted for lack of space. It suffices to say that they are spanned by $k_{\lb{a}} \bar{m}_{\rb{b}}$, $k_{\lb{a}} m_{\rb{b}}$, $k_{\lb{a}} u_{\rb{b}}$, $k_{\lb{a}} \ell_{\rb{b}}$, and $\bar{m}_{\lb{a}} m_{\rb{b}}$.
\begin{table}[!ht]
\begin{center}
{\renewcommand{\arraystretch}{1.1}
\renewcommand{\tabcolsep}{0.2cm}
\begin{tabular}{|c|c|}
\hline
Eigenvalue ( multiplicity ) & Eigenvectors \\
\hline
$\Psi_2$ (2) & $k_{\lb{a}} \bar{m}_{\rb{b}}$ \\
$\bar{\Psi}_2$ (2) & $k_{\lb{a}} m_{\rb{b}}$ \\
$\Psi_2^0$ (2) & $k_{\lb{a}} u_{\rb{b}}$\\
\multirow{2}{*}{$-\Psi_2^0$ (2)} & $\frac{\bar{\Psi}^2_3}{\Psi^0_2 + \Psi_2} k_{\lb{a}} \bar{m}_{\rb{b}} +  \frac{\bar{\Psi}^1_3}{2 \Psi^0_2 } k_{\lb{a}} u_{\rb{b}} -  \bar{m}_{\lb{a}} u_{\rb{b}} $ \\
  & $\frac{\Psi^2_3}{\Psi^0_2 + \bar{\Psi}_2} k_{\lb{a}} m_{\rb{b}} +  \frac{\Psi^1_3 }{2 \Psi^0_2 } k_{\lb{a}} u_{\rb{b}} - m_{\lb{a}} u_{\rb{b}} $ \\
$- (\Psi_2 + \bar{\Psi}_2 ) + \sqrt{ (\Psi_2 - \bar{\Psi}_2)^2 + (\Psi_2^0)^2 }$ (1) & omitted \\
$- (\Psi_2 + \bar{\Psi}_2 ) - \sqrt{ (\Psi_2 - \bar{\Psi}_2)^2 + (\Psi_2^0)^2 }$ (1) & omitted \\
\hline
\end{tabular}}
\end{center}
\caption{Eigenvalues and eigenvectors of the Weyl curvature operator of a five-dimensional algebraically special spacetime}\label{table-evalues-WeylOp}
\end{table}

When $\Psi_{1 4} \neq 0$, the Weyl curvature operator has the same Jordan normal form. While the eigenvectors for the eigenvalues $\Psi_2$, $\bar{\Psi_2}$ and $\Psi_2^0$ remain the same as in the algebraically special case, the others differ: a complex conjugate pair of eigenvectors for $-\Psi^0_2$ now read
\begin{align*}
 \frac{\Psi^2_3 \bar{\Psi}^2_3 - \Psi_{1 4} \bar{\Psi}_{1 4}}{\Psi^0_2 + \Psi_2} k_{\lb{a}} \bar{m}_{\rb{b}}  + \Psi^2_3 \left(  \frac{\bar{\Psi}^1_3}{2 \Psi^0_2} k_{\lb{a}} u_{\rb{b}} - \bar{m}_{\lb{a}} u_{\rb{b}} \right) - \bar{\Psi}_{1 4} \left( \frac{\Psi^1_3}{2 \Psi^0_2} k_{\lb{a}} u_{\rb{b}} - m_{\lb{a}} u_{\rb{b}} \right) \, , \\
\frac{\Psi^2_3 \bar{\Psi}^2_3 - \Psi_{1 4} \bar{\Psi}_{1 4}}{\Psi^0_2 + \bar{\Psi}_2} k_{\lb{a}} m_{\rb{b}} + \bar{\Psi}^2_3 \left(  \frac{\Psi^1_3}{2 \Psi^0_2} k_{\lb{a}} u_{\rb{b}} - m_{\lb{a}} u_{\rb{b}} \right) - \Psi_{1 4} \left( \frac{\bar{\Psi}^1_3}{2 \Psi^0_2} k_{\lb{a}} u_{\rb{b}} - \bar{m}_{\lb{a}} u_{\rb{b}} \right) \, .
\end{align*}
Contrary to the algebraically special case (with $\Psi_{1 4} = 0$), these eigenforms are not null. The eigenvectors for $- (\Psi_2 + \bar{\Psi}_2 ) \pm \sqrt{ (\Psi_2 - \bar{\Psi}_2)^2 + (\Psi_2^0)^2 }$ are again rather complicated, but as in the algebraically special case, they are spanned by $k_{\lb{a}} \bar{m}_{\rb{b}}$, $k_{\lb{a}} m_{\rb{b}}$, $k_{\lb{a}} u_{\rb{b}}$, $k_{\lb{a}} \ell_{\rb{b}}$, and $\bar{m}_{\lb{a}} m_{\rb{b}}$.

\paragraph{The Rho and Cotton-York tensors}
The Rho tensor is symmetric, and thus has 15 independent components, which will be labelled in the usual fashion. Another important curvature
tensor is the \emph{Cotton-York tensor} defined to be
\begin{align*}
 A_{a b c} & := 2 \nabla_{\lb{b}} \Rho _{\rb{c} a} \, ,
\end{align*}
or equivalently, from the contracted Bianchi identity,
\begin{align*}
 \nabla^d C_{d a b c} & = -2 A_{a b c} \, .
\end{align*}
This implies that $A_{[a b c]} = 0$ and $A \ind{^a _{a b}} = 0$.

\subsection{Gauge transformations}
The null structure is invariant under the subgroup of the sim group generated by $( \R \oplus \uu (1) ) \ltimes \R^3$. This group can be  split into three types of transformations:
\begin{itemize}
 \item null rotations fixing $k^a$
\begin{align}\label{eq-Nrot}
\begin{aligned}
 k^a & \mapsto \hat{k}^a = k^a \, , &
 \ell^a & \mapsto \hat{\ell}^a = \ell^a + z m^a + \bar{z} \bar{m}^a + r u^a - ( z \bar{z} + \tfrac{1}{2} r^2 ) k^a\\
 m^a & \mapsto \hat{m}^a = m^a - \bar{z} k^a \, , &
 u^a & \mapsto \hat{u}^a = u^a - r k^a \, .
\end{aligned}
\end{align}
where $z \in \C$, $r \in \R$;
 \item boosts
\begin{align}\label{eq-boost}
 k^a & \mapsto \hat{k}^a = b k^a \, , &
 \ell^a & \mapsto \hat{\ell}^a = b^{-1} \ell^a \, , &
 m^a & \mapsto \hat{m}^a = m^a \, , &
 u^a & \mapsto \hat{u}^a = u^a \, .
\end{align}
where $b \in \R$;
 \item unitary rotations
\begin{align}\label{eq-Urot}
 k^a & \mapsto \hat{k}^a =  k^a \, , &
 \ell^a & \mapsto \hat{\ell}^a = \ell^a \, , &
 m^a & \mapsto \hat{m}^a = z m^a \, , &
 u^a & \mapsto \hat{u}^a = u^a \, .
\end{align}
where $z \in \U(1)$.
\end{itemize}

\paragraph{Transformation rules of the connection $1$-form}
Here, we give the transformation rules of the connection components under a boost \eqref{eq-boost} only, the other cases being somewhat more straightforward.
{ \small
\begin{align*}
\hat{\kappa} & = \kappa \, ,
&
\hat{\mdigamma} & = \mdigamma \, ,
&
\hat{\sigma} & = \sigma - z \kappa \, ,
&
\hat{\psi} & = \psi - z \mdigamma \, , 
&
\hat{\rho} & = \rho - \bar{z} \kappa \, , \\
\hat{\yogh} & = \yogh - r \mdigamma \, ,
&
\hat{\eta} & = \eta - r \kappa \, ,
&
\hat{\chi} & = \chi + r \kappa  - z \mdigamma \, ,
&
\hat{\phi} & = \phi - z \chi + r \sigma  - r z  \kappa - z \psi + z^2 \mdigamma \, ,
&
\hat{\epsilon} & = \epsilon + \bar{z} \kappa + \tfrac{1}{2} r \mdigamma \, ,
\end{align*}
\begin{align*}
 \hat{\alpha} & = \alpha - \bar{z} \epsilon + \bar{z} \rho - \bar{z}^2 \kappa + r \psi - r z \mdigamma \, ,
& \hat{\beta} & = \beta - z \epsilon + \bar{z} \sigma - z \bar{z} \kappa + r \psi - r z \mdigamma \, , \\
\bar{\hat{\pi}} & =   \bar{\pi}  + 2 z \bar{\epsilon} + z^2 \bar{\kappa}  - \tfrac{1}{2} r^2 \kappa - r \chi + r z \mdigamma  + D z
&
\hat{\tau} & = \tau + z \rho + \bar{z} \sigma + r \eta - ( z \bar{z} + \tfrac{1}{2} r^2 ) \kappa \, , \\
\hat{\mstigma} & = \mstigma + z \bar{\psi} + \bar{z} \psi + r \yogh - ( z \bar{z} + \tfrac{1}{2} r^2 ) \mdigamma \, ,
&
\hat{\mae} & = \mae - r ( \chi  - \eta ) - r^2 \kappa + r z \mdigamma - z \yogh \, , \\
\hat{\upsilon} & = \upsilon - \bar{z} \chi + r \rho  - r \bar{z} \kappa  - z \bar{\psi} + z \bar{z} \mdigamma \, ,
&
\hat{\theta}  & = \theta  - r \epsilon  + \bar{z} \eta - \bar{z} r \kappa + r \yogh - r^2 \mdigamma  \, ,
\end{align*}
\begin{align*}
\hat{\mqoppa} & =   \mqoppa  + r ( \epsilon + \bar{\epsilon} ) + z \bar{\chi}  + \bar{z}  \chi  + r ( z \bar{\kappa}  + \bar{z} \kappa ) + \tfrac{1}{2} r^2 \mdigamma - z \bar{z}  \mdigamma + D r
\end{align*}
\begin{align*}
\bar{\hat{\lambda}}  =  \bar{\lambda}  -  r \phi  -   \tfrac{1}{2} r^2 \sigma   +  2 z \bar{\alpha}  -  z \bar{\pi}   +  r z \chi  +  
 r z \psi  +  \tfrac{1}{2} r^2 z \kappa  -  2 z^2 \bar{\epsilon}  -  r z^2 \mdigamma  +  
 z^2 \bar{\rho}  - z^3  \bar{\kappa}  +  \delta z  -  z D z 
\end{align*}
\begin{align*}
\bar{\hat{\zeta}} & =  \bar{\zeta}    - r \mae -  r \bar{\pi}  +  r^2 \chi  -  \tfrac{1}{2} r^2 \eta  +   \tfrac{1}{2} r^3 \kappa    -  
 2 r z \bar{\epsilon}  -  r^2 z \mdigamma  +  2 z \bar{\theta}  +  r z \yogh  +  z^2 \bar{\eta}  -  
 z^2 \bar{\kappa}  +   \mathring{\delta} z   -  r D z 
\end{align*}
\begin{align*}
\bar{\hat{\mu}} & =   \bar{\mu}  -   \tfrac{1}{2} r^2 \rho   -  r \upsilon  +  2 z \bar{\beta}  +  r z \bar{\psi}  +  z^2 \bar{\sigma}  -  
 \bar{z} \bar{\pi}  +  r \bar{z} \chi  +  \tfrac{1}{2} r^2 \bar{z} \kappa  -  2 z \bar{z} \bar{\epsilon}  -  
 r z \bar{z} \mdigamma  -  z^2 \bar{z} \bar{\kappa}   +   \bar{\delta} z    -   \bar{z} D z 
\end{align*}
\begin{multline*}
\hat{\gamma}  = \gamma  + z \alpha + \bar{z} \beta  + r \theta  - ( z \bar{z} + \tfrac{1}{2}r^2 ) \epsilon  + \bar{z} \tau + \bar{z} z \rho + \bar{z}^2 \sigma + \bar{z} r \eta - \bar{z} ( z \bar{z} + \tfrac{1}{2}r^2 ) \kappa + r \mstigma + r z \bar{\psi} + r \bar{z} \psi \\ + r^2 \yogh
 - r ( z \bar{z} + \tfrac{1}{2} r^2 ) \mdigamma
\end{multline*}
\begin{multline*}
\hat{\xi}  =  \xi  + r  \bar{\alpha}  +  r \beta  +   \tfrac{1}{2} r^2 \psi   -  z \mqoppa  -  r z \epsilon  -  
 r z \bar{\epsilon}  -  \tfrac{1}{2} r^2 z \mdigamma  +  r z \bar{\rho}  +  z \bar{\upsilon}  -  z^2 \bar{\chi}  - 
  r z^2 \bar{\kappa}  +  \bar{z} \phi  +  r \bar{z} \sigma  -  z \bar{z} \chi  \\ -  z \bar{z} \psi  -  
 r z \bar{z} \kappa  +  z^2 \bar{z} \mdigamma  +  \delta r  -  z D r  -  ( z \bar{z}  +  \tfrac{1}{2} r^2 )  D r
\end{multline*}
\begin{multline*}
\hat{\omega}  = \omega  +  r \mae  -   \tfrac{1}{2} r^2 \chi   +  r^2 \eta  -  ( \tfrac{1}{2} r^3 +  
 r z \bar{z} ) \kappa   +  r \tau  +  
 r z \rho  -  z \mstigma  +  z \upsilon  -  r z \yogh  -   z^2 \bar{\psi} +  
 ( z \bar{z} + \tfrac{1}{2} r ^2 ) \mdigamma \\  +  \bar{z} \phi  +  r \bar{z} \sigma  -  z \bar{z} \chi  -  z \bar{z} \psi  
\end{multline*}
\begin{multline*}
\hat{\wynn}  =  \wynn    - r \mqoppa  + ( r z \bar{z} - \tfrac{1}{2} 
 r^3 ) \mdigamma  -  r^2 ( \epsilon  + \bar{\epsilon} ) +  r ( \theta  +  \bar{\theta} ) + ( \tfrac{1}{2} r^2 -  z \bar{z} )\yogh   +  z  \bar{\mae} +  \bar{z} \mae   \\ -  
 r z \bar{\chi} -  r \bar{z} \chi  +  r z \bar{\eta}  +  r \bar{z} \eta  -  
 r^2 \bar{z} \kappa -  r^2 z \bar{\kappa}     +  \mathring{\delta} r  -   r D r 
\end{multline*}
\begin{multline*}
\bar{\hat{\nu}}  =  \bar{\nu}  -  r \omega  -  r^2 \mae  -   ( \tfrac{1}{2} r^2 + z \bar{z} ) \bar{\pi}    +  ( \tfrac{1}{2} r^3 +  r z \bar{z} ) \chi    -   \tfrac{1}{2} r^3 \eta  + (  \tfrac{1}{4}
 r^4 +  \tfrac{1}{2} r^2 z \bar{z} ) \kappa  -   \tfrac{1}{2} r^2 \tau    +  2 z \bar{\gamma}  +  z \bar{\mu}  - ( r^2 z +  2 z^2 \bar{z}  ) \bar{\epsilon}  \\ - ( 
 \tfrac{1}{2} r^3 z +  
 r z^2 \bar{z} ) \mdigamma -  \tfrac{1}{2} r^2 z \rho  +  r z \mstigma  +  2 r z \bar{\theta}  -  
 r z \upsilon  +  r^2 z \yogh  +  2 z^2 \bar{\beta}  +  r z^2 \bar{\eta}  +  r z^2 \bar{\psi}  - ( 
 \tfrac{1}{2} r^2 z^2 +  z^3 \bar{z} ) \bar{\kappa}  +  z^2 \bar{\tau}  +  z^3 \bar{\sigma}  \\ +  \bar{z} \bar{\lambda}  -  r \bar{z} \phi  -  
 \tfrac{1}{2} r^2 \bar{z} \sigma  +  2 z \bar{z} \bar{\alpha}  +  
 r z \bar{z} \psi   +  z^2 \bar{z} \bar{\rho}   +  r \bar{\zeta}  +  \tilde{D} z  +  z \bar{\delta} z  +  \bar{z} \delta z  -   (  z \bar{z}  +  \tfrac{1}{2} r^2  )  D z  +  r \mathring{\delta} z
\end{multline*}
\begin{multline*}
\hat{\msampi} =  \msampi  + r  ( \gamma  +  \bar{\gamma} ) - (  \tfrac{1}{2} r^2 +  z \bar{z} ) \mqoppa    -  ( \tfrac{1}{2} r^3 +  r z \bar{z} ) ( \epsilon + \bar{\epsilon} ) -  ( \tfrac{1}{4} r^4 -  z^2 \bar{z}^2 )\mdigamma  + ( \tfrac{1}{2} r^2 -  z \bar{z} )\mstigma  +  
 r^2 ( \theta  +  \bar{\theta} ) -  r \wynn  +  ( \tfrac{1}{2} r^3 -  r z \bar{z} ) \yogh \\ +  z \bar{\omega} +  \bar{z} \omega +  r z \bar{\mae}  +  r \bar{z} \mae  + 
  r z \alpha +  r \bar{z} \bar{\alpha} +  r z \bar{\beta} +  r \bar{z} \beta -  
 ( \tfrac{1}{2} r^2 \bar{z} +  z \bar{z}^2 ) \chi  - ( \tfrac{1}{2} r^2 z +  z^2 \bar{z} ) \bar{\chi}  +  r^2 z \bar{\eta} +  r^2 \bar{z} \eta   \\  +  (
 \tfrac{1}{2} r^2 z -  z^2 \bar{z} ) \bar{\psi} + ( \tfrac{1}{2} r^2 \bar{z} -  z \bar{z}^2 ) \psi    - ( \tfrac{1}{2} r^3 z +  
 r z^2 \bar{z} ) \bar{\kappa}  - ( \tfrac{1}{2} r^3 \bar{z} + 
  r z \bar{z}^2 ) \kappa  +  r z \bar{\tau}  +  
 r \bar{z} \tau  +  z \bar{\xi} +  \bar{z} \xi  +  z^2 \bar{\phi} +  \bar{z}^2 \phi   +  
 r z^2 \bar{\sigma}  +  r \bar{z}^2 \sigma \\       +  
 r z \bar{z} \rho  +  r z \bar{z} \bar{\rho}  +  z \bar{z} \upsilon +  
 z \bar{z} \bar{\upsilon}    +  \tilde{D} r  +   z \bar{\delta} r  +  \bar{z} \delta r  +  r \mathring{\delta} r 
\end{multline*}
}

\paragraph{Transformation rules of the Weyl tensor}
To make contact with the CMPP classification, we shall state the transformations rule of the Weyl tensor under boosts in terms of boost weights, as given in Table \ref{table-Weylboostweight}. We recall that a scalar quantity $f$ is said to be of boost weight $w$, if, under a boost \eqref{eq-boost}, it transforms according to the rule $f \mapsto \hat{f} = b^w f$.
\begin{table}[!ht]
\begin{center}
{\renewcommand{\arraystretch}{1.1}
\renewcommand{\tabcolsep}{0.2cm}
\begin{tabular}{|l|c|}
\hline
Weyl tensor components & Boost weights \\
\hline
 $\Psi_0$, $\Psi^0_0$, $\Psi^0_1$ & 2 \\
 $\Psi_{0 3}$, $\Psi_1$, $\Psi^1_1$, $\Psi^2_1$ & 1 \\
 $\Psi^0_{1 3}$, $\Psi^1_{1 3}$, $\Psi^2_{1 3}$, $\Psi_2$, $\Psi^0_2$ & 0 \\
 $\Psi_{1 4}$, $\Psi_3$, $\Psi^1_3$,  $\Psi^2_3$ & -1 \\
 $\Psi_4$, $\Psi^0_4$, $\Psi^0_3$ & -2 \\
\hline
\end{tabular}}
\end{center}
\caption{Boost weights of the components of the Weyl tensor}\label{table-Weylboostweight}
\end{table}

We omit the transformation rules for unitary rotations, which are pretty straightforward. On the other hand, under a null rotation \eqref{eq-Nrot}, the Weyl tensor transforms according to
{ \small
\begin{align*}
\hat{\Psi}_0 & =  \Psi_0 \, , & \hat{\Psi}_1 & =  \Psi_1   +   z \Psi^0_1     +   \bar{z} \Psi_0    +   r \Psi^0_0 \, , \\
\hat{\Psi}^0_0 & =  \Psi^0_0 \, , & \hat{\Psi}^0_1 & =  \Psi^0_1 \, , \\
\hat{\Psi}_{0 3} & =  \Psi_{ 0 3 }   +   r \Psi_0    -   z \Psi^0_0 \, , & \hat{\Psi}^1_1 & =  \Psi^1_1    +   r \Psi^0_0   +   2 z \Psi^0_1 \, ,& \hat{\Psi}^2_1 & =  \Psi^2_1   +   r \Psi^0_1    -   z \bar{\Psi}^0_0 \, ,
\end{align*}
\begin{align*}
\hat{\Psi}^1_{ 1 3 } & =    -  \tfrac{1}{2} r^2 \Psi^0_0     +  z^2 \bar{\Psi}^0_0  +  \Psi^1_{ 1 3 }  -  r \Psi^1_1  -  
 2 r z \Psi^0_1  -  2 z \Psi^2_1 \\
\hat{\Psi}^0_{ 1 3 }  & =  r z \Psi^0_0  +  z \Psi^1_1  -   \tfrac{1}{2}  r^2 \Psi_0    +  z^2 \Psi^0_1  +  \Psi^0_{ 1 3 }  -  
 r \Psi_{ 0 3 } \\
\hat{\Psi}^2_{ 1 3 } & =  r^2 \bar{\Psi}^0_0  -   ( \tfrac{1}{2} r^2   +  z \bar{z} )  \bar{\Psi}^0_0  +  r \bar{z} \Psi^0_1  +  r \bar{\Psi}_1  +  
 r z \bar{\Psi}_0  +  \Psi^2_{ 1 3 }  +  \bar{z} \Psi^2_1  +  z \bar{\Psi}_{ 0 3 } \, , \\
\hat{\Psi}_{1 4} & =    - \bar{z}^3 \Psi^0_0  +  \tfrac{3}{2} r^2 \bar{z} \bar{\Psi}^0_0  -  3 \bar{z} \bar{\Psi}^1_{ 1 3 }  +  3 r \bar{z} \bar{\Psi}^1_1  +  
 3 r \bar{z}^2 \Psi^0_1  -   \tfrac{1}{2} r^3 \bar{\Psi}_0    +  3 r \bar{\Psi}^0_{ 1 3 }  +  
 3 \bar{z}^2 \bar{\Psi}^2_1  -   \tfrac{3}{2} r^2 \bar{\Psi}_{ 0 3 }  +  \Psi_{ 1 4 }
\end{align*}
\begin{align*}
\hat{\Psi}_2 & = r \bar{z} \Psi^0_0  -  \bar{z} \Psi^1_1  +  2 \bar{z} \Psi_1  +  \bar{z}^2 \Psi_0  -  \tfrac{1}{2}
 r^2 \Psi^0_1   +  \Psi_2  -  r \bar{\Psi}^2_1 \\
\hat{\Psi}^0_2 & =  r \bar{z} \Psi^0_0  +  r z \bar{\Psi}^0_0  +  \Psi^0_2  +  \bar{z} \Psi^1_1  +  
 z \bar{\Psi}^1_1  +   (  - r^2  +  2 z \bar{z} )  \Psi^0_1  -  r \Psi^2_1  -  r \bar{\Psi}^2_1
\end{align*}
\begin{multline*}
\hat{\Psi}_3 =    - 2 r \bar{z}^2 \Psi^0_0  +   ( \tfrac{1}{2} r^3  -  r z \bar{z} )  \bar{\Psi}^0_0  -  \bar{z} \Psi^0_2  -  
 2 r \bar{\Psi}^1_{ 1 3 }  +  \bar{z}^2 \Psi^1_1  +   ( r^2  -  z \bar{z} )  \bar{\Psi}^1_1  -  3 \bar{z}^2 \Psi_1  -  
 \bar{z}^3 \Psi_0 \\ +   (  \tfrac{5}{2} r^2 \bar{z}   -  z \bar{z}^2 )  \Psi^0_1  -  3 \bar{z} \Psi_2  +   \tfrac{1}{2}
 r^2 \bar{\Psi}_1    +  \tfrac{1}{2} r^2 z \bar{\Psi}_0  -  z \bar{\Psi}^0_{ 1 3 }  +  
 r \Psi^2_{ 1 3 }  +  \Psi_3  +  r \bar{z} \Psi^2_1  +  4 r \bar{z} \bar{\Psi}^2_1  +  r z \bar{\Psi}_{ 0 3 }
\end{multline*}
\begin{multline*}
\hat{\Psi}^0_3 =   ( -   \tfrac{1}{2} r^3 \bar{z}    +  r z \bar{z}^2 )  \Psi^0_0  +   (  -  \tfrac{1}{2} r^3 z    +  
    r z^2 \bar{z} )  \bar{\Psi}^0_0  +   (  - r^2  +  2 z \bar{z} )  \Psi^0_2  +  r \bar{z} \Psi^1_{ 1 3 }  +  
 r z \bar{\Psi}^1_{ 1 3 } \\ +   (  -  \tfrac{1}{2} r^2 \bar{z}     +  z \bar{z}^2 )  \Psi^1_1  +   (  - \tfrac{1}{2} r^2 z    +  
    z^2 \bar{z} )  \bar{\Psi}^1_1  -  r^2 \bar{z} \Psi_1  -  
 \tfrac{1}{2} r^2 \bar{z}^2 \Psi_0  +   ( \tfrac{1}{4} r^4  -  2 r^2 z \bar{z}  +  z^2 \bar{z}^2 )  \Psi^0_1  +  
 \bar{z}^2 \Psi^0_{ 1 3 } \\ -   \tfrac{1}{2} r^2 \Psi_2   -  r^2 z \bar{\Psi}_1  -  
 \tfrac{1}{2} r^2 z^2 \bar{\Psi}_0  -   \tfrac{1}{2} r^2 \bar{\Psi}_2   +  z^2 \bar{\Psi}^0_{ 1 3 }  -  
 2 r \bar{z} \bar{\Psi}^2_{ 1 3 }  -  2 r z \Psi^2_{ 1 3 }  -  \bar{z} \bar{\Psi}^1_3  -  
 z \Psi^1_3  +  \Psi^0_3  -  r \bar{z}^2 \Psi_{ 0 3 } \\ +   ( \tfrac{1}{2} r^3   -  2 r z \bar{z} )  \Psi^2_1  - 
  r \bar{\Psi}^2_3  +   ( \tfrac{1}{2} r^3  -  2 r z \bar{z} )  \bar{\Psi}^2_1  -  r z^2 \bar{\Psi}_{ 0 3 }  -  
 r \Psi^2_3
\end{multline*}
\begin{multline*}
\hat{\Psi}^1_3 =   - r \bar{z}^2 \Psi^0_0  +   ( \tfrac{1}{2} r^3 -  2 r z \bar{z} )  \bar{\Psi}^0_0  -  2 \bar{z} \Psi^0_2  -  
 r \bar{\Psi}^1_{ 1 3 }  -  
 \bar{z}^2 \Psi^1_1  +   ( \tfrac{1}{2} r^2  -  2 z \bar{z} )  \bar{\Psi}^1_1  +   ( 2 r^2 \bar{z}  -  
    2 z \bar{z}^2 )  \Psi^0_1  +  r^2 \bar{\Psi}_1 \\ +  r^2 z \bar{\Psi}_0  -  
 2 z \bar{\Psi}^0_{ 1 3 }  +  2 r \Psi^2_{ 1 3 }  +  \Psi^1_3  +  2 r \bar{z} \Psi^2_1  +  
 2 r \bar{z} \bar{\Psi}^2_1  +  2 r z \bar{\Psi}_{ 0 3 }
\end{multline*}
\begin{multline*}
\hat{\Psi}^2_3 = \tfrac{1}{2} r^2 \bar{z} \Psi^0_0  +   ( r^2 z  -  z^2 \bar{z} )  \bar{\Psi}^0_0  +  r \Psi^0_2  -  
 \bar{z} \Psi^1_{ 1 3 }  +  r \bar{z} \Psi^1_1  +   (  -  \tfrac{1}{2} r^3   +  2 r z \bar{z} )  \Psi^0_1  +  
 2 r z \bar{\Psi}_1  +  r z^2 \bar{\Psi}_0  +  r \bar{\Psi}_2 \\ +  
 2 z \Psi^2_{ 1 3 }  +   (  - r^2  +  2 z \bar{z} )  \Psi^2_1  -   \tfrac{1}{2} r^2 \bar{\Psi}^2_1    +  
 z^2 \bar{\Psi}_{ 0 3 }  +  \Psi^2_3
\end{multline*}
\begin{multline*}
\hat{\Psi}_4 =  2 r \bar{z}^3 \Psi^0_0  -  r^3 \bar{z} \bar{\Psi}^0_0  +  6 r \bar{z} \bar{\Psi}^1_{ 1 3 }  -  2 \bar{z}^3 \Psi^1_1  -  
 3 r^2 \bar{z} \bar{\Psi}^1_1  +  4 \bar{z}^3 \Psi_1  +  \bar{z}^4 \Psi_0  -  
 3 r^2 \bar{z}^2 \Psi^0_1  +  6 \bar{z}^2 \Psi_2  +   \tfrac{1}{4} r^4 \bar{\Psi}_0   \\  -  
 3 r^2 \bar{\Psi}^0_{ 1 3 }  +  2 \bar{z} \Psi^1_3  -  4 \bar{z} \Psi_3  +  \Psi_4  -  
 6 r \bar{z}^2 \bar{\Psi}^2_1  +  r^3 \bar{\Psi}_{ 0 3 }  -  2 r \Psi_{ 1 4 }
\end{multline*}
\begin{multline*}
\hat{\Psi}^0_4 =   (  \tfrac{3}{2} r^2 \bar{z}^2  -  z \bar{z}^3 )  \Psi^0_0  +   (  -  \tfrac{1}{4} r^4    +  \tfrac{3}{2} r^2 z \bar{z}  )  \bar{\Psi}^0_0  +  
 3 r \bar{z} \Psi^0_2  +   3 (  \tfrac{1}{2} r^2   -  z \bar{z} )  \bar{\Psi}^1_{ 1 3 }  +   (  -  \tfrac{1}{2} r^3    +  
    3 r z \bar{z} )  \bar{\Psi}^1_1 \\  +  3 r \bar{z}^2 \Psi_1  +  
 r \bar{z}^3 \Psi_0  +   (  -  \tfrac{3}{2} r^3 \bar{z}    +  r z \bar{z}  +  2 r z \bar{z}^2 )  \Psi^0_1  +  
 3 r \bar{z} \Psi_2  -   \tfrac{1}{2} r^3 \bar{\Psi}_1   -  \tfrac{1}{2} r^3 z \bar{\Psi}_0  +  
 3 r z \bar{\Psi}^0_{ 1 3 }  +  3 \bar{z}^2 \bar{\Psi}^2_{ 1 3 } \\  -   \tfrac{3}{2} r^2 \Psi^2_{ 1 3 }   +  \Psi^0_4  -  
 r \Psi^1_3  -  r \Psi_3  +  \bar{z}^3 \Psi_{ 0 3 }  -  \tfrac{3}{2} r^2 \bar{z} \Psi^2_1  +  
 3 \bar{z} \bar{\Psi}^2_3  +   (  - 3 r^2 \bar{z}  +  3 z \bar{z}^2 )  \bar{\Psi}^2_1  -  \tfrac{3}{2} r^2 z \bar{\Psi}_{ 0 3 }  + 
  z \Psi_{ 1 4 }
\end{multline*}
}

\section{The Ricci identity} \label{App-Ricci}
The Ricci identity are given by the well-known formula
\begin{align*}
 2 \partial \ind{_{\lb{c}}} \Gamma \ind{_{\rb{d} a b}} & = 2 \Gamma \ind{_{\lb{c}| a} ^e} \Gamma \ind{_{|\rb{d} e b}} + 2  \Gamma \ind{_{[cd]} ^e} \Gamma \ind{_{e a b}} + \Riem \ind{_{c d a b}} \, .
 \end{align*}
Taking components of the above formula and combining some of these yield a set of 54 equations:
{ \small
\begin{multline}\label{eq-R-1}
  - \tilde{D} \epsilon + D \gamma =  - 2 \epsilon \gamma  -  \bar{\epsilon} \gamma  -  \epsilon \bar{\gamma}  -  \kappa \nu +   
 \tfrac{1}{2} \bar{\chi} \omega  -  \tfrac{1}{2}  \chi \bar{\omega} + \beta \pi + \alpha \bar{\pi}  -  \tfrac{1}{2}  
 \mdigamma \msampi +  \tfrac{1}{2} \mqoppa \mstigma  -  \alpha \tau + \pi \tau  -  
 \beta \bar{\tau} + \mqoppa \theta  -  
 \mstigma \theta  + \tfrac{1}{2} \Psi^0_2 + \Psi_2  +  \Rho _{ 1  \tilde{1} } 
\end{multline}
\begin{align}\label{eq-R-2}
   - \tilde{D} \kappa + D \tau =  - 3 \gamma \kappa  -  \bar{\gamma} \kappa  -  
  \mdigamma \omega + \eta \mqoppa + \bar{\pi} \rho + \pi \sigma + \chi \mstigma  -  
  \eta \mstigma + \epsilon \tau  -  \bar{\epsilon} \tau  -  \rho \tau  -  \sigma \bar{\tau} + \Psi_1  +   \Rho _{ 1  2 } 
\end{align}
\begin{align}\label{eq-R-3}
  D \nu  -  \tilde{D} \pi =  - 3 \epsilon \nu  -  \bar{\epsilon} \nu + \gamma \pi  -  
  \bar{\gamma} \pi + \mu \pi + \lambda \bar{\pi}  -  \bar{\omega} \mqoppa + \bar{\chi} \msampi  -  
  \lambda \tau  -  \mu \bar{\tau} + \mqoppa \zeta  -  \mstigma \zeta   - \Psi_3  -   \Rho _{ \tilde{1}  \bar{2} } 
\end{align}
\begin{align}\label{eq-R-4}
   - \tilde{D} \mdigamma + D \mstigma =  - 2 \mdigamma \gamma  -  2 \mdigamma \bar{\gamma} + 
  \bar{\kappa} \omega + \kappa \bar{\omega} + \pi \psi + \bar{\pi} \bar{\psi}  -  \bar{\chi} \tau  -  
  \bar{\psi} \tau  -  \chi \bar{\tau}  -  \psi \bar{\tau} + \mqoppa \yogh  -  \mstigma \yogh  - \Psi^2_1 - \bar{\Psi}^2_1  +  \Rho _{ 1  0 } 
\end{align}
\begin{align}\label{eq-R-5}
   - \tilde{D} \mqoppa + D \msampi =  - \chi \nu  -  \bar{\chi} \bar{\nu} + \omega \pi + 
  \bar{\omega} \bar{\pi}  -  2 \epsilon \msampi  -  2 \bar{\epsilon} \msampi + \mqoppa \wynn  -  
  \mstigma \wynn + \pi \xi  -  \bar{\tau} \xi + \bar{\pi} \bar{\xi}  -  \tau \bar{\xi}  + \bar{\Psi}^2_3 + \Psi^2_3  
 -   \Rho _{ \tilde{1}  0 } 
\end{align}
\begin{align}\label{eq-R-6}
   - \tilde{D} \chi + D \omega =  - 2 \chi \gamma + \mdigamma \bar{\nu}  -  
  2 \bar{\epsilon} \omega + \phi \pi + \mae \mqoppa  -  \kappa \msampi  -  \mae \mstigma  -  
  \bar{\pi} \mstigma + \mqoppa \tau  -  \phi \bar{\tau} + \bar{\pi} \upsilon  -  \tau \upsilon +  \bar{\Psi}^2_{ 1 3 } - \Psi^1_{ 1 3 }  
\end{align}
\begin{multline}\label{eq-R-7}
 D \beta  -  \delta \epsilon =  - \bar{\alpha} \epsilon  -  \beta \bar{\epsilon}  -  \gamma \kappa  -  \kappa \mu + \tfrac{1}{2}  
 \bar{\chi} \phi  -  \epsilon \bar{\pi} +  \tfrac{1}{2} \psi \mqoppa  -  \beta \bar{\rho}  -  \alpha \sigma +
  \pi \sigma + \chi \theta  -  \psi \theta  -  \tfrac{1}{2}  \chi \bar{\upsilon}  -  \tfrac{1}{2}  
 \mdigamma \xi  - \tfrac{1}{2} \Psi^1_1 + \Psi_1 
\end{multline}
\begin{align}\label{eq-R-8}
   - \delta \kappa + D \sigma = 
 \chi \eta  -  \bar{\alpha} \kappa  -  3 \beta \kappa  -  \mdigamma \phi  -  \kappa \bar{\pi} + 
  \chi \psi  -  \eta \psi + 3 \epsilon \sigma  -  \bar{\epsilon} \sigma  -  \rho \sigma  -  
  \bar{\rho} \sigma  -  \kappa \tau  + \Psi_0 
\end{align}
\begin{align}\label{eq-R-9}
  D \mu  -  \delta \pi =  - \epsilon \mu  -  \bar{\epsilon} \mu  -  \kappa \nu  -  \bar{\alpha} \pi + 
  \beta \pi  -  \pi \bar{\pi}  -  \mu \bar{\rho}  -  \lambda \sigma  -  \mqoppa \bar{\upsilon} + 
  \bar{\chi} \xi + \chi \zeta  -  \psi \zeta  + \Psi_2 -  \Rho _{ 2  \bar{2} } -   \Rho _{ 1  \tilde{1} }
\end{align}
\begin{align}\label{eq-R-10}
   - \delta \mdigamma + D \psi =  - 2 \bar{\alpha} \mdigamma  -  2 \beta \mdigamma + 
  \bar{\kappa} \phi  -  \mdigamma \bar{\pi} + 2 \epsilon \psi  -  \chi \bar{\rho}  -  \psi \bar{\rho}  -  
  \bar{\chi} \sigma  -  \bar{\psi} \sigma  -  \kappa \mstigma + \kappa \bar{\upsilon} + \chi \yogh  - 
   \psi \yogh  + \Psi^0_0 
\end{align}
\begin{align}\label{eq-R-11}
   - \delta \mqoppa + D \xi =  - \bar{\chi} \bar{\lambda}  -  \chi \mu + \phi \pi  -  \bar{\pi} \mqoppa  -  
  \kappa \msampi + \bar{\pi} \bar{\upsilon} + \chi \wynn  -  \psi \wynn  -  2 \bar{\epsilon} \xi  -  
  \bar{\rho} \xi  -  \sigma \bar{\xi}  + \bar{\Psi}^2_{ 1 3 } -   \Rho _{ 2  0 } 
\end{align}
\begin{align}\label{eq-R-12}
   - \delta \chi + D \phi = 
 \mae \chi  -  2 \beta \chi + \mdigamma \bar{\lambda}  -  \kappa \omega + 
  2 \epsilon \phi  -  2 \bar{\epsilon} \phi  -  \chi \bar{\pi}  -  \mae \psi  -  \bar{\pi} \psi  -  
  \phi \bar{\rho} + \mqoppa \sigma  -  \sigma \upsilon  -  \kappa \xi  + \Psi_{ 0 3 } 
\end{align}
\begin{multline}\label{eq-R-13}
 D \alpha  -  \bar{\delta} \epsilon = \alpha \bar{\epsilon}  -  2 \alpha \epsilon  -  \bar{\beta} \epsilon  -  \gamma \bar{\kappa}  -  
 \kappa \lambda  -  \tfrac{1}{2}  \chi \bar{\phi}  -  \epsilon \pi +  \tfrac{1}{2} \bar{\psi} \mqoppa  -  
 \alpha \rho + \pi \rho  -  \beta \bar{\sigma} + \bar{\chi} \theta  -  \bar{\psi} \theta +  \tfrac{1}{2} 
 \bar{\chi} \upsilon  -  \tfrac{1}{2}  \mdigamma \bar{\xi}  - \tfrac{1}{2} \bar{\Psi}^1_1 -  \Rho _{ 1  \bar{2} } 
\end{multline}
\begin{align}\label{eq-R-14}
   - \bar{\delta} \kappa + D \rho = 
 \bar{\chi} \eta  -  3 \alpha \kappa  -  \bar{\beta} \kappa  -  \kappa \pi + \chi \bar{\psi}  -  
  \eta \bar{\psi} + \bar{\epsilon} \rho + \epsilon \rho  -  \rho^2  -  \bar{\sigma} \sigma  -  
  \bar{\kappa} \tau  -  \mdigamma \upsilon + \Psi^0_1 -  \Rho _{ 1  1 }
\end{align}
\begin{align}\label{eq-R-15}
  D \lambda  -  \bar{\delta} \pi = 
 \bar{\epsilon} \lambda  -  3 \epsilon \lambda  -  \bar{\kappa} \nu + \alpha \pi  -  
  \bar{\beta} \pi  -  \pi^2  -  \bar{\phi} \mqoppa  -  \lambda \rho  -  \mu \bar{\sigma} + \bar{\chi} \bar{\xi} + 
  \bar{\chi} \zeta  -  \bar{\psi} \zeta + \bar{\Psi}^0_{ 1 3 } -   \Rho _{ \bar{2}  \bar{2} } 
\end{align}
\begin{align}\label{eq-R-16}
   - \bar{\delta} \chi + D \upsilon = 
 \mae \bar{\chi}  -  2 \alpha \chi + \mdigamma \bar{\mu}  -  \bar{\kappa} \omega  -  \chi \pi  -  
  \mae \bar{\psi}  -  \bar{\pi} \bar{\psi} + \mqoppa \rho  -  \phi \bar{\sigma}  -  \rho \upsilon  -  
  \kappa \bar{\xi} + \Psi^2_1 +  \Rho _{ 1  0 } 
\end{align}
\begin{multline}\label{eq-R-17}
 \tilde{D} \beta  -  \delta \gamma = \bar{\alpha} \gamma + 2 \beta \gamma  -  \beta \bar{\gamma}  -  \alpha \bar{\lambda}  -  \beta \mu  -  
 \epsilon \bar{\nu} + \tfrac{1}{2}  \bar{\omega} \phi +  \tfrac{1}{2} \psi \msampi + \nu \sigma  -  \gamma \tau  - 
  \mu \tau + \omega \theta  -  \tfrac{1}{2}  \omega \bar{\upsilon}  -  \tfrac{1}{2}  \mstigma \xi  -  
 \theta \xi - \tfrac{1}{2} \bar{\Psi}^1_3  -  \Rho _{ \tilde{1}  2 }
\end{multline}
\begin{align}\label{eq-R-18}
  \tilde{D} \sigma  -  \delta \tau =  - \kappa \bar{\nu} + \eta \omega + \omega \psi  -  
  \bar{\lambda} \rho + 3 \gamma \sigma  -  \bar{\gamma} \sigma  -  \mu \sigma  -  \phi \mstigma +
   \bar{\alpha} \tau  -  \beta \tau  -  \tau^2  -  \eta \xi  + \Psi^0_{ 1 3 } -  \Rho _{ 2  2 } 
\end{align}
\begin{align}\label{eq-R-19}
  \tilde{D} \mu  -  \delta \nu =  - \lambda \bar{\lambda}  -  \gamma \mu  -  \bar{\gamma} \mu  -  \mu^2 + 
  \bar{\alpha} \nu + 3 \beta \nu  -  \bar{\nu} \pi  -  \nu \tau  -  \msampi \bar{\upsilon} + 
  \bar{\omega} \xi + \omega \zeta  -  \xi \zeta  + \Psi^0_3 -  \Rho _{ \tilde{1}  \tilde{1} }
\end{align}
\begin{align}\label{eq-R-20}
  \tilde{D} \psi  -  \delta \mstigma =  - \mdigamma \bar{\nu} + 2 \gamma \psi  -  \mu \psi  -  
  \bar{\lambda} \bar{\psi}  -  \omega \bar{\rho}  -  \bar{\omega} \sigma  -  \mstigma \tau + \phi \bar{\tau} + 
  \tau \bar{\upsilon} + \omega \yogh  -  \xi \yogh  + \Psi^1_{ 1 3 } -   \Rho _{ 2  0 } 
\end{align}
\begin{align}\label{eq-R-21}
   - \delta \msampi + \tilde{D} \xi =  - \mu \omega  -  \bar{\lambda} \bar{\omega} + \nu \phi  -  
  \bar{\nu} \mqoppa + 2 \bar{\alpha} \msampi + 2 \beta \msampi  -  \msampi \tau + 
  \bar{\nu} \bar{\upsilon} + \omega \wynn  -  2 \bar{\gamma} \xi  -  \mu \xi  -  \wynn \xi  -  
  \bar{\lambda} \bar{\xi}  + \bar{\Psi}^0_4 
\end{align}
\begin{align}\label{eq-R-22}
   - \delta \omega + \tilde{D} \phi =  - \chi \bar{\nu} + \mae \omega + 2 \bar{\alpha} \omega + 
  2 \gamma \phi  -  2 \bar{\gamma} \phi  -  \mu \phi  -  \bar{\nu} \psi + \msampi \sigma + 
  \bar{\lambda} \mstigma  -  \omega \tau  -  \bar{\lambda} \upsilon  -  \mae \xi  -  \tau \xi  + \bar{\Psi}_{ 1 4 } 
\end{align}
\begin{multline}\label{eq-R-23}
 \tilde{D} \alpha  -  \bar{\delta} \gamma = \alpha \bar{\gamma} + \bar{\beta} \gamma  -  \beta \lambda  -  \alpha \bar{\mu}  -  
 \epsilon \nu  -  \tfrac{1}{2}  \omega \bar{\phi} + \nu \rho +  \tfrac{1}{2} \bar{\psi} \msampi  -  
 \gamma \bar{\tau}  -  \lambda \tau + \bar{\omega} \theta +  \tfrac{1}{2} \bar{\omega} \upsilon  -   \tfrac{1}{2} 
 \mstigma \bar{\xi}  -  \theta \bar{\xi} + \Psi_3 - \tfrac{1}{2} \Psi^1_3 
\end{multline}
\begin{align}\label{eq-R-24}
  \tilde{D} \rho  -  \bar{\delta} \tau =  - \kappa \nu + \eta \bar{\omega} + \omega \bar{\psi} + 
  \bar{\gamma} \rho + \gamma \rho  -  \bar{\mu} \rho  -  \lambda \sigma  -  \alpha \tau + 
  \bar{\beta} \tau  -  \bar{\tau} \tau  -  \mstigma \upsilon  -  \eta \bar{\xi}  + \Psi_2 -  \Rho _{ 2  \bar{2} } -   \Rho _{ 1  \tilde{1} }
\end{align}
\begin{align}\label{eq-R-25}
  \tilde{D} \lambda  -  \bar{\delta} \nu = 
 \bar{\gamma} \lambda  -  3 \gamma \lambda  -  \lambda \bar{\mu}  -  \lambda \mu + 
  3 \alpha \nu + \bar{\beta} \nu  -  \nu \pi  -  \bar{\phi} \msampi  -  \nu \bar{\tau} + \bar{\omega} \bar{\xi} +
   \bar{\omega} \zeta  -  \bar{\xi} \zeta  + \Psi_4 
\end{align}
\begin{align}\label{eq-R-26}
   - \bar{\delta} \omega + \tilde{D} \upsilon =  - \chi \nu + \mae \bar{\omega} + 2 \bar{\beta} \omega  - 
   \lambda \phi  -  \bar{\nu} \bar{\psi} + \rho \msampi + \bar{\mu} \mstigma  -  \omega \bar{\tau}  -  
  \bar{\mu} \upsilon  -  \mae \bar{\xi}  -  \tau \bar{\xi} + \bar{\Psi}^2_3 +  \Rho _{ \tilde{1}  0 } 
\end{align}
\begin{multline}\label{eq-R-27}
 \delta \alpha  -  \bar{\delta} \beta = \alpha \bar{\alpha}  -  2 \alpha \beta + \beta \bar{\beta}  -  \epsilon \mu + 
 \epsilon \bar{\mu}  -   \tfrac{1}{2} \phi \bar{\phi} + \gamma \rho + \mu \rho  -  \gamma \bar{\rho}  -  
 \lambda \sigma  -  \theta \upsilon + \theta \bar{\upsilon} +   \tfrac{1}{2}
 \upsilon \bar{\upsilon} + \tfrac{1}{2}  \bar{\psi} \xi  -   \tfrac{1}{2} 
 \psi \bar{\xi} +  \Psi_2  - \tfrac{1}{2} \Psi^0_2 +  \Rho _{ 2  \bar{2} } 
\end{multline}
\begin{align}\label{eq-R-28}
  \delta \rho  -  \bar{\delta} \sigma =  - \kappa \mu + \kappa \bar{\mu} + \phi \bar{\psi} + 
  \bar{\alpha} \rho + \beta \rho  -  3 \alpha \sigma + \bar{\beta} \sigma + \rho \tau  -  
  \bar{\rho} \tau  -  \eta \upsilon  -  \psi \upsilon + \eta \bar{\upsilon}   - \Psi^1_1 + \Psi_1   -   \Rho _{ 1  2 }
\end{align}
\begin{align}\label{eq-R-29}
  \delta \lambda  -  \bar{\delta} \mu = 
 \bar{\alpha} \lambda  -  3 \beta \lambda + \alpha \mu + \bar{\beta} \mu  -  \mu \pi + 
  \bar{\mu} \pi + \nu \rho  -  \nu \bar{\rho}  -  \bar{\phi} \xi + \bar{\upsilon} \bar{\xi}  -  \upsilon \zeta + 
  \bar{\upsilon} \zeta  -  \Psi_3 + \Psi^1_3  +  \Rho _{ \tilde{1}  \bar{2} } 
\end{align}
\begin{align}\label{eq-R-30}
   - \bar{\delta} \psi + \delta \bar{\psi} =  - \mdigamma \mu + \mdigamma \bar{\mu}  -  2 \alpha \psi + 
  2 \bar{\alpha} \bar{\psi} + \bar{\phi} \sigma  -  \phi \bar{\sigma} + \rho \mstigma  -  \bar{\rho} \mstigma +
   \bar{\rho} \upsilon  -  \rho \bar{\upsilon}  -  \upsilon \yogh + \bar{\upsilon} \yogh + \Psi^2_1 - \bar{\Psi}^2_1 
\end{align}
\begin{align}\label{eq-R-31}
   - \bar{\delta} \xi + \delta \bar{\xi} =  - \lambda \phi + \bar{\lambda} \bar{\phi}  -  \mu \mqoppa + 
  \bar{\mu} \mqoppa + \rho \msampi  -  \bar{\rho} \msampi + \mu \upsilon  -  \bar{\mu} \bar{\upsilon}  -  
  \upsilon \wynn + \bar{\upsilon} \wynn + 2 \bar{\beta} \xi  -  2 \beta \bar{\xi}  -  \Psi^2_3 + \bar{\Psi}^2_3 
\end{align}
\begin{align}\label{eq-R-32}
   - \bar{\delta} \phi + \delta \upsilon =  - \chi \mu + \chi \bar{\mu}  -  2 \alpha \phi + 
  2 \bar{\beta} \phi + \bar{\mu} \psi  -  \bar{\lambda} \bar{\psi} + \omega \rho  -  \omega \bar{\rho}  -  
  \mae \upsilon + \mae \bar{\upsilon} + \rho \xi  -  \sigma \bar{\xi} +  \bar{\Psi}^2_{ 1 3 } + \Psi^1_{ 1 3 } + \Rho _{ 2  0 } 
\end{align}
\begin{multline}\label{eq-R-33}
  - \mathring{\delta} \epsilon + D \theta =  -   \tfrac{1}{2}  \bar{\mae} \chi  -  \alpha \chi +   \tfrac{1}{2} \mae \bar{\chi}  -  \beta \bar{\chi}  -  \alpha \eta  -  
 \beta \bar{\eta}  -  \mdigamma \gamma + \eta \pi  -  \epsilon \mqoppa  -  \epsilon \theta  -  
 \epsilon \bar{\theta}  -   \tfrac{1}{2} \mdigamma \wynn + \tfrac{1}{2}  \mqoppa \yogh  -  \theta \yogh  -  
 \kappa \zeta - \bar{\Psi}^2_1   + \tfrac{1}{2}  \Rho _{ 1  0 } 
\end{multline}
\begin{align}\label{eq-R-34}
  D \eta  -  \mathring{\delta} \kappa =  - \mae \mdigamma + 2 \epsilon \eta  -  \kappa \mqoppa  -  
  \chi \rho  -  \eta \rho  -  \bar{\chi} \sigma  -  \bar{\eta} \sigma  -  \mdigamma \tau  -  
  3 \kappa \theta  -  \kappa \bar{\theta} + \chi \yogh  -  \eta \yogh  + \Psi^0_0
\end{align}
\begin{align}\label{eq-R-35}
   - \mathring{\delta} \pi + D \zeta =  - \chi \lambda  -  \eta \lambda  -  \bar{\chi} \mu  -  \bar{\eta} \mu  -  
  \mdigamma \nu  -  \bar{\mae} \mqoppa  -  \pi \mqoppa + \pi \theta  -  \pi \bar{\theta} + 
  \bar{\chi} \wynn  -  2 \epsilon \zeta  -  \yogh \zeta + \bar{\Psi}^1_{ 1 3 } -   \Rho _{ \bar{2}  0 } 
\end{align}
\begin{multline}\label{eq-R-36}
   - \mathring{\delta} \mdigamma + D \yogh =  - \bar{\chi} \eta  -  \chi \bar{\eta} + \bar{\mae} \kappa + 
  \mae \bar{\kappa}  -  \bar{\chi} \psi  -  \bar{\eta} \psi  -  \chi \bar{\psi}  -  \eta \bar{\psi}  -  
  \mdigamma \mqoppa  -  \mdigamma \mstigma  -  2 \mdigamma \theta  -  
  2 \mdigamma \bar{\theta} + \epsilon \yogh + \bar{\epsilon} \yogh  -  \yogh^2 - 2 \Psi^0_1 - \Rho _{ 1  1 }
\end{multline}
\begin{align}\label{eq-R-37}
   - \mathring{\delta} \mqoppa + D \wynn = 
 \mae \pi + \bar{\mae} \bar{\pi}  -  \mqoppa^2  -  \mdigamma \msampi  -  \epsilon \wynn  -  
  \bar{\epsilon} \wynn  -  \bar{\chi} \xi  -  \bar{\eta} \xi  -  \chi \bar{\xi}  -  \eta \bar{\xi}  -  \wynn \yogh  -  
  \chi \zeta  -  \bar{\chi} \bar{\zeta}  + \Psi^0_2 -   \Rho _{ 0  0 } - \Rho _{ \tilde{1}  1 }
\end{align}
\begin{align}\label{eq-R-38}
  D \mae  -  \mathring{\delta} \chi = 
 \mae \epsilon  -  \mae \bar{\epsilon}  -  \mdigamma \omega  -  \bar{\chi} \phi  -  \bar{\eta} \phi  -  
  \chi \mqoppa + \eta \mqoppa  -  2 \chi \theta  -  \chi \upsilon  -  \eta \upsilon  -  
  \kappa \wynn  -  \mae \yogh  -  \bar{\pi} \yogh + \mdigamma \bar{\zeta}  + \Psi^1_1 - \Rho _{ 1  2 }
\end{align}
\begin{multline}\label{eq-R-39}
  - \mathring{\delta} \gamma + \tilde{D} \theta = \eta \nu  -  \tfrac{1}{2}  \bar{\mae} \omega  -  \alpha \omega +  \tfrac{1}{2} \mae \bar{\omega}  -  \beta \bar{\omega}  -  
 \epsilon \msampi  -  \gamma \mstigma + \gamma \theta + \gamma \bar{\theta}  -   \tfrac{1}{2} 
 \mstigma \wynn  -  \theta \wynn +  \tfrac{1}{2} \msampi \yogh  -  \beta \zeta  -  \tau \zeta  - 
  \alpha \bar{\zeta}  + \bar{\Psi}^2_3  -  \tfrac{1}{2} \Rho _{ \tilde{1}  0 } 
\end{multline}
\begin{align}\label{eq-R-40}
  \tilde{D} \eta  -  \mathring{\delta} \tau = 
 2 \eta \gamma  -  \omega \rho  -  \kappa \msampi  -  \bar{\omega} \sigma  -  \mae \mstigma  -  
  \mstigma \tau  -  \tau \theta + \tau \bar{\theta}  -  \eta \wynn + \omega \yogh  -  
  \sigma \zeta  -  \rho \bar{\zeta}  + \bar{\Psi}^2_{ 1 3 } -   \Rho _{ 2  0 } 
\end{align}
\begin{align}\label{eq-R-41}
   - \mathring{\delta} \nu + \tilde{D} \zeta =  - \lambda \omega  -  \mu \bar{\omega}  -  \bar{\mae} \msampi  -  
  \pi \msampi  -  \nu \mstigma + 3 \nu \theta + \nu \bar{\theta} + \bar{\omega} \wynn  -  
  2 \gamma \zeta  -  \mu \zeta  -  \wynn \zeta  -  \lambda \bar{\zeta}  + \Psi^0_4 
\end{align}
\begin{multline}\label{eq-R-42}
   - \mathring{\delta} \mstigma + \tilde{D} \yogh =  - \bar{\eta} \omega  -  \eta \bar{\omega}  -  \bar{\omega} \psi  -  
  \omega \bar{\psi}  -  \mdigamma \msampi  -  \mstigma^2 + \bar{\mae} \tau + \mae \bar{\tau} + 
  \gamma \yogh + \bar{\gamma} \yogh  -  \wynn \yogh  -  \psi \zeta  -  \bar{\psi} \bar{\zeta} + \Psi^0_2 -   \Rho _{ 0  0 } - \Rho _{ \tilde{1}  1 }
\end{multline}
\begin{multline}\label{eq-R-43}
   - \mathring{\delta} \msampi + \tilde{D} \wynn = 
 \mae \nu + \bar{\mae} \bar{\nu}  -  \mqoppa \msampi  -  \msampi \mstigma + 2 \msampi \theta + 
  2 \msampi \bar{\theta}  -  \gamma \wynn  -  \bar{\gamma} \wynn  -  \wynn^2  -  \bar{\omega} \xi  -  
  \omega \bar{\xi}  -  \omega \zeta  -  \xi \zeta  -  \bar{\omega} \bar{\zeta}  -  \bar{\xi} \bar{\zeta} - 2 \Psi^0_3  - \Rho _{ \tilde{1}  \tilde{1} }
\end{multline}
\begin{align}\label{eq-R-44}
  \tilde{D} \mae  -  \mathring{\delta} \omega = 
 \mae \gamma  -  \mae \bar{\gamma}  -  \bar{\omega} \phi  -  \chi \msampi + \eta \msampi  -  
  \omega \mstigma + 2 \omega \bar{\theta}  -  \omega \upsilon  -  \mae \wynn  -  \tau \wynn  - 
   \bar{\nu} \yogh  -  \phi \zeta + \mstigma \bar{\zeta}  -  \upsilon \bar{\zeta}  + \bar{\Psi}^1_3 - \Rho _{ \tilde{1}  2 }
\end{align}
\begin{multline}\label{eq-R-45}
  - \mathring{\delta} \beta + \delta \theta = \eta \gamma + \eta \mu  -   \tfrac{1}{2} \bar{\mae} \phi  -  \alpha \phi  -  \gamma \psi  -  \mae \theta  -  
 \beta \theta + \beta \bar{\theta} +  \tfrac{1}{2} \mae \bar{\upsilon}  -  \beta \bar{\upsilon}  -   \tfrac{1}{2} 
 \psi \wynn  -  \epsilon \xi +  \tfrac{1}{2} \xi \yogh  -  \sigma \zeta + 
 \epsilon \bar{\zeta} + \bar{\Psi}^2_{ 1 3 } +  \tfrac{1}{2} \Rho _{ 2  0 } 
\end{multline}
\begin{align}\label{eq-R-46}
  \delta \eta  -  \mathring{\delta} \sigma =  - \mae \eta + 2 \beta \eta  -  \mae \psi  -  \phi \rho + 
  \eta \tau  -  \psi \tau  -  3 \sigma \theta + \sigma \bar{\theta}  -  \sigma \bar{\upsilon}  -  
  \kappa \xi + \phi \yogh + \kappa \bar{\zeta}  + \Psi_{ 0 3 } 
\end{align}
\begin{align}\label{eq-R-47}
   - \mathring{\delta} \mu + \delta \zeta = 
 \eta \nu  -  \lambda \phi  -  \nu \psi + \mu \theta + \mu \bar{\theta}  -  \mu \bar{\upsilon} + 
  \bar{\upsilon} \wynn  -  \bar{\mae} \xi  -  \pi \xi  -  \mae \zeta  -  2 \beta \zeta + \pi \bar{\zeta} + \bar{\Psi}^2_3 +  \Rho _{ \tilde{1}  0 }
\end{align}
\begin{align}\label{eq-R-48}
   - \mathring{\delta} \psi + \delta \yogh =  - \bar{\eta} \phi  -  \phi \bar{\psi} + \mae \bar{\rho} + \bar{\mae} \sigma + 
  \eta \mstigma  -  \psi \mstigma  -  2 \psi \theta  -  \eta \bar{\upsilon}  -  
  \psi \bar{\upsilon}  -  \mdigamma \xi  -  \mae \yogh + \bar{\alpha} \yogh + \beta \yogh + 
  \mdigamma \bar{\zeta}  + \Psi^1_1 - \Rho _{ 1  2 }
\end{align}
\begin{align}\label{eq-R-49}
  \delta \wynn  -  \mathring{\delta} \xi = 
 \bar{\mae} \bar{\lambda} + \mae \mu + \eta \msampi  -  \psi \msampi  -  \mae \wynn  -  \bar{\alpha} \wynn  - 
   \beta \wynn  -  \mqoppa \xi + 2 \bar{\theta} \xi  -  \bar{\upsilon} \xi  -  \phi \bar{\xi}  -  
  \phi \zeta + \mqoppa \bar{\zeta}  -  \bar{\upsilon} \bar{\zeta}  + \bar{\Psi}^1_3 - \Rho _{ \tilde{1}  2 }
\end{align}
\begin{multline}\label{eq-R-50}
  \delta \mae  -  \mathring{\delta} \phi =  - \mae^2  -  \mae \bar{\alpha} + \mae \beta + \eta \omega  -  
  \omega \psi  -  2 \phi \theta + 2 \phi \bar{\theta}  -  \phi \upsilon  -  
  \phi \bar{\upsilon}  -  \sigma \wynn  -  \chi \xi + \eta \xi  -  \bar{\lambda} \yogh + 
  \chi \bar{\zeta} + \psi \bar{\zeta} - 2 \Psi^0_{ 1 3 } - \Rho _{ 2  2 }
\end{multline}
\begin{multline}\label{eq-R-51}
  - \mathring{\delta} \alpha + \bar{\delta} \theta = \bar{\eta} \gamma + \eta \lambda +  \tfrac{1}{2} \mae \bar{\phi}  -  \beta \bar{\phi}  -  \gamma \bar{\psi}  -  
 \alpha \bar{\theta}  -  \bar{\mae} \theta + \alpha \theta  -  \tfrac{1}{2}  \bar{\mae} \upsilon  -  
 \alpha \upsilon  -  \tfrac{1}{2}  \bar{\psi} \wynn  -  \epsilon \bar{\xi} +  \tfrac{1}{2} \bar{\xi} \yogh + 
 \epsilon \zeta  -  \rho \zeta - \bar{\Psi}^1_{ 1 3 }  - \tfrac{1}{2}  \Rho _{ \bar{2}  0 } 
\end{multline}
\begin{align}\label{eq-R-52}
  \bar{\delta} \eta  -  \mathring{\delta} \rho =  - \bar{\mae} \eta + 2 \alpha \eta  -  \mae \bar{\psi}  -  
  \bar{\phi} \sigma + \bar{\eta} \tau  -  \bar{\psi} \tau  -  \rho \bar{\theta}  -  \rho \theta  -  
  \rho \upsilon  -  \kappa \bar{\xi} + \upsilon \yogh + \kappa \zeta + \bar{\Psi}^2_1 +  \Rho _{ 1  0 } 
\end{align}
\begin{align}\label{eq-R-53}
   - \mathring{\delta} \lambda + \bar{\delta} \zeta = 
 \bar{\eta} \nu  -  \mu \bar{\phi}  -  \nu \bar{\psi}  -  \lambda \bar{\theta} + 3 \lambda \theta  -  
  \lambda \upsilon + \bar{\phi} \wynn  -  \bar{\mae} \bar{\xi}  -  \pi \bar{\xi}  -  \bar{\mae} \zeta  -  
  2 \alpha \zeta + \pi \zeta  + \Psi_{ 1 4 } 
\end{align}
\begin{align}\label{eq-R-54}
  \bar{\delta} \mae  -  \mathring{\delta} \upsilon =  - \bar{\mae} \mae + \mae \alpha  -  \mae \bar{\beta} + 
  \bar{\eta} \omega  -  \bar{\phi} \phi  -  \omega \bar{\psi}  -  \upsilon^2  -  \rho \wynn  -  
  \chi \bar{\xi} + \eta \bar{\xi}  -  \bar{\mu} \yogh + \bar{\psi} \bar{\zeta} + \chi \zeta   - 2 \Psi^0_{ 1 3 }  - \Rho _{ 2  2 }
\end{align}
}

\section{The Bianchi identity}\label{App-Bianchi}
Here, we give the components of the Bianchi identity according to the formula
\begin{align*}
\partial \ind{_{\lb{a}}} C \ind{_{b \rb{c} d e}}  & = - 2 g \ind{_{\lb{a} | \lb{d}}} A \ind{_{\rb{e} | b \rb{c}}} - 2 \Gamma \ind{_{\lb{a} b} ^f} C \ind{_{\rb{c} f d e}} - 2 \Gamma \ind{_{\lb{a}| \lb{d}} ^f} C \ind{_{\rb{e}|f| b \rb{c}}} \, ,
\end{align*}
and obtain a set of 58 equations:
{ \small
\begin{multline}\label{eqB1}
 - \delta \Psi_1 + \tilde{D} \Psi_0 - D \Psi^0_{ 1 3 } =  A _{ 2  1  2 } +  ( 2 \omega - \xi )  \Psi^0_0 + \sigma \Psi^0_2 +  (  - \chi - \psi )  \Psi^1_{ 1 3 } +  (  - \bar{\pi} + 2 \tau )  \Psi^1_1 +  (  - 2 \beta - 4 \tau )  \Psi_1   +  ( 4 \gamma - \mu )  \Psi_0 \\ - \bar{\lambda} \Psi^0_1  +  (  - 2 \epsilon + 2 \bar{\epsilon} + \bar{\rho} )  \Psi^0_{ 1 3 } + 3 \sigma \Psi_2 +  (  - \chi + 2 \psi )  \bar{\Psi}^2_{ 1 3 } - \kappa \bar{\Psi}^1_3   +  ( \mqoppa - 2 \mstigma )  \Psi_{ 0 3 } + \phi \Psi^2_1   + \mdigamma \bar{\Psi}_{ 1 4 } + \phi \bar{\Psi}^2_1
\end{multline}
\begin{multline}\label{eqB2}
 - D \Psi^1_1 + D \Psi_1 + \bar{\delta} \Psi_0 - \delta \Psi^0_1 =  - A _{ 1  1  2 } +  ( 2 \upsilon - \bar{\upsilon} )  \Psi^0_0 - \phi \bar{\Psi}^0_0 + \kappa \Psi^0_2 - \mdigamma \Psi^1_{ 1 3 } +  (  - 2 \epsilon + 2 \rho + \bar{\rho} )  \Psi^1_1 + \sigma \bar{\Psi}^1_1 \\ +  ( 2 \epsilon - 4 \rho )  \Psi_1  +  ( 4 \alpha + \pi )  \Psi_0   +  (  - 2 \bar{\alpha} - 2 \beta - \bar{\pi} )  \Psi^0_1 + \bar{\kappa} \Psi^0_{ 1 3 } - 3 \kappa \Psi_2 - \mdigamma \bar{\Psi}^2_{ 1 3 } +  ( \bar{\chi} - 2 \bar{\psi} )  \Psi_{ 0 3 } +  ( \chi + \psi )  \Psi^2_1 +  (  - 2 \chi + \psi )  \bar{\Psi}^2_1
\end{multline}
\begin{multline}\label{eqB7}
\bar{\delta} \Psi_1 + D \Psi_2 - \delta \bar{\Psi}_1 - D \bar{\Psi}_2 =   - A _{ 1  2  \bar{2} } + \bar{\xi} \Psi^0_0 - \xi \bar{\Psi}^0_0 +  (  - \rho + \bar{\rho} )  \Psi^0_2 +  (  - \bar{\chi} + \bar{\psi} )  \Psi^1_{ 1 3 } +  ( \chi - \psi )  \bar{\Psi}^1_{ 1 3 }  - \pi \Psi^1_1   + \bar{\pi} \bar{\Psi}^1_1 \\ +  ( 2 \alpha + 2 \pi )  \Psi_1 + \lambda \Psi_0 +  (  - \mu + \bar{\mu} )  \Psi^0_1 - \bar{\sigma} \Psi^0_{ 1 3 } - 3 \rho \Psi_2  +  (  - 2 \bar{\alpha} - 2 \bar{\pi} )  \bar{\Psi}_1 - \bar{\lambda} \bar{\Psi}_0 + 3 \bar{\rho} \bar{\Psi}_2 + \sigma \bar{\Psi}^0_{ 1 3 }   +  ( \bar{\chi} - 2 \bar{\psi} )  \bar{\Psi}^2_{ 1 3 } \\ +  (  - \chi + 2 \psi )  \Psi^2_{ 1 3 } + \bar{\kappa} \bar{\Psi}^1_3 - \kappa \Psi^1_3 - 2 \bar{\kappa} \bar{\Psi}_3 + 2 \kappa \Psi_3  +  ( \mqoppa - \upsilon + \bar{\upsilon} )  \Psi^2_1 - \mdigamma \bar{\Psi}^2_3 +  (  - \mqoppa - \upsilon + \bar{\upsilon} )  \bar{\Psi}^2_1 + \mdigamma \Psi^2_3
\end{multline}
\begin{multline}\label{eqB9}
 - \bar{\delta} \Psi^0_{ 1 3 } + \delta \bar{\Psi}_2 + D \bar{\Psi}^1_3 - D \bar{\Psi}_3 =  A _{ \tilde{1}  1  2 } - A _{ 2  2  \bar{2} } - \bar{\pi} \Psi^0_2 +  ( \mqoppa - \upsilon + \bar{\upsilon} )  \Psi^1_{ 1 3 } - \bar{\mu} \Psi^1_1 - \bar{\lambda} \bar{\Psi}^1_1 +  (  - 2 \alpha + 2 \bar{\beta} - \pi )  \Psi^0_{ 1 3 } \\ + 2 \bar{\lambda} \bar{\Psi}_1 + 3 \bar{\pi} \bar{\Psi}_2 +  ( \mqoppa - \upsilon )  \bar{\Psi}^2_{ 1 3 } + \phi \Psi^2_{ 1 3 }  +  (  - 2 \bar{\epsilon} - \rho - \bar{\rho} )  \bar{\Psi}^1_3 +  ( 2 \bar{\epsilon} + 2 \bar{\rho} )  \bar{\Psi}_3 - \bar{\kappa} \bar{\Psi}_4 + \kappa \Psi^0_3 + \bar{\xi} \Psi_{ 0 3 } - \xi \Psi^2_1 +  (  - \bar{\chi} + \bar{\psi} )  \bar{\Psi}_{ 1 4 } \\ - \chi \bar{\Psi}^2_3 +  ( 2 \chi - \psi )  \Psi^2_3
\end{multline}
\begin{multline}\label{eqB10}
 - D \Psi^0_2 - \bar{\delta} \Psi^1_1 - \delta \bar{\Psi}^1_1 + \bar{\delta} \Psi_1 + D \Psi_2 + \delta \bar{\Psi}_1 + D \bar{\Psi}_2 =  A _{ 2  1  \bar{2} } + A _{ \bar{2}  1  2 } +  ( \rho + \bar{\rho} )  \Psi^0_2 +  ( 2 \bar{\chi} - \bar{\psi} )  \Psi^1_{ 1 3 } +  ( 2 \chi - \psi )  \bar{\Psi}^1_{ 1 3 } \\ +  (  - 2 \alpha - 2 \pi )  \Psi^1_1 +  (  - 2 \bar{\alpha} - 2 \bar{\pi} )  \bar{\Psi}^1_1 +  ( 2 \alpha + 2 \pi )  \Psi_1 + \lambda \Psi_0   +  (  - \mu - \bar{\mu} )  \Psi^0_1 + \bar{\sigma} \Psi^0_{ 1 3 } - 3 \rho \Psi_2 +  ( 2 \bar{\alpha} + 2 \bar{\pi} )  \bar{\Psi}_1 + \bar{\lambda} \bar{\Psi}_0 - 3 \bar{\rho} \bar{\Psi}_2 \\ + \sigma \bar{\Psi}^0_{ 1 3 }   +  ( 2 \bar{\chi} - \bar{\psi} )  \bar{\Psi}^2_{ 1 3 } +  ( 2 \chi - \psi )  \Psi^2_{ 1 3 } - 2 \bar{\kappa} \bar{\Psi}^1_3 - 2 \kappa \Psi^1_3   + 2 \bar{\kappa} \bar{\Psi}_3 + 2 \kappa \Psi_3 + \bar{\phi} \Psi_{ 0 3 } +  ( \upsilon - 2 \bar{\upsilon} )  \Psi^2_1 \\ +  (  - 2 \upsilon + \bar{\upsilon} )  \bar{\Psi}^2_1 + \phi \bar{\Psi}_{ 0 3 }
\end{multline}
\begin{multline}\label{eqB11}
\tilde{D} \Psi_2 + \delta \Psi_3 - D \Psi^0_3 =  A _{ \tilde{1}  1  \tilde{1} } - A _{ \bar{2}  \tilde{1}  2 } - \mu \Psi^0_2 +  ( \omega - 2 \xi )  \bar{\Psi}^1_{ 1 3 } - \nu \Psi^1_1 + 2 \nu \Psi_1 - 3 \mu \Psi_2 - \bar{\lambda} \bar{\Psi}^0_{ 1 3 }  + \bar{\omega} \bar{\Psi}^2_{ 1 3 } \\ + \xi \Psi^2_{ 1 3 } - \bar{\chi} \bar{\Psi}^0_4   +  (  - \chi + \psi )  \Psi^0_4 + \pi \bar{\Psi}^1_3 +  ( \bar{\pi} - \tau )  \Psi^1_3 +  (  - 2 \beta + 2 \tau )  \Psi_3 +  ( 2 \epsilon + 2 \bar{\epsilon} + \bar{\rho} )  \Psi^0_3 + \sigma \Psi_4 \\ +  ( \mqoppa - \mstigma - \bar{\upsilon} )  \bar{\Psi}^2_3 - \msampi \bar{\Psi}^2_1 +  ( \mqoppa - \bar{\upsilon} )  \Psi^2_3
\end{multline}
\begin{multline}\label{eqB12}
\bar{\delta} \Psi_2 - \delta \bar{\Psi}^0_{ 1 3 } + D \Psi^1_3 - D \Psi_3 =  A _{ \tilde{1}  1  \bar{2} } + A _{ \bar{2}  2  \bar{2} } - \pi \Psi^0_2 +  ( \mqoppa + \upsilon - \bar{\upsilon} )  \bar{\Psi}^1_{ 1 3 } - \lambda \Psi^1_1 - \mu \bar{\Psi}^1_1 + 2 \lambda \Psi_1   + 3 \pi \Psi_2 \\ +  (  - 2 \bar{\alpha} + 2 \beta - \bar{\pi} )  \bar{\Psi}^0_{ 1 3 } + \bar{\phi} \bar{\Psi}^2_{ 1 3 }   +  ( \mqoppa - \bar{\upsilon} )  \Psi^2_{ 1 3 } +  (  - 2 \epsilon - \rho - \bar{\rho} )  \Psi^1_3 +  ( 2 \epsilon + 2 \rho )  \Psi_3  + \bar{\kappa} \Psi^0_3 - \kappa \Psi_4 \\ +  ( 2 \bar{\chi} - \bar{\psi} )  \bar{\Psi}^2_3 - \bar{\xi} \bar{\Psi}^2_1 + \xi \bar{\Psi}_{ 0 3 } - \bar{\chi} \Psi^2_3 +  (  - \chi + \psi )  \Psi_{ 1 4 }
\end{multline}
\begin{multline}\label{eqB13}
\tilde{D} \Psi^0_0 - D \Psi^1_{ 1 3 } + \delta \Psi^2_1 + \delta \bar{\Psi}^2_1 =   A _{ 0  1  2 } +  ( 3 \gamma + \bar{\gamma} - \mu )  \Psi^0_0 - \bar{\lambda} \bar{\Psi}^0_0 +  (  - \chi + 2 \psi )  \Psi^0_2 +  (  - \epsilon + \bar{\epsilon} + 2 \bar{\rho} )  \Psi^1_{ 1 3 } + 2 \sigma \bar{\Psi}^1_{ 1 3 } \\ +  ( \mqoppa - \mstigma )  \Psi^1_1 +  (  - \mstigma + \bar{\upsilon} )  \Psi_1 - \bar{\omega} \Psi_0   +  (  - 3 \omega + 2 \xi )  \Psi^0_1 + \bar{\chi} \Psi^0_{ 1 3 } + \psi \Psi_2 + \phi \bar{\Psi}_1 +  ( \chi + \psi )  \bar{\Psi}_2 - \bar{\rho} \bar{\Psi}^2_{ 1 3 } - \sigma \Psi^2_{ 1 3 } - \mdigamma \bar{\Psi}_3 + \bar{\tau} \Psi_{ 0 3 } \\ +  ( \bar{\alpha} + \beta + 2 \bar{\pi} )  \Psi^2_1   - \bar{\kappa} \bar{\Psi}_{ 1 4 } +  ( \bar{\alpha} + \beta + 3 \tau )  \bar{\Psi}^2_1 - \kappa \Psi^2_3
\end{multline}
\begin{multline}\label{eqB14}
\delta \Psi^1_{ 1 3 } - \delta \bar{\Psi}^2_{ 1 3 } + \tilde{D} \Psi_{ 0 3 } - D \bar{\Psi}_{ 1 4 } =   - \bar{\nu} \Psi^0_0 +  (  - \bar{\alpha} + \beta + 3 \bar{\pi} )  \Psi^1_{ 1 3 } +  ( 2 \omega - \xi )  \Psi^1_1 +  (  - \omega - \xi )  \Psi_1 + \msampi \Psi_0 +  (  - 3 \mqoppa + 3 \mstigma )  \Psi^0_{ 1 3 } \\ + \phi \Psi_2 - \phi \bar{\Psi}_2 +  ( \bar{\alpha} - \beta - 3 \tau )  \bar{\Psi}^2_{ 1 3 }   + \kappa \bar{\Psi}^0_4 +  (  - 2 \chi + \psi )  \bar{\Psi}^1_3 +  ( \chi + \psi )  \bar{\Psi}_3 - \mdigamma \bar{\Psi}_4 +  ( 3 \gamma - \bar{\gamma} - \mu )  \Psi_{ 0 3 } - 2 \bar{\lambda} \Psi^2_1 \\ +  (  - \epsilon + 3 \bar{\epsilon} + \bar{\rho} )  \bar{\Psi}_{ 1 4 } + 2 \sigma \bar{\Psi}^2_3 - \bar{\lambda} \bar{\Psi}^2_1 + \sigma \Psi^2_3
\end{multline}
\begin{multline}\label{eqB15}
\bar{\delta} \Psi^0_0 - \delta \bar{\Psi}^0_0 + D \Psi^2_1 - D \bar{\Psi}^2_1 =   ( 3 \alpha + \bar{\beta} + \pi )  \Psi^0_0 +  (  - 3 \bar{\alpha} - \beta - \bar{\pi} )  \bar{\Psi}^0_0 + 2 \bar{\kappa} \Psi^1_{ 1 3 } - 2 \kappa \bar{\Psi}^1_{ 1 3 } +  ( 2 \bar{\chi} - \bar{\psi} )  \Psi^1_1 \\ +  (  - 2 \chi + \psi )  \bar{\Psi}^1_1 +  (  - \bar{\chi} - \bar{\psi} )  \Psi_1 - \bar{\phi} \Psi_0   +  (  - 3 \upsilon + 3 \bar{\upsilon} )  \Psi^0_1 - \mdigamma \Psi_2 +  ( \chi + \psi )  \bar{\Psi}_1 + \phi \bar{\Psi}_0 + \mdigamma \bar{\Psi}_2 + \bar{\kappa} \bar{\Psi}^2_{ 1 3 } - \kappa \Psi^2_{ 1 3 } + \bar{\sigma} \Psi_{ 0 3 } \\ +  ( \epsilon + \bar{\epsilon} - 3 \bar{\rho} )  \Psi^2_1 +  (  - \epsilon - \bar{\epsilon} + 3 \rho )  \bar{\Psi}^2_1 - \sigma \bar{\Psi}_{ 0 3 }
\end{multline}
\begin{multline}\label{eqB16}
D \Psi^1_{ 1 3 } + D \bar{\Psi}^2_{ 1 3 } + \bar{\delta} \Psi_{ 0 3 } - \delta \Psi^2_1 =   A _{ 0  1  2 } - \bar{\mu} \Psi^0_0 + \bar{\lambda} \bar{\Psi}^0_0 +  ( 2 \chi - \psi )  \Psi^0_2 +  ( \epsilon - \bar{\epsilon} - 2 \bar{\rho} )  \Psi^1_{ 1 3 } +  (  - \mqoppa + 2 \upsilon - \bar{\upsilon} )  \Psi^1_1 \\ + \phi \bar{\Psi}^1_1 +  ( \mqoppa - \upsilon )  \Psi_1 + \bar{\xi} \Psi_0 - \xi \Psi^0_1 +  (  - 2 \bar{\chi} + 3 \bar{\psi} )  \Psi^0_{ 1 3 }   - \chi \Psi_2 - \phi \bar{\Psi}_1 +  (  - \chi - \psi )  \bar{\Psi}_2 +  ( \epsilon - \bar{\epsilon} - 3 \rho )  \bar{\Psi}^2_{ 1 3 } + \sigma \Psi^2_{ 1 3 } - \mdigamma \bar{\Psi}^1_3 + \mdigamma \bar{\Psi}_3 \\ +  ( 3 \alpha - \bar{\beta} + \pi )  \Psi_{ 0 3 }  +  (  - \bar{\alpha} - \beta - 2 \bar{\pi} )  \Psi^2_1 + \bar{\kappa} \bar{\Psi}_{ 1 4 } - 2 \kappa \bar{\Psi}^2_3 + \bar{\pi} \bar{\Psi}^2_1 + \kappa \Psi^2_3
\end{multline}
\begin{multline}\label{eqB17}
\tilde{D} \bar{\Psi}^2_{ 1 3 } - D \bar{\Psi}^0_4 - \delta \bar{\Psi}^2_3 - \delta \Psi^2_3 =    - A _{ 0  \tilde{1}  2 } +  ( \omega - 2 \xi )  \Psi^0_2 + \mu \Psi^1_{ 1 3 } + \bar{\lambda} \bar{\Psi}^1_{ 1 3 } + \msampi \Psi_1 - \bar{\omega} \Psi^0_{ 1 3 } +  (  - \omega - \xi )  \Psi_2 - \xi \bar{\Psi}_2 \\ +  ( \gamma - \bar{\gamma} - 2 \mu )  \bar{\Psi}^2_{ 1 3 } - 2 \bar{\lambda} \Psi^2_{ 1 3 }   +  ( \epsilon + 3 \bar{\epsilon} + \bar{\rho} )  \bar{\Psi}^0_4 + \sigma \Psi^0_4 +  ( \mqoppa - \mstigma )  \bar{\Psi}^1_3 +  ( \mqoppa - \bar{\upsilon} )  \bar{\Psi}_3 + \bar{\chi} \bar{\Psi}_4 - \phi \Psi_3 +  ( 3 \chi - 2 \psi )  \Psi^0_3 + \nu \Psi_{ 0 3 } - \pi \bar{\Psi}_{ 1 4 } \\ +  ( \bar{\alpha} + \beta - 2 \tau )  \bar{\Psi}^2_3 + \bar{\nu} \bar{\Psi}^2_1 +  ( \bar{\alpha} + \beta - 3 \bar{\pi} )  \Psi^2_3
\end{multline}
\begin{multline}\label{eqB19}
\bar{\delta} \bar{\Psi}^2_{ 1 3 } - \delta \Psi^2_{ 1 3 } + D \bar{\Psi}^2_3 - D \Psi^2_3 =   A _{ 0  2  \bar{2} } +  ( \upsilon - \bar{\upsilon} )  \Psi^0_2 + \pi \Psi^1_{ 1 3 } - \bar{\pi} \bar{\Psi}^1_{ 1 3 } + \bar{\xi} \Psi_1 - \bar{\phi} \Psi^0_{ 1 3 } +  ( \mqoppa - \upsilon )  \Psi_2   - \xi \bar{\Psi}_1 \\ +  (  - \mqoppa + \bar{\upsilon} )  \bar{\Psi}_2 + \phi \bar{\Psi}^0_{ 1 3 }   +  ( \alpha - \bar{\beta} + 2 \pi )  \bar{\Psi}^2_{ 1 3 } +  (  - \bar{\alpha} + \beta - 2 \bar{\pi} )  \Psi^2_{ 1 3 } + \bar{\kappa} \bar{\Psi}^0_4 - \kappa \Psi^0_4 +  ( 2 \bar{\chi} - \bar{\psi} )  \bar{\Psi}^1_3 +  (  - 2 \chi + \psi )  \Psi^1_3 \\ - \bar{\chi} \bar{\Psi}_3 + \chi \Psi_3 + \lambda \Psi_{ 0 3 } - \mu \Psi^2_1  +  (  - \epsilon - \bar{\epsilon} - 2 \rho )  \bar{\Psi}^2_3 + \bar{\mu} \bar{\Psi}^2_1 - \bar{\lambda} \bar{\Psi}_{ 0 3 } +  ( \epsilon + \bar{\epsilon} + 2 \bar{\rho} )  \Psi^2_3
\end{multline}
\begin{multline}\label{eqB21}
\bar{\delta} \Psi_1 - \tilde{D} \Psi^0_1 + D \Psi_2 =   - A _{ 1  1  \tilde{1} } - A _{ 2  1  \bar{2} } +  (  - \bar{\omega} + \bar{\xi} )  \Psi^0_0 - \omega \bar{\Psi}^0_0 - \rho \Psi^0_2 + \bar{\psi} \Psi^1_{ 1 3 } + \chi \bar{\Psi}^1_{ 1 3 } +  (  - \pi + \bar{\tau} )  \Psi^1_1 + \tau \bar{\Psi}^1_1 \\ +  ( 2 \alpha + 2 \pi )  \Psi_1 + \lambda \Psi_0 +  (  - 2 \gamma - 2 \bar{\gamma} + \bar{\mu} )  \Psi^0_1  - \bar{\sigma} \Psi^0_{ 1 3 } - 3 \rho \Psi_2 +  ( \bar{\chi} - 2 \bar{\psi} )  \bar{\Psi}^2_{ 1 3 } - \kappa \Psi^1_3 + 2 \kappa \Psi_3 +  ( \mstigma - \upsilon )  \Psi^2_1 \\ - \mdigamma \bar{\Psi}^2_3 +  (  - \mqoppa + \mstigma - \upsilon )  \bar{\Psi}^2_1
\end{multline}
\begin{multline}\label{eqB22}
\tilde{D} \Psi^1_1 - \tilde{D} \Psi_1 - \bar{\delta} \Psi^0_{ 1 3 } + \delta \Psi_2 =   A _{ 1  \tilde{1}  2 } - A _{ 2  2  \bar{2} } - \tau \Psi^0_2 +  ( \mstigma - \upsilon )  \Psi^1_{ 1 3 } + \phi \bar{\Psi}^1_{ 1 3 } +  ( 2 \gamma - \mu - \bar{\mu} )  \Psi^1_1 +  (  - 2 \gamma + 2 \mu )  \Psi_1 \\ - \nu \Psi_0 + \bar{\nu} \Psi^0_1 +  (  - 2 \alpha + 2 \bar{\beta} - \bar{\tau} )  \Psi^0_{ 1 3 }   + 3 \tau \Psi_2 +  ( \mstigma - \upsilon + \bar{\upsilon} )  \bar{\Psi}^2_{ 1 3 } - \rho \bar{\Psi}^1_3 - \sigma \Psi^1_3 + 2 \sigma \Psi_3 +  (  - \bar{\omega} + \bar{\xi} )  \Psi_{ 0 3 } \\ - \omega \Psi^2_1 + \bar{\psi} \bar{\Psi}_{ 1 4 } - \psi \bar{\Psi}^2_3 +  ( 2 \omega - \xi )  \bar{\Psi}^2_1
\end{multline}
\begin{multline}\label{eqB23}
\bar{\delta} \Psi^0_2 + \bar{\delta} \Psi_2 - \tilde{D} \bar{\Psi}_1 + \bar{\delta} \bar{\Psi}_2 - D \Psi_3 =   - A _{ 1  \tilde{1}  \bar{2} } - A _{ \tilde{1}  1  \bar{2} } - \msampi \bar{\Psi}^0_0 +  ( \pi + \bar{\tau} )  \Psi^0_2 +  ( 2 \mqoppa - \mstigma )  \bar{\Psi}^1_{ 1 3 } + 2 \lambda \Psi_1 - \nu \Psi^0_1 \\ + 3 \pi \Psi_2 +  (  - 2 \bar{\gamma} + 2 \bar{\mu} )  \bar{\Psi}_1 - \bar{\nu} \bar{\Psi}_0 + 3 \bar{\tau} \bar{\Psi}_2 +  ( \bar{\pi} + \tau )  \bar{\Psi}^0_{ 1 3 }  +  (  - \mqoppa + 2 \mstigma )  \Psi^2_{ 1 3 }  - \mdigamma \Psi^0_4 + 2 \bar{\sigma} \bar{\Psi}_3 +  ( 2 \epsilon + 2 \rho )  \Psi_3 - \bar{\kappa} \Psi^0_3 - \kappa \Psi_4 \\ +  ( \bar{\omega} - 2 \bar{\xi} )  \Psi^2_1 +  ( \bar{\chi} - 2 \bar{\psi} )  \bar{\Psi}^2_3 +  ( \bar{\omega} - 2 \bar{\xi} )  \bar{\Psi}^2_1 +  ( \bar{\chi} - 2 \bar{\psi} )  \Psi^2_3
\end{multline}
\begin{multline}\label{eqB26}
 - \tilde{D} \bar{\Psi}^1_1 + \bar{\delta} \Psi_2 + \tilde{D} \bar{\Psi}_1 - \bar{\delta} \bar{\Psi}_2 + D \Psi^1_3 - D \Psi_3 =   A _{ \bar{2}  1  \tilde{1} } +  (  - \pi + \bar{\tau} )  \Psi^0_2 - \bar{\phi} \Psi^1_{ 1 3 } +  ( \mqoppa - \mstigma + \upsilon )  \bar{\Psi}^1_{ 1 3 } - \lambda \Psi^1_1 +  (  - 2 \bar{\gamma} + \bar{\mu} )  \bar{\Psi}^1_1 \\ + 2 \lambda \Psi_1 - \nu \Psi^0_1 + 3 \pi \Psi_2 +  ( 2 \bar{\gamma} - 2 \bar{\mu} )  \bar{\Psi}_1 + \bar{\nu} \bar{\Psi}_0 - 3 \bar{\tau} \bar{\Psi}_2  +  (  - \bar{\pi} + \tau )  \bar{\Psi}^0_{ 1 3 } + \bar{\phi} \bar{\Psi}^2_{ 1 3 } +  ( \mqoppa - \mstigma - \upsilon )  \Psi^2_{ 1 3 } + \bar{\sigma} \bar{\Psi}^1_3 +  (  - 2 \epsilon - \rho )  \Psi^1_3 \\ - 2 \bar{\sigma} \bar{\Psi}_3 +  ( 2 \epsilon + 2 \rho )  \Psi_3 + \bar{\kappa} \Psi^0_3 - \kappa \Psi_4   +  (  - 2 \bar{\omega} + \bar{\xi} )  \Psi^2_1 +  ( 2 \bar{\chi} - \bar{\psi} )  \bar{\Psi}^2_3 +  ( \bar{\omega} - \bar{\xi} )  \bar{\Psi}^2_1 + \omega \bar{\Psi}_{ 0 3 } +  (  - \bar{\chi} + \bar{\psi} )  \Psi^2_3 - \chi \Psi_{ 1 4 }
\end{multline}
\begin{multline}\label{eqB27}
 - \tilde{D} \Psi_2 + \tilde{D} \bar{\Psi}_2 + \bar{\delta} \bar{\Psi}_3 - \delta \Psi_3 =    - A _{ \tilde{1}  2  \bar{2} } +  ( \mu - \bar{\mu} )  \Psi^0_2 +  ( \bar{\omega} - 2 \bar{\xi} )  \Psi^1_{ 1 3 } +  (  - \omega + 2 \xi )  \bar{\Psi}^1_{ 1 3 } + \nu \Psi^1_1 - \bar{\nu} \bar{\Psi}^1_1 - 2 \nu \Psi_1 - \lambda \Psi^0_{ 1 3 } \\ + 3 \mu \Psi_2  + 2 \bar{\nu} \bar{\Psi}_1   - 3 \bar{\mu} \bar{\Psi}_2 + \bar{\lambda} \bar{\Psi}^0_{ 1 3 }  +  (  - \bar{\omega} + \bar{\xi} )  \bar{\Psi}^2_{ 1 3 } +  ( \omega - \xi )  \Psi^2_{ 1 3 }   + \bar{\psi} \bar{\Psi}^0_4 - \psi \Psi^0_4 - \bar{\tau} \bar{\Psi}^1_3 + \tau \Psi^1_3  +  (  - 2 \bar{\beta} + 2 \bar{\tau} )  \bar{\Psi}_3 \\ + \bar{\sigma} \bar{\Psi}_4 +  ( 2 \beta - 2 \tau )  \Psi_3 +  ( \rho - \bar{\rho} )  \Psi^0_3 - \sigma \Psi_4   - \msampi \Psi^2_1  +  ( \mstigma - \upsilon + \bar{\upsilon} )  \bar{\Psi}^2_3 + \msampi \bar{\Psi}^2_1 +  (  - \mstigma - \upsilon + \bar{\upsilon} )  \Psi^2_3
\end{multline}
\begin{multline}\label{eqB30}
\tilde{D} \Psi^0_2 - \tilde{D} \Psi_2 - \tilde{D} \bar{\Psi}_2 + \bar{\delta} \bar{\Psi}^1_3 + \delta \Psi^1_3 - \bar{\delta} \bar{\Psi}_3 - \delta \Psi_3 =  - A _{ 2  \tilde{1}  \bar{2} } - A _{ \bar{2}  \tilde{1}  2 } +  (  - \mu - \bar{\mu} )  \Psi^0_2 +  (  - 2 \bar{\omega} + \bar{\xi} )  \Psi^1_{ 1 3 } +  (  - 2 \omega + \xi )  \bar{\Psi}^1_{ 1 3 } \\ + 2 \nu \Psi^1_1 + 2 \bar{\nu} \bar{\Psi}^1_1 - 2 \nu \Psi_1 - \lambda \Psi^0_{ 1 3 } + 3 \mu \Psi_2 - 2 \bar{\nu} \bar{\Psi}_1 + 3 \bar{\mu} \bar{\Psi}_2   - \bar{\lambda} \bar{\Psi}^0_{ 1 3 }    +  (  - 2 \bar{\omega} + \bar{\xi} )  \bar{\Psi}^2_{ 1 3 } +  (  - 2 \omega + \xi )  \Psi^2_{ 1 3 }  +  (  - 2 \bar{\beta} + 2 \bar{\tau} )  \bar{\Psi}^1_3 \\ +  (  - 2 \beta + 2 \tau )  \Psi^1_3  +  ( 2 \bar{\beta} - 2 \bar{\tau} )  \bar{\Psi}_3 - \bar{\sigma} \bar{\Psi}_4 +  ( 2 \beta - 2 \tau )  \Psi_3   +  ( \rho + \bar{\rho} )  \Psi^0_3  - \sigma \Psi_4  - \bar{\phi} \bar{\Psi}_{ 1 4 } +  (  - \upsilon + 2 \bar{\upsilon} )  \bar{\Psi}^2_3 +  ( 2 \upsilon - \bar{\upsilon} )  \Psi^2_3 - \phi \Psi_{ 1 4 }
\end{multline}
\begin{multline}\label{eqB31}
 - \tilde{D} \bar{\Psi}^0_{ 1 3 } - \bar{\delta} \Psi_3 + D \Psi_4 =  A _{ \bar{2}  \tilde{1}  \bar{2} } + \lambda \Psi^0_2 +  (  - \bar{\omega} + 2 \bar{\xi} )  \bar{\Psi}^1_{ 1 3 } - \nu \bar{\Psi}^1_1 + 3 \lambda \Psi_2 +  ( 2 \gamma - 2 \bar{\gamma} + \bar{\mu} )  \bar{\Psi}^0_{ 1 3 }   +  (  - \bar{\omega} - \bar{\xi} )  \Psi^2_{ 1 3 }  \\ +  ( 2 \bar{\chi} - \bar{\psi} )  \Psi^0_4 +  ( 2 \pi - \bar{\tau} )  \Psi^1_3 +  ( 2 \alpha - 4 \pi )  \Psi_3 - \bar{\sigma} \Psi^0_3 +  (  - 4 \epsilon - \rho )  \Psi_4 + \bar{\phi} \bar{\Psi}^2_3 + \msampi \bar{\Psi}_{ 0 3 } + \bar{\phi} \Psi^2_3 +  (  - 2 \mqoppa + \mstigma )  \Psi_{ 1 4 }
\end{multline}
\begin{multline}\label{eqB32}
 - \tilde{D} \Psi^1_3 + \tilde{D} \Psi_3 - \bar{\delta} \Psi^0_3 + \delta \Psi_4 =    - A _{ \tilde{1}  \tilde{1}  \bar{2} } + \nu \Psi^0_2 - \msampi \bar{\Psi}^1_{ 1 3 } - 3 \nu \Psi_2 + \bar{\nu} \bar{\Psi}^0_{ 1 3 } - \msampi \Psi^2_{ 1 3 } - \bar{\phi} \bar{\Psi}^0_4 +  (  - \upsilon + 2 \bar{\upsilon} )  \Psi^0_4 \\ + \lambda \bar{\Psi}^1_3  +  ( 2 \gamma + 2 \mu + \bar{\mu} )  \Psi^1_3 +  (  - 2 \gamma - 4 \mu )  \Psi_3 +  ( 2 \alpha + 2 \bar{\beta} - \bar{\tau} )  \Psi^0_3 +  (  - 4 \beta + \tau )  \Psi_4 +  (  - 2 \bar{\omega} + \bar{\xi} )  \bar{\Psi}^2_3 +  ( \bar{\omega} + \bar{\xi} )  \Psi^2_3 +  ( \omega - 2 \xi )  \Psi_{ 1 4 }
\end{multline}
\begin{multline}\label{eqB34}
 - \bar{\delta} \Psi^1_{ 1 3 } + \bar{\delta} \bar{\Psi}^2_{ 1 3 } - \tilde{D} \Psi^2_1 + D \bar{\Psi}^2_3 =   A _{ 0  1  \tilde{1} } + \bar{\nu} \bar{\Psi}^0_0 +  ( \mqoppa - \mstigma )  \Psi^0_2 +  (  - \alpha + \bar{\beta} - 2 \bar{\tau} )  \Psi^1_{ 1 3 } - \bar{\pi} \bar{\Psi}^1_{ 1 3 } +  (  - \bar{\omega} + \bar{\xi} )  \Psi^1_1 + \omega \bar{\Psi}^1_1 + \bar{\xi} \Psi_1 - \msampi \Psi^0_1 \\ +  ( \mqoppa - \upsilon )  \Psi_2 - \omega \bar{\Psi}_1 +  (  - \mstigma + \upsilon )  \bar{\Psi}_2   +  ( \alpha - \bar{\beta} + 2 \pi )  \bar{\Psi}^2_{ 1 3 } + \tau \Psi^2_{ 1 3 } - \kappa \Psi^0_4 +  ( \bar{\chi} - \bar{\psi} )  \bar{\Psi}^1_3 - \chi \Psi^1_3 - \bar{\psi} \bar{\Psi}_3 + \chi \Psi_3 + \mdigamma \Psi^0_3 + \lambda \Psi_{ 0 3 } \\ +  (  - \gamma - \bar{\gamma} + 2 \bar{\mu} )  \Psi^2_1   - \bar{\sigma} \bar{\Psi}_{ 1 4 } +  (  - \epsilon - \bar{\epsilon} - 2 \rho )  \bar{\Psi}^2_3 + \bar{\mu} \bar{\Psi}^2_1 - \rho \Psi^2_3
\end{multline}
\begin{multline}\label{eqB35}
 - \bar{\delta} \Psi^1_{ 1 3 } + \delta \bar{\Psi}^1_{ 1 3 } - \tilde{D} \Psi^2_1 + \tilde{D} \bar{\Psi}^2_1 =    - A _{ 0  2  \bar{2} } - \nu \Psi^0_0 + \bar{\nu} \bar{\Psi}^0_0 +  (  - \upsilon + \bar{\upsilon} )  \Psi^0_2 +  (  - \alpha + \bar{\beta} - 2 \bar{\tau} )  \Psi^1_{ 1 3 } +  ( \bar{\alpha} - \beta + 2 \tau )  \bar{\Psi}^1_{ 1 3 } \\ +  (  - 2 \bar{\omega} + \bar{\xi} )  \Psi^1_1 +  ( 2 \omega - \xi )  \bar{\Psi}^1_1 + \bar{\omega} \Psi_1 + \bar{\phi} \Psi^0_{ 1 3 }  +  ( \mstigma - \bar{\upsilon} )  \Psi_2   - \omega \bar{\Psi}_1 +  (  - \mstigma + \upsilon )  \bar{\Psi}_2 - \phi \bar{\Psi}^0_{ 1 3 } - \bar{\tau} \bar{\Psi}^2_{ 1 3 } + \tau \Psi^2_{ 1 3 } - \bar{\psi} \bar{\Psi}_3 + \psi \Psi_3 \\ +  (  - \gamma - \bar{\gamma} + 2 \bar{\mu} )  \Psi^2_1 - \bar{\sigma} \bar{\Psi}_{ 1 4 } + \bar{\rho} \bar{\Psi}^2_3 +  ( \gamma + \bar{\gamma} - 2 \mu )  \bar{\Psi}^2_1   - \rho \Psi^2_3 + \sigma \Psi_{ 1 4 }
\end{multline}
\begin{multline}\label{eqB36}
 - \tilde{D} \Psi^1_{ 1 3 } - \tilde{D} \bar{\Psi}^2_{ 1 3 } - \bar{\delta} \bar{\Psi}_{ 1 4 } + \delta \bar{\Psi}^2_3 =   - A _{ 0  \tilde{1}  2 } +  (  - 2 \omega + \xi )  \Psi^0_2 +  (  - \gamma + \bar{\gamma} + 3 \bar{\mu} )  \Psi^1_{ 1 3 } - \bar{\lambda} \bar{\Psi}^1_{ 1 3 } + \msampi \Psi^1_1 - \msampi \Psi_1 \\ +  ( 2 \bar{\omega} - 3 \bar{\xi} )  \Psi^0_{ 1 3 } +  ( \omega + \xi )  \Psi_2 + \omega \bar{\Psi}_2   +  (  - \gamma + \bar{\gamma} + 2 \mu )  \bar{\Psi}^2_{ 1 3 } + \rho \bar{\Psi}^0_4 - \sigma \Psi^0_4 +  ( \mstigma - 2 \upsilon + \bar{\upsilon} )  \bar{\Psi}^1_3 - \phi \Psi^1_3 \\ +  (  - \mstigma + \upsilon )  \bar{\Psi}_3 - \bar{\psi} \bar{\Psi}_4 + \phi \Psi_3 + \psi \Psi^0_3 - \nu \Psi_{ 0 3 } + 2 \bar{\nu} \Psi^2_1   +  (  - \alpha + 3 \bar{\beta} - \bar{\tau} )  \bar{\Psi}_{ 1 4 } +  (  - \bar{\alpha} - \beta + 2 \tau )  \bar{\Psi}^2_3 - \bar{\nu} \bar{\Psi}^2_1 - \tau \Psi^2_3
\end{multline}
\begin{multline}\label{eqB39}
 - \bar{\delta} \bar{\Psi}^0_4 + \delta \Psi^0_4 - \tilde{D} \bar{\Psi}^2_3 + \tilde{D} \Psi^2_3 =   - \nu \Psi^1_{ 1 3 } + \bar{\nu} \bar{\Psi}^1_{ 1 3 } - \msampi \Psi_2 + \msampi \bar{\Psi}_2 - 2 \nu \bar{\Psi}^2_{ 1 3 } + 2 \bar{\nu} \Psi^2_{ 1 3 } +  ( \alpha + 3 \bar{\beta} - \bar{\tau} )  \bar{\Psi}^0_4 \\ +  (  - \bar{\alpha} - 3 \beta + \tau )  \Psi^0_4 +  (  - 2 \bar{\omega} + \bar{\xi} )  \bar{\Psi}^1_3 +  ( 2 \omega - \xi )  \Psi^1_3   +  ( \bar{\omega} + \bar{\xi} )  \bar{\Psi}_3 + \bar{\phi} \bar{\Psi}_4 +  (  - \omega - \xi )  \Psi_3 +  ( 3 \upsilon - 3 \bar{\upsilon} )  \Psi^0_3 - \phi \Psi_4 - \lambda \bar{\Psi}_{ 1 4 } \\ +  ( \gamma + \bar{\gamma} + 3 \mu )  \bar{\Psi}^2_3  +  (  - \gamma - \bar{\gamma} - 3 \bar{\mu} )  \Psi^2_3  + \bar{\lambda} \Psi_{ 1 4 }
\end{multline}
\begin{multline}\label{eqB41}
 - \delta \Psi^0_0 + \mathring{\delta} \Psi_0 + D \Psi_{ 0 3 } =   ( 2 \mae - \bar{\alpha} - 3 \beta - \bar{\pi} )  \Psi^0_0 +  ( 2 \chi + 2 \eta + \psi )  \Psi^1_1 +  (  - \chi - 4 \eta + \psi )  \Psi_1 \\ +  ( \mqoppa + 4 \theta + \bar{\upsilon} )  \Psi_0 + 3 \phi \Psi^0_1 + 3 \mdigamma \Psi^0_{ 1 3 }   - 3 \kappa \bar{\Psi}^2_{ 1 3 } +  ( 3 \epsilon - \bar{\epsilon} - \bar{\rho} - 2 \yogh )  \Psi_{ 0 3 } - 3 \sigma \bar{\Psi}^2_1
\end{multline}
\begin{multline}\label{eqB42}
 - D \Psi^1_{ 1 3 } + \mathring{\delta} \Psi_1 + D \bar{\Psi}^2_{ 1 3 } + \delta \Psi^2_1 + \delta \bar{\Psi}^2_1 =   - A _{ 1  2  0 } +  (  - \mu + \wynn )  \Psi^0_0 - \bar{\lambda} \bar{\Psi}^0_0 +  (  - \eta + 2 \psi )  \Psi^0_2 +  (  - \epsilon + \bar{\epsilon} + 2 \bar{\rho} + \yogh )  \Psi^1_{ 1 3 } \\ + 2 \sigma \bar{\Psi}^1_{ 1 3 } + \mqoppa \Psi^1_1 +  ( \mqoppa + 2 \theta + \bar{\upsilon} )  \Psi_1 + \zeta \Psi_0   +  ( 2 \xi + \bar{\zeta} )  \Psi^0_1 - \bar{\eta} \Psi^0_{ 1 3 } +  (  - \chi - 3 \eta + \psi )  \Psi_2 + \phi \bar{\Psi}_1 +  ( \chi + \psi )  \bar{\Psi}_2 +  ( \epsilon - \bar{\epsilon} - \bar{\rho} - 2 \yogh )  \bar{\Psi}^2_{ 1 3 } \\ - \sigma \Psi^2_{ 1 3 }   - \mdigamma \bar{\Psi}^1_3 - \mdigamma \bar{\Psi}_3 + \pi \Psi_{ 0 3 }  +  (  - \mae + \bar{\alpha} + \beta + 2 \bar{\pi} )  \Psi^2_1 - \bar{\kappa} \bar{\Psi}_{ 1 4 } - 2 \kappa \bar{\Psi}^2_3 +  (  - \mae + \bar{\alpha} + \beta + \bar{\pi} )  \bar{\Psi}^2_1 - \kappa \Psi^2_3
\end{multline}
\begin{multline}\label{eqB43}
 - \delta \bar{\Psi}^0_0 + \mathring{\delta} \Psi^0_1 + D \Psi^2_1 =   A _{ 1  1  0 } + \bar{\mae} \Psi^0_0 +  ( \mae - 3 \bar{\alpha} - \beta - \bar{\pi} )  \bar{\Psi}^0_0 + \mdigamma \Psi^0_2 + 2 \bar{\kappa} \Psi^1_{ 1 3 } +  ( \bar{\chi} - \bar{\eta} )  \Psi^1_1   +  (  - \chi - \eta + \psi )  \bar{\Psi}^1_1  \\ +  ( \mqoppa + 2 \theta + 2 \bar{\theta} + 3 \bar{\upsilon} )  \Psi^0_1 +  ( \chi + \psi )  \bar{\Psi}_1 + \phi \bar{\Psi}_0 + \mdigamma \bar{\Psi}_2 - \kappa \Psi^2_{ 1 3 } +  ( \epsilon + \bar{\epsilon} - 3 \bar{\rho} - \yogh )  \Psi^2_1 - \yogh \bar{\Psi}^2_1 - \sigma \bar{\Psi}_{ 0 3 }
\end{multline}
\begin{multline}\label{eqB44}
\delta \Psi^1_{ 1 3 } - \mathring{\delta} \Psi^0_{ 1 3 } - D \bar{\Psi}_{ 1 4 } =   - A _{ 2  2  0 } + \phi \Psi^0_2 +  (  - \mae - \bar{\alpha} + \beta + 3 \bar{\pi} )  \Psi^1_{ 1 3 } +  (  - \xi - \bar{\zeta} )  \Psi^1_1 +  (  - 3 \mqoppa - 2 \theta + 2 \bar{\theta} - \bar{\upsilon} )  \Psi^0_{ 1 3 } \\ - \phi \bar{\Psi}_2 - \mae \bar{\Psi}^2_{ 1 3 } + \kappa \bar{\Psi}^0_4 +  (  - 2 \chi - \eta )  \bar{\Psi}^1_3 +  ( \chi + \psi )  \bar{\Psi}_3 - \mdigamma \bar{\Psi}_4 + \wynn \Psi_{ 0 3 } - 2 \bar{\lambda} \Psi^2_1 +  (  - \epsilon + 3 \bar{\epsilon} + \bar{\rho} + \yogh )  \bar{\Psi}_{ 1 4 } + \sigma \Psi^2_3
\end{multline}
\begin{multline}\label{eqB45}
D \Psi^1_{ 1 3 } - \mathring{\delta} \Psi^1_1 + \mathring{\delta} \Psi_1 + D \bar{\Psi}^2_{ 1 3 } - \delta \Psi^2_1 + \delta \bar{\Psi}^2_1 =  A _{ 2  1  0 } - \mu \Psi^0_0 + \bar{\lambda} \bar{\Psi}^0_0 +  ( 2 \chi + \eta )  \Psi^0_2 +  ( \epsilon - \bar{\epsilon} - 2 \bar{\rho} - \yogh )  \Psi^1_{ 1 3 } + 2 \sigma \bar{\Psi}^1_{ 1 3 } \\ +  (  - \mqoppa - 2 \theta - 2 \bar{\upsilon} )  \Psi^1_1  + 2 \phi \bar{\Psi}^1_1 +  ( \mqoppa + 2 \theta + \bar{\upsilon} )  \Psi_1  + \zeta \Psi_0 - \bar{\zeta} \Psi^0_1 +  (  - 2 \bar{\chi} + \bar{\eta} )  \Psi^0_{ 1 3 } +  (  - \chi - 3 \eta + \psi )  \Psi_2 - \phi \bar{\Psi}_1 +  (  - \chi - \psi )  \bar{\Psi}_2 \\ +  ( \epsilon - \bar{\epsilon} - \bar{\rho} - \yogh )  \bar{\Psi}^2_{ 1 3 }   + \sigma \Psi^2_{ 1 3 }   - \mdigamma \bar{\Psi}^1_3 + \mdigamma \bar{\Psi}_3 +  ( \bar{\mae} + \pi )  \Psi_{ 0 3 }   +  ( \mae - \bar{\alpha} - \beta - 2 \bar{\pi} )  \Psi^2_1 + \bar{\kappa} \bar{\Psi}_{ 1 4 } - 2 \kappa \bar{\Psi}^2_3 \\ +  (  - 2 \mae   + \bar{\alpha} + \beta + \bar{\pi} )  \bar{\Psi}^2_1 + \kappa \Psi^2_3
\end{multline}
\begin{multline}\label{eqB46}
 - \delta \bar{\Psi}^1_{ 1 3 } + \mathring{\delta} \Psi_2 + D \bar{\Psi}^2_3 =   A _{ \tilde{1}  1  0 } + A _{ \bar{2}  2  0 } +  ( \mqoppa - \bar{\upsilon} )  \Psi^0_2 +  ( \mae - \bar{\alpha} + \beta - \bar{\pi} )  \bar{\Psi}^1_{ 1 3 } - \zeta \Psi^1_1 + \xi \bar{\Psi}^1_1 + 2 \zeta \Psi_1   +  ( \mqoppa + \bar{\upsilon} )  \Psi_2 \\ + \phi \bar{\Psi}^0_{ 1 3 }   +  ( \bar{\mae} + 2 \pi )  \bar{\Psi}^2_{ 1 3 } - \kappa \Psi^0_4 + \bar{\chi} \bar{\Psi}^1_3 +  (  - \chi - \eta )  \Psi^1_3 +  ( \chi + 2 \eta - \psi )  \Psi_3 + \mdigamma \Psi^0_3 +  (  - \epsilon - \bar{\epsilon} - \bar{\rho} - \yogh )  \bar{\Psi}^2_3 +  ( 2 \mu - \wynn )  \bar{\Psi}^2_1 - \sigma \Psi_{ 1 4 }
\end{multline}
\begin{multline}\label{eqB47}
\mathring{\delta} \Psi^0_0 + D \Psi^1_1 + 2 \delta \Psi^0_1 =   - A _{ 1  1  2 } + ( \mqoppa + 3 \theta + \bar{\theta} + 2 \bar{\upsilon} )  \Psi^0_0 + 2 \phi \bar{\Psi}^0_0 - 2 \kappa \Psi^0_2 + 2 \mdigamma \Psi^1_{ 1 3 } +  ( 2 \epsilon - 2 \bar{\rho} - \yogh )  \Psi^1_1  - 2 \sigma \bar{\Psi}^1_1 \\ - \yogh \Psi_1 - \bar{\mae} \Psi_0   +  (  - 3 \mae + 4 \bar{\alpha} + 4 \beta + 2 \bar{\pi} )  \Psi^0_1 - 2 \bar{\kappa} \Psi^0_{ 1 3 } - \mdigamma \bar{\Psi}^2_{ 1 3 } +  (  - \bar{\chi} + \bar{\eta} )  \Psi_{ 0 3 } +  (  - 2 \chi - 2 \psi )  \Psi^2_1 +  ( \chi + 3 \eta - 2 \psi )  \bar{\Psi}^2_1
\end{multline}
\begin{multline}\label{eqB48}
 - \delta \Psi^1_1 - 2 D \Psi^0_{ 1 3 } + \mathring{\delta} \Psi_{ 0 3 } =   - A _{ 2  1  2 } + (  - \xi - \bar{\zeta} )  \Psi^0_0 + 2 \sigma \Psi^0_2 +  (  - 2 \chi - 2 \psi )  \Psi^1_{ 1 3 } +  ( 2 \mae - 2 \beta - 2 \bar{\pi} )  \Psi^1_1 - \mae \Psi_1 + \wynn \Psi_0 - 2 \bar{\lambda} \Psi^0_1 \\ +  (  - 4 \epsilon + 4 \bar{\epsilon} + 2 \bar{\rho} + 3 \yogh )  \Psi^0_{ 1 3 }  +  (  - 2 \chi - 3 \eta + \psi )  \bar{\Psi}^2_{ 1 3 } - 2 \kappa \bar{\Psi}^1_3 +  ( 2 \mqoppa + 3 \theta - \bar{\theta} + \bar{\upsilon} )  \Psi_{ 0 3 } + 2 \phi \Psi^2_1 + 2 \mdigamma \bar{\Psi}_{ 1 4 } - \phi \bar{\Psi}^2_1
\end{multline}
\begin{multline}\label{eqB49}
 - \delta \Psi^0_2 + \mathring{\delta} \bar{\Psi}^2_{ 1 3 } + D \bar{\Psi}^1_3 =   A _{ 0  2  0 } - A _{ \tilde{1}  1  2 } +  ( \mae - 2 \bar{\pi} )  \Psi^0_2 +  (  - \mqoppa + \bar{\upsilon} )  \Psi^1_{ 1 3 } + \phi \bar{\Psi}^1_{ 1 3 } - \mu \Psi^1_1 - \bar{\lambda} \bar{\Psi}^1_1 + \wynn \Psi_1 \\ +  (  - \bar{\mae} - 2 \pi )  \Psi^0_{ 1 3 } - \mae \Psi_2   +  ( 2 \mqoppa + \theta - \bar{\theta} + \bar{\upsilon} )  \bar{\Psi}^2_{ 1 3 } + \phi \Psi^2_{ 1 3 } + \mdigamma \bar{\Psi}^0_4 +  (  - 2 \bar{\epsilon} - \bar{\rho} - \yogh )  \bar{\Psi}^1_3 - \sigma \Psi^1_3 + 2 \kappa \Psi^0_3 + \zeta \Psi_{ 0 3 } + \xi \Psi^2_1 \\ - \bar{\chi} \bar{\Psi}_{ 1 4 } +  (  - 2 \chi - 2 \eta + \psi )  \bar{\Psi}^2_3 +  ( \xi + \bar{\zeta} )  \bar{\Psi}^2_1 +  ( \chi + \psi )  \Psi^2_3
\end{multline}
\begin{multline}\label{eqB51}
 - \tilde{D} \Psi^0_0 + \mathring{\delta} \Psi_1 + D \bar{\Psi}^2_{ 1 3 } =   - A _{ 2  1  0 } +  (  - 3 \gamma - \bar{\gamma} + \wynn )  \Psi^0_0 +  ( \chi - \eta )  \Psi^0_2 + \yogh \Psi^1_{ 1 3 } + \mstigma \Psi^1_1 +  ( \mqoppa + \mstigma + 2 \theta )  \Psi_1  +  ( \bar{\omega} + \zeta )  \Psi_0 \\ +  ( 3 \omega + \bar{\zeta} )  \Psi^0_1   +  (  - \bar{\chi} - \bar{\eta} )  \Psi^0_{ 1 3 } +  (  - \chi - 3 \eta )  \Psi_2 +  ( \epsilon - \bar{\epsilon} - 2 \yogh )  \bar{\Psi}^2_{ 1 3 } - \mdigamma \bar{\Psi}^1_3 +  ( \pi - \bar{\tau} )  \Psi_{ 0 3 } - \mae \Psi^2_1 - 2 \kappa \bar{\Psi}^2_3 +  (  - \mae + \bar{\pi} - 3 \tau )  \bar{\Psi}^2_1
\end{multline}
\begin{multline}\label{eqB53}
 - \mathring{\delta} \Psi^0_{ 1 3 } + \delta \bar{\Psi}^2_{ 1 3 } - \tilde{D} \Psi_{ 0 3 } =   - A _{ 2  2  0 } + \bar{\nu} \Psi^0_0 + \phi \Psi^0_2 - \mae \Psi^1_{ 1 3 } +  (  - 2 \omega - \bar{\zeta} )  \Psi^1_1 +  ( \omega + \xi )  \Psi_1 - \msampi \Psi_0 \\ +  (  - 3 \mstigma - 2 \theta + 2 \bar{\theta} - \bar{\upsilon} )  \Psi^0_{ 1 3 }   - \phi \Psi_2 +  (  - \mae - \bar{\alpha} + \beta + 3 \tau )  \bar{\Psi}^2_{ 1 3 } +  (  - \eta - \psi )  \bar{\Psi}^1_3 +  (  - 3 \gamma + \bar{\gamma} + \mu + \wynn )  \Psi_{ 0 3 } + \yogh \bar{\Psi}_{ 1 4 } - 2 \sigma \bar{\Psi}^2_3 + \bar{\lambda} \bar{\Psi}^2_1
\end{multline}
\begin{multline}\label{eqB54}
 - \mathring{\delta} \Psi^1_1 + \mathring{\delta} \Psi_1 - \bar{\delta} \Psi_{ 0 3 } + \delta \bar{\Psi}^2_1 =   A _{ 1  2  0 } +  (  - \mu + \bar{\mu} )  \Psi^0_0 +  ( \eta + \psi )  \Psi^0_2 - \yogh \Psi^1_{ 1 3 } + 2 \sigma \bar{\Psi}^1_{ 1 3 } +  (  - 2 \theta - 2 \upsilon - \bar{\upsilon} )  \Psi^1_1 + \phi \bar{\Psi}^1_1 \\ +  ( 2 \theta + \upsilon + \bar{\upsilon} )  \Psi_1 +  (  - \bar{\xi} + \zeta )  \Psi_0 +  ( \xi - \bar{\zeta} )  \Psi^0_1 +  ( \bar{\eta} - 3 \bar{\psi} )  \Psi^0_{ 1 3 } +  (  - 3 \eta + \psi )  \Psi_2 +  ( 3 \rho - \bar{\rho} - \yogh )  \bar{\Psi}^2_{ 1 3 } +  ( \bar{\mae} - 3 \alpha + \bar{\beta} )  \Psi_{ 0 3 } + \mae \Psi^2_1 \\ +  (  - 2 \mae + \bar{\alpha} + \beta )  \bar{\Psi}^2_1
\end{multline}
\begin{multline}\label{eqB55}
\mathring{\delta} \Psi^0_2 + \mathring{\delta} \Psi_2 + \mathring{\delta} \bar{\Psi}_2 + \tilde{D} \Psi^2_1 + D \bar{\Psi}^2_3 + \tilde{D} \bar{\Psi}^2_1 + D \Psi^2_3 =   - A _{ 1  \tilde{1}  0 } - A _{ \tilde{1}  1  0 } - \nu \Psi^0_0 - \bar{\nu} \bar{\Psi}^0_0 +  ( 2 \mqoppa + 2 \mstigma )  \Psi^0_2 +  (  - \pi + 2 \bar{\tau} )  \Psi^1_{ 1 3 } \\ +  (  - \bar{\pi} + 2 \tau )  \bar{\Psi}^1_{ 1 3 } +  ( \bar{\omega} + 2 \zeta )  \Psi_1   + 2 \msampi \Psi^0_1 +  ( \mqoppa + \mstigma )  \Psi_2   +  ( \omega + 2 \bar{\zeta} )  \bar{\Psi}_1  +  ( \mqoppa + \mstigma )  \bar{\Psi}_2  +  ( 2 \pi - \bar{\tau} )  \bar{\Psi}^2_{ 1 3 } +  ( 2 \bar{\pi} - \tau )  \Psi^2_{ 1 3 } - \bar{\kappa} \bar{\Psi}^0_4 - \kappa \Psi^0_4 \\ +  ( \bar{\chi} + 2 \bar{\eta} )  \bar{\Psi}_3 +  ( \chi + 2 \eta )  \Psi_3  + 2 \mdigamma \Psi^0_3 +  ( \gamma + \bar{\gamma} - 2 \wynn )  \Psi^2_1   +  (  - \epsilon - \bar{\epsilon} - 2 \yogh )  \bar{\Psi}^2_3 +  ( \gamma + \bar{\gamma} - 2 \wynn )  \bar{\Psi}^2_1 +  (  - \epsilon - \bar{\epsilon} - 2 \yogh )  \Psi^2_3
\end{multline}
\begin{multline}\label{eqB58}
\mathring{\delta} \Psi_2 - \mathring{\delta} \bar{\Psi}_2 - \tilde{D} \Psi^2_1 + D \bar{\Psi}^2_3 + \tilde{D} \bar{\Psi}^2_1 - D \Psi^2_3 =   - \nu \Psi^0_0 + \bar{\nu} \bar{\Psi}^0_0 +  (  - \bar{\mae} + \pi - 2 \bar{\tau} )  \Psi^1_{ 1 3 } +  ( \mae - \bar{\pi} + 2 \tau )  \bar{\Psi}^1_{ 1 3 } +  (  - 2 \bar{\omega} - \zeta )  \Psi^1_1 \\ +  ( 2 \omega + \bar{\zeta} )  \bar{\Psi}^1_1   +  ( \bar{\omega} + 2 \zeta )  \Psi_1 +  ( \mqoppa + \mstigma )  \Psi_2   +  (  - \omega - 2 \bar{\zeta} )  \bar{\Psi}_1 +  (  - \mqoppa - \mstigma )  \bar{\Psi}_2 +  ( \bar{\mae} + 2 \pi - \bar{\tau} )  \bar{\Psi}^2_{ 1 3 }  +  (  - \mae - 2 \bar{\pi} + \tau )  \Psi^2_{ 1 3 } \\ + \bar{\kappa} \bar{\Psi}^0_4 - \kappa \Psi^0_4  +  ( 2 \bar{\chi} + \bar{\eta} )  \bar{\Psi}^1_3 +  (  - 2 \chi - \eta )  \Psi^1_3   +  (  - \bar{\chi} - 2 \bar{\eta} )  \bar{\Psi}_3  +  ( \chi + 2 \eta )  \Psi_3 +  (  - \gamma - \bar{\gamma} + \wynn )  \Psi^2_1 +  (  - \epsilon - \bar{\epsilon} - \yogh )  \bar{\Psi}^2_3 \\ +  ( \gamma + \bar{\gamma} - \wynn )  \bar{\Psi}^2_1 +  ( \epsilon + \bar{\epsilon} + \yogh )  \Psi^2_3
\end{multline}
\begin{multline}\label{eqB65}
 - \delta \bar{\Psi}^0_4 + \mathring{\delta} \bar{\Psi}_4 + \tilde{D} \bar{\Psi}_{ 1 4 } =   - 3 \bar{\nu} \Psi^1_{ 1 3 } + 3 \msampi \Psi^0_{ 1 3 } +  ( 2 \mae + 3 \bar{\alpha} + \beta - \tau )  \bar{\Psi}^0_4 +  ( 2 \omega + \xi + 2 \bar{\zeta} )  \bar{\Psi}^1_3 \\ +  (  - \omega + \xi - 4 \bar{\zeta} )  \bar{\Psi}_3  +  ( \mstigma - 4 \bar{\theta} + \bar{\upsilon} )  \bar{\Psi}_4 + 3 \phi \Psi^0_3 +  ( \gamma - 3 \bar{\gamma} - \mu - 2 \wynn )  \bar{\Psi}_{ 1 4 } - 3 \bar{\lambda} \Psi^2_3
\end{multline}
\begin{multline}\label{eqB66}
 - \tilde{D} \Psi^1_{ 1 3 } - \tilde{D} \bar{\Psi}^2_{ 1 3 } + \mathring{\delta} \bar{\Psi}^1_3 - \mathring{\delta} \bar{\Psi}_3 + \delta \bar{\Psi}^2_3 - \delta \Psi^2_3 =   - A _{ 2  \tilde{1}  0 } +  (  - 2 \omega - \bar{\zeta} )  \Psi^0_2 +  (  - \gamma + \bar{\gamma} + \mu + \wynn )  \Psi^1_{ 1 3 } - \bar{\lambda} \bar{\Psi}^1_{ 1 3 } + \msampi \Psi^1_1 - \msampi \Psi_1 \\ +  ( 2 \bar{\omega} - \zeta )  \Psi^0_{ 1 3 } +  ( \omega + \xi )  \Psi_2   +  ( \omega - \xi + 3 \bar{\zeta} )  \bar{\Psi}_2 +  (  - \gamma + \bar{\gamma} + 2 \mu + \wynn )  \bar{\Psi}^2_{ 1 3 } - 2 \bar{\lambda} \Psi^2_{ 1 3 } + \bar{\rho} \bar{\Psi}^0_4 - \sigma \Psi^0_4 +  ( \mstigma - 2 \bar{\theta} + 2 \bar{\upsilon} )  \bar{\Psi}^1_3 \\ - 2 \phi \Psi^1_3  +  (  - \mstigma + 2 \bar{\theta} - \bar{\upsilon} )  \bar{\Psi}_3 - \bar{\eta} \bar{\Psi}_4 + \phi \Psi_3 + \eta \Psi^0_3 - \nu \Psi_{ 0 3 } + 2 \bar{\nu} \Psi^2_1  +  (  - \bar{\mae} - \bar{\tau} )  \bar{\Psi}_{ 1 4 } \\ +  (  - \mae - \bar{\alpha} - \beta + 2 \tau )  \bar{\Psi}^2_3 - \bar{\nu} \bar{\Psi}^2_1   +  ( 2 \mae + \bar{\alpha} + \beta - \tau )  \Psi^2_3
\end{multline}
\begin{multline}\label{eqB67}
\bar{\delta} \Psi^1_{ 1 3 } - \delta \bar{\Psi}^1_{ 1 3 } + \mathring{\delta} \Psi_2 - \mathring{\delta} \bar{\Psi}_2 - \bar{\delta} \bar{\Psi}^2_{ 1 3 } + \delta \Psi^2_{ 1 3 } =   (  - \bar{\mae} + \alpha - \bar{\beta} )  \Psi^1_{ 1 3 } +  ( \mae - \bar{\alpha} + \beta )  \bar{\Psi}^1_{ 1 3 } +  (  - \bar{\xi} - \zeta )  \Psi^1_1 +  ( \xi + \bar{\zeta} )  \bar{\Psi}^1_1 \\ +  (  - \bar{\xi} + 2 \zeta )  \Psi_1  +  ( \upsilon + \bar{\upsilon} )  \Psi_2  +  ( \xi - 2 \bar{\zeta} )  \bar{\Psi}_1 +  (  - \upsilon - \bar{\upsilon} )  \bar{\Psi}_2   +  ( \bar{\mae} - \alpha + \bar{\beta} )  \bar{\Psi}^2_{ 1 3 }  +  (  - \mae + \bar{\alpha} - \beta )  \Psi^2_{ 1 3 } +  ( \bar{\eta} + \bar{\psi} )  \bar{\Psi}^1_3 \\ +  (  - \eta - \psi )  \Psi^1_3   +  (  - 2 \bar{\eta} + \bar{\psi} )  \bar{\Psi}_3 +  ( 2 \eta - \psi )  \Psi_3   - \lambda \Psi_{ 0 3 }   +  ( \mu - 2 \bar{\mu} + \wynn )  \Psi^2_1  + \bar{\sigma} \bar{\Psi}_{ 1 4 } +  ( 2 \rho - \bar{\rho} - \yogh )  \bar{\Psi}^2_3 \\ +  ( 2 \mu - \bar{\mu} - \wynn )  \bar{\Psi}^2_1 + \bar{\lambda} \bar{\Psi}_{ 0 3 } +  ( \rho - 2 \bar{\rho} + \yogh )  \Psi^2_3 - \sigma \Psi_{ 1 4 }
\end{multline}
\begin{multline}\label{eqB69}
\mathring{\delta} \bar{\Psi}^1_3 - \mathring{\delta} \bar{\Psi}_3 + \bar{\delta} \bar{\Psi}_{ 1 4 } - \delta \Psi^2_3 =   - A _{ \tilde{1}  2  0 } +  (  - \xi - \bar{\zeta} )  \Psi^0_2 +  ( \mu - 3 \bar{\mu} + \wynn )  \Psi^1_{ 1 3 } +  ( 3 \bar{\xi} - \zeta )  \Psi^0_{ 1 3 } +  (  - \xi + 3 \bar{\zeta} )  \bar{\Psi}_2 + \wynn \bar{\Psi}^2_{ 1 3 } - 2 \bar{\lambda} \Psi^2_{ 1 3 } \\ +  (  - \rho + \bar{\rho} )  \bar{\Psi}^0_4  +  (  - 2 \bar{\theta} + 2 \upsilon + \bar{\upsilon} )  \bar{\Psi}^1_3 - \phi \Psi^1_3 +  ( 2 \bar{\theta} - \upsilon - \bar{\upsilon} )  \bar{\Psi}_3 +  (  - \bar{\eta} + \bar{\psi} )  \bar{\Psi}_4 +  ( \eta - \psi )  \Psi^0_3 +  (  - \bar{\mae} + \alpha - 3 \bar{\beta} )  \bar{\Psi}_{ 1 4 } \\ - \mae \bar{\Psi}^2_3 +  ( 2 \mae + \bar{\alpha} + \beta )  \Psi^2_3
\end{multline}
\begin{multline}\label{eqB70}
 - \mathring{\delta} \Psi^0_2 - \bar{\delta} \Psi^1_{ 1 3 } - \delta \bar{\Psi}^1_{ 1 3 } + \mathring{\delta} \Psi_2 + \mathring{\delta} \bar{\Psi}_2 - \bar{\delta} \bar{\Psi}^2_{ 1 3 } - \delta \Psi^2_{ 1 3 } =   - A _{ 2  \bar{2}  0 } - A _{ \bar{2}  2  0 } +  (  - 2 \upsilon - 2 \bar{\upsilon} )  \Psi^0_2 +  ( 2 \bar{\mae} - \alpha + \bar{\beta} )  \Psi^1_{ 1 3 } \\ +  ( 2 \mae - \bar{\alpha} + \beta )  \bar{\Psi}^1_{ 1 3 }  +  ( \bar{\xi} - 2 \zeta )  \Psi^1_1 +  ( \xi - 2 \bar{\zeta} )  \bar{\Psi}^1_1 +  (  - \bar{\xi} + 2 \zeta )  \Psi_1 + 2 \bar{\phi} \Psi^0_{ 1 3 }   +  ( \upsilon + \bar{\upsilon} )  \Psi_2 +  (  - \xi + 2 \bar{\zeta} )  \bar{\Psi}_1 +  ( \upsilon + \bar{\upsilon} )  \bar{\Psi}_2 + 2 \phi \bar{\Psi}^0_{ 1 3 } \\ +  ( 2 \bar{\mae} - \alpha + \bar{\beta} )  \bar{\Psi}^2_{ 1 3 } +  ( 2 \mae - \bar{\alpha} + \beta )  \Psi^2_{ 1 3 } +  (  - 2 \bar{\eta} + \bar{\psi} )  \bar{\Psi}^1_3 +  (  - 2 \eta + \psi )  \Psi^1_3   +  ( 2 \bar{\eta} - \bar{\psi} )  \bar{\Psi}_3  +  ( 2 \eta - \psi )  \Psi_3 - \lambda \Psi_{ 0 3 } +  (  - \mu + 2 \bar{\mu} )  \Psi^2_1 \\ - \bar{\sigma} \bar{\Psi}_{ 1 4 }  +  ( 2 \rho - \bar{\rho} )  \bar{\Psi}^2_3 +  ( 2 \mu - \bar{\mu} )  \bar{\Psi}^2_1 - \bar{\lambda} \bar{\Psi}_{ 0 3 } +  (  - \rho + 2 \bar{\rho} )  \Psi^2_3 - \sigma \Psi_{ 1 4 }
\end{multline}
\begin{multline}\label{eqB71}
 - \tilde{D} \bar{\Psi}^1_{ 1 3 } + D \Psi^0_4 - \mathring{\delta} \Psi_3 =   A _{ \bar{2}  \tilde{1}  0 } +  (  - \bar{\omega} + \zeta )  \Psi^0_2 +  ( \gamma - \bar{\gamma} + 2 \wynn )  \bar{\Psi}^1_{ 1 3 } + \msampi \bar{\Psi}^1_1 +  ( \bar{\omega} + 3 \zeta )  \Psi_2   +  ( \omega + \bar{\zeta} )  \bar{\Psi}^0_{ 1 3 } - \wynn \Psi^2_{ 1 3 }  \\ +  (  - 3 \epsilon - \bar{\epsilon} - \yogh )  \Psi^0_4 - \mqoppa \Psi^1_3 +  (  - \mqoppa - \mstigma + 2 \theta )  \Psi_3 +  (  - 3 \bar{\chi} - \bar{\eta} )  \Psi^0_3 +  (  - \chi - \eta )  \Psi_4 +  ( \bar{\mae} + 3 \pi - \bar{\tau} )  \bar{\Psi}^2_3 + 2 \nu \bar{\Psi}^2_1 + \bar{\mae} \Psi^2_3 +  ( \bar{\pi} - \tau )  \Psi_{ 1 4 }
\end{multline}
\begin{multline}\label{eqB73}
\delta \Psi^0_4 - \mathring{\delta} \Psi^0_3 - \tilde{D} \bar{\Psi}^2_3 =   - A _{ \tilde{1}  \tilde{1}  0 } - \msampi \Psi^0_2 + \bar{\nu} \bar{\Psi}^1_{ 1 3 } - \msampi \Psi_2 - 2 \nu \bar{\Psi}^2_{ 1 3 } - \bar{\mae} \bar{\Psi}^0_4 +  (  - \mae - \bar{\alpha} - 3 \beta + \tau )  \Psi^0_4 +  (  - \bar{\omega} + \zeta )  \bar{\Psi}^1_3 \\ +  ( \omega - \xi + \bar{\zeta} )  \Psi^1_3  +  (  - \omega - \xi )  \Psi_3 +  (  - \mstigma + 2 \theta + 2 \bar{\theta} - 3 \bar{\upsilon} )  \Psi^0_3 - \phi \Psi_4 +  ( \gamma + \bar{\gamma} + 3 \mu + \wynn )  \bar{\Psi}^2_3 + \wynn \Psi^2_3 + \bar{\lambda} \Psi_{ 1 4 }
\end{multline}
\begin{multline}\label{eqB75}
D \Psi^0_2 + 2 \tilde{D} \Psi^0_1 - \mathring{\delta} \Psi^2_1 - \mathring{\delta} \bar{\Psi}^2_1 =    - A _{ 0  1  0 } - A _{ 1  1  \tilde{1} } +  ( 2 \bar{\omega} + \zeta )  \Psi^0_0 +  ( 2 \omega + \bar{\zeta} )  \bar{\Psi}^0_0 - 2 \yogh \Psi^0_2 +  (  - \bar{\chi} - 2 \bar{\eta} )  \Psi^1_{ 1 3 } +  (  - \chi - 2 \eta )  \bar{\Psi}^1_{ 1 3 } \\ +  ( \pi - 2 \bar{\tau} )  \Psi^1_1 +  ( \bar{\pi} - 2 \tau )  \bar{\Psi}^1_1 - \bar{\mae} \Psi_1   +  ( 4 \gamma + 4 \bar{\gamma} - 2 \wynn )  \Psi^0_1 - \yogh \Psi_2 - \mae \bar{\Psi}_1 - \yogh \bar{\Psi}_2  +  (  - \bar{\chi} + \bar{\eta} )  \bar{\Psi}^2_{ 1 3 } +  (  - \chi + \eta )  \Psi^2_{ 1 3 } + \bar{\kappa} \bar{\Psi}^1_3 \\ + \kappa \Psi^1_3   +  (  - \mqoppa - 2 \mstigma - \theta - \bar{\theta} )  \Psi^2_1 - \mdigamma \bar{\Psi}^2_3   +  (  - \mqoppa - 2 \mstigma - \theta - \bar{\theta} )  \bar{\Psi}^2_1 - \mdigamma \Psi^2_3
\end{multline}
\begin{multline}\label{eqB76}
 - \mathring{\delta} \Psi^1_{ 1 3 } - \tilde{D} \Psi^1_1 + \mathring{\delta} \bar{\Psi}^2_{ 1 3 } + D \bar{\Psi}^1_3 =   - A _{ 2  1  \tilde{1} } - \msampi \Psi^0_0 +  (  - 2 \bar{\pi} + 2 \tau )  \Psi^0_2 +  (  - \mqoppa - 2 \mstigma - \theta + \bar{\theta} )  \Psi^1_{ 1 3 } +  (  - 2 \gamma + \wynn )  \Psi^1_1 + \wynn \Psi_1 - 2 \bar{\nu} \Psi^0_1 \\ +  (  - 2 \pi + 2 \bar{\tau} )  \Psi^0_{ 1 3 } - \mae \Psi_2 + \mae \bar{\Psi}_2   +  ( 2 \mqoppa + \mstigma + \theta - \bar{\theta} )  \bar{\Psi}^2_{ 1 3 } + \mdigamma \bar{\Psi}^0_4 +  (  - 2 \bar{\epsilon} - \yogh )  \bar{\Psi}^1_3 - \yogh \bar{\Psi}_3 + 2 \kappa \Psi^0_3 +  ( \bar{\omega} + \zeta )  \Psi_{ 0 3 } \\ +  ( 2 \omega + 2 \bar{\zeta} )  \Psi^2_1 +  (  - \bar{\chi} - \bar{\eta} )  \bar{\Psi}_{ 1 4 } +  (  - 2 \chi - 2 \eta )  \bar{\Psi}^2_3   +  (  - \omega + \bar{\zeta} )  \bar{\Psi}^2_1 +  ( \chi - \eta )  \Psi^2_3
\end{multline}
\begin{multline}\label{eqB78}
 - D \Psi^0_2 - \bar{\delta} \Psi^1_1 + \mathring{\delta} \Psi^2_1 =   - A _{ 0  1  0 } - A _{ 2  1  \bar{2} }  - \bar{\xi} \Psi^0_0 - \bar{\zeta} \bar{\Psi}^0_0 +  ( 2 \rho + \yogh )  \Psi^0_2 +  ( \bar{\chi} + 2 \bar{\eta} - 2 \bar{\psi} )  \Psi^1_{ 1 3 } + \chi \bar{\Psi}^1_{ 1 3 } +  ( \bar{\mae} - 2 \alpha - \pi )  \Psi^1_1 \\ +  (  - \mae - \bar{\pi} )  \bar{\Psi}^1_1 +  (  - 2 \bar{\mu} + \wynn )  \Psi^0_1   + 2 \bar{\sigma} \Psi^0_{ 1 3 } + \mae \bar{\Psi}_1 + \yogh \bar{\Psi}_2 +  ( \bar{\chi} + \bar{\psi} )  \bar{\Psi}^2_{ 1 3 } +  ( \chi - \eta )  \Psi^2_{ 1 3 } - \bar{\kappa} \bar{\Psi}^1_3 - \kappa \Psi^1_3 + \bar{\phi} \Psi_{ 0 3 } \\ +  ( \mqoppa + \theta + \bar{\theta} + 2 \upsilon )  \Psi^2_1 + \mdigamma \bar{\Psi}^2_3 +  ( \mqoppa - \upsilon )  \bar{\Psi}^2_1 + \mdigamma \Psi^2_3
\end{multline}
\begin{multline}\label{eqB79}
\delta \Psi^0_2 - \mathring{\delta} \Psi^1_{ 1 3 } - \tilde{D} \Psi^1_1 =    - A _{ 0  2  0 } + A _{ 1  \tilde{1}  2 } - \msampi \Psi^0_0 +  (  - \mae + 2 \tau )  \Psi^0_2 +  (  - 2 \mstigma - \theta + \bar{\theta} - \bar{\upsilon} )  \Psi^1_{ 1 3 } - \phi \bar{\Psi}^1_{ 1 3 } +  (  - 2 \gamma + \mu + \wynn )  \Psi^1_1 + \bar{\lambda} \bar{\Psi}^1_1 \\ - 2 \bar{\nu} \Psi^0_1 +  ( \bar{\mae} + 2 \bar{\tau} )  \Psi^0_{ 1 3 }   + \mae \bar{\Psi}_2 +  ( \mstigma - \bar{\upsilon} )  \bar{\Psi}^2_{ 1 3 } - \phi \Psi^2_{ 1 3 } + \bar{\rho} \bar{\Psi}^1_3 + \sigma \Psi^1_3 - \yogh \bar{\Psi}_3 + \bar{\omega} \Psi_{ 0 3 } +  ( 2 \omega - \xi + 2 \bar{\zeta} )  \Psi^2_1 - \bar{\eta} \bar{\Psi}_{ 1 4 } - \psi \bar{\Psi}^2_3 \\ +  (  - \omega - \xi )  \bar{\Psi}^2_1 +  (  - \eta - \psi )  \Psi^2_3
\end{multline}
\begin{multline}\label{eqB80}
2 \tilde{D} \Psi^0_{ 1 3 } + \delta \bar{\Psi}^1_3 - \mathring{\delta} \bar{\Psi}_{ 1 4 } =   A _{ 2  \tilde{1}  2 } - 2 \bar{\lambda} \Psi^0_2 +  ( 2 \omega - \xi + 3 \bar{\zeta} )  \Psi^1_{ 1 3 } + 2 \bar{\nu} \Psi^1_1 +  ( 4 \gamma - 4 \bar{\gamma} - 2 \mu - 3 \wynn )  \Psi^0_{ 1 3 } +  ( 2 \omega + 2 \xi )  \bar{\Psi}^2_{ 1 3 } \\ +  ( \eta + \psi )  \bar{\Psi}^0_4   +  (  - 2 \mae - 2 \bar{\alpha} + 2 \tau )  \bar{\Psi}^1_3 + \mae \bar{\Psi}_3 - \yogh \bar{\Psi}_4 + 2 \sigma \Psi^0_3 - 2 \msampi \Psi_{ 0 3 } +  (  - 2 \mstigma - \theta + 3 \bar{\theta} - \bar{\upsilon} )  \bar{\Psi}_{ 1 4 } - 2 \phi \bar{\Psi}^2_3 + \phi \Psi^2_3
\end{multline}
\begin{multline}\label{eqB81}
 - \bar{\delta} \Psi^1_1 + \delta \bar{\Psi}^1_1 + \mathring{\delta} \Psi^2_1 - \mathring{\delta} \bar{\Psi}^2_1 =   - A _{ 1  2  \bar{2} } + (  - \bar{\xi} + \zeta )  \Psi^0_0 +  ( \xi - \bar{\zeta} )  \bar{\Psi}^0_0 +  ( 2 \rho - 2 \bar{\rho} )  \Psi^0_2 +  ( 2 \bar{\eta} - 2 \bar{\psi} )  \Psi^1_{ 1 3 } +  (  - 2 \eta + 2 \psi )  \bar{\Psi}^1_{ 1 3 } \\ +  ( 2 \bar{\mae} - 2 \alpha )  \Psi^1_1 +  (  - 2 \mae + 2 \bar{\alpha} )  \bar{\Psi}^1_1 - \bar{\mae} \Psi_1   +  ( 2 \mu - 2 \bar{\mu} )  \Psi^0_1 + 2 \bar{\sigma} \Psi^0_{ 1 3 } - \yogh \Psi_2 + \mae \bar{\Psi}_1 + \yogh \bar{\Psi}_2 - 2 \sigma \bar{\Psi}^0_{ 1 3 } +  ( \bar{\eta} + \bar{\psi} )  \bar{\Psi}^2_{ 1 3 } \\ +  (  - \eta - \psi )  \Psi^2_{ 1 3 } + \bar{\phi} \Psi_{ 0 3 } +  ( \theta + \bar{\theta} + 2 \upsilon + \bar{\upsilon} )  \Psi^2_1  +  (  - \theta - \bar{\theta} - \upsilon - 2 \bar{\upsilon} )  \bar{\Psi}^2_1 - \phi \bar{\Psi}_{ 0 3 }
\end{multline}
\begin{multline}\label{eqB82}
 - \delta \Psi^0_2 + \mathring{\delta} \Psi^1_{ 1 3 } + 2 \bar{\delta} \Psi^0_{ 1 3 } + \mathring{\delta} \bar{\Psi}^2_{ 1 3 } =    - A _{ 0  2  0 } - A _{ 2  2  \bar{2} } + 2 \mae \Psi^0_2 +  ( \theta - \bar{\theta} + 2 \upsilon + \bar{\upsilon} )  \Psi^1_{ 1 3 } + \phi \bar{\Psi}^1_{ 1 3 } +  (  - \mu + 2 \bar{\mu} - \wynn )  \Psi^1_1 - \bar{\lambda} \bar{\Psi}^1_1 \\ + \wynn \Psi_1   +  (  - 2 \bar{\mae} + 4 \alpha - 4 \bar{\beta} )  \Psi^0_{ 1 3 } - \mae \Psi_2 - \mae \bar{\Psi}_2   +  ( \theta - \bar{\theta} + 2 \upsilon + \bar{\upsilon} )  \bar{\Psi}^2_{ 1 3 } + \phi \Psi^2_{ 1 3 } +  ( 2 \rho - \bar{\rho} - \yogh )  \bar{\Psi}^1_3 - \sigma \Psi^1_3 + \yogh \bar{\Psi}_3 \\ +  (  - 2 \bar{\xi} + \zeta )  \Psi_{ 0 3 }  +  ( \xi - 2 \bar{\zeta} )  \Psi^2_1 +  ( \bar{\eta} - 2 \bar{\psi} )  \bar{\Psi}_{ 1 4 } +  (  - 2 \eta + \psi )  \bar{\Psi}^2_3  +  ( \xi + \bar{\zeta} )  \bar{\Psi}^2_1 +  ( \eta + \psi )  \Psi^2_3
\end{multline}
\begin{multline}\label{eqB83}
 - \tilde{D} \Psi^0_2 - 2 D \Psi^0_3 + \mathring{\delta} \bar{\Psi}^2_3 + \mathring{\delta} \Psi^2_3 =   A _{ 0  \tilde{1}  0 } - A _{ \tilde{1}  1  \tilde{1} } + 2 \wynn \Psi^0_2 +  ( \bar{\omega} - \zeta )  \Psi^1_{ 1 3 } +  ( \omega - \bar{\zeta} )  \bar{\Psi}^1_{ 1 3 } - \nu \Psi^1_1 - \bar{\nu} \bar{\Psi}^1_1 + \wynn \Psi_2 + \wynn \bar{\Psi}_2 \\ +  ( \bar{\omega} + 2 \zeta )  \bar{\Psi}^2_{ 1 3 } +  ( \omega + 2 \bar{\zeta} )  \Psi^2_{ 1 3 } +  (  - 2 \bar{\chi} - \bar{\eta} )  \bar{\Psi}^0_4   +  (  - 2 \chi - \eta )  \Psi^0_4  +  ( 2 \pi - \bar{\tau} )  \bar{\Psi}^1_3 +  ( 2 \bar{\pi} - \tau )  \Psi^1_3 + \bar{\mae} \bar{\Psi}_3 + \mae \Psi_3 \\ +  ( 4 \epsilon + 4 \bar{\epsilon} + 2 \yogh )  \Psi^0_3 + \msampi \Psi^2_1 +  ( 2 \mqoppa + \mstigma - \theta - \bar{\theta} )  \bar{\Psi}^2_3 + \msampi \bar{\Psi}^2_1   +  ( 2 \mqoppa + \mstigma - \theta - \bar{\theta} )  \Psi^2_3
\end{multline}
\begin{multline}\label{eqB89}
 - \bar{\delta} \bar{\Psi}^1_3 + \delta \Psi^1_3 + \mathring{\delta} \bar{\Psi}^2_3 - \mathring{\delta} \Psi^2_3 =  - A _{ \tilde{1}  2  \bar{2} } + (  - 2 \mu + 2 \bar{\mu} )  \Psi^0_2 +  ( \bar{\xi} + \zeta )  \Psi^1_{ 1 3 } +  (  - \xi - \bar{\zeta} )  \bar{\Psi}^1_{ 1 3 } + 2 \lambda \Psi^0_{ 1 3 } + \wynn \Psi_2 - \wynn \bar{\Psi}_2 - 2 \bar{\lambda} \bar{\Psi}^0_{ 1 3 } \\ +  (  - 2 \bar{\xi} + 2 \zeta )  \bar{\Psi}^2_{ 1 3 } +  ( 2 \xi - 2 \bar{\zeta} )  \Psi^2_{ 1 3 }   +  ( \bar{\eta} - \bar{\psi} )  \bar{\Psi}^0_4 +  (  - \eta + \psi )  \Psi^0_4 +  ( 2 \bar{\mae} + 2 \bar{\beta} )  \bar{\Psi}^1_3 +  (  - 2 \mae - 2 \beta )  \Psi^1_3 - \bar{\mae} \bar{\Psi}_3 + \mae \Psi_3 \\ +  (  - 2 \rho + 2 \bar{\rho} )  \Psi^0_3 + \bar{\phi} \bar{\Psi}_{ 1 4 } +  (  - \theta - \bar{\theta} + 2 \upsilon + \bar{\upsilon} )  \bar{\Psi}^2_3   +  ( \theta + \bar{\theta} - \upsilon - 2 \bar{\upsilon} )  \Psi^2_3 - \phi \Psi_{ 1 4 }
\end{multline}
\begin{multline}\label{eqB91}
\mathring{\delta} \Psi_2 - \bar{\delta} \bar{\Psi}^2_{ 1 3 } + \tilde{D} \bar{\Psi}^2_1 =  A _{ 1  \tilde{1}  0 } + A _{ 2  \bar{2}  0 } - \nu \Psi^0_0 +  ( \mstigma - \upsilon )  \Psi^0_2 +  ( \mae + 2 \tau )  \bar{\Psi}^1_{ 1 3 } +  (  - \bar{\omega} - \zeta )  \Psi^1_1 + \omega \bar{\Psi}^1_1 +  ( \bar{\omega} - \bar{\xi} + 2 \zeta )  \Psi_1 + \msampi \Psi^0_1 \\ + \bar{\phi} \Psi^0_{ 1 3 }   +  ( \mstigma + \upsilon )  \Psi_2 +  ( \bar{\mae} - \alpha + \bar{\beta} - \bar{\tau} )  \bar{\Psi}^2_{ 1 3 } + \bar{\psi} \bar{\Psi}^1_3 - \eta \Psi^1_3 + 2 \eta \Psi_3 - \lambda \Psi_{ 0 3 } +  ( 2 \rho - \yogh )  \bar{\Psi}^2_3 +  ( \gamma + \bar{\gamma} - \bar{\mu} - \wynn )  \bar{\Psi}^2_1
\end{multline}
\begin{multline}\label{eqB98}
 - \tilde{D} \Psi^0_2 - \bar{\delta} \bar{\Psi}^1_3 + \mathring{\delta} \bar{\Psi}^2_3 =   - A _{ 0  \tilde{1}  0 } - A _{ 2  \tilde{1}  \bar{2} } +  ( 2 \bar{\mu} + \wynn )  \Psi^0_2 +  ( \bar{\omega} + \bar{\xi} )  \Psi^1_{ 1 3 } +  ( \omega - \bar{\zeta} )  \bar{\Psi}^1_{ 1 3 } - \nu \Psi^1_1 - \bar{\nu} \bar{\Psi}^1_1 + 2 \lambda \Psi^0_{ 1 3 } + \wynn \Psi_2 \\ +  ( \bar{\omega} - 2 \bar{\xi} + 2 \zeta )  \bar{\Psi}^2_{ 1 3 } + \omega \Psi^2_{ 1 3 }  - \bar{\psi} \bar{\Psi}^0_4 - \eta \Psi^0_4+  ( \bar{\mae} + 2 \bar{\beta} - \bar{\tau} )  \bar{\Psi}^1_3 +  (  - \mae - \tau )  \Psi^1_3 + \mae \Psi_3 +  (  - 2 \rho + \yogh )  \Psi^0_3 + \msampi \Psi^2_1 + \bar{\phi} \bar{\Psi}_{ 1 4 } \\ +  ( \mstigma - \theta - \bar{\theta} + 2 \upsilon )  \bar{\Psi}^2_3 + \msampi \bar{\Psi}^2_1   +  ( \mstigma - \upsilon )  \Psi^2_3
\end{multline}
\begin{multline}\label{eqB99}
\mathring{\delta} \Psi^0_4 + \tilde{D} \Psi^1_3 + 2 \bar{\delta} \Psi^0_3 =   - A _{ \tilde{1}  \tilde{1}  \bar{2} }  - 2 \nu \Psi^0_2 - \msampi \bar{\Psi}^1_{ 1 3 } - 2 \bar{\nu} \bar{\Psi}^0_{ 1 3 } + 2 \msampi \Psi^2_{ 1 3 } + 2 \bar{\phi} \bar{\Psi}^0_4 +  ( \mstigma - 3 \theta - \bar{\theta} + 2 \upsilon )  \Psi^0_4 - 2 \lambda \bar{\Psi}^1_3 \\ +  (  - 2 \gamma - 2 \bar{\mu} - \wynn )  \Psi^1_3 - \wynn \Psi_3   +  (  - 3 \bar{\mae} - 4 \alpha - 4 \bar{\beta} + 2 \bar{\tau} )  \Psi^0_3 - \mae \Psi_4 +  ( \bar{\omega} - 2 \bar{\xi} + 3 \zeta )  \bar{\Psi}^2_3 +  (  - 2 \bar{\omega} - 2 \bar{\xi} )  \Psi^2_3 +  (  - \omega + \bar{\zeta} )  \Psi_{ 1 4 }
\end{multline}
}

\bibliography{biblio}
\bibliographystyle{alpha}

\end{document}